\documentclass[11pt]{article}

\usepackage{amsmath,amssymb,amsthm}
\usepackage{xparse}
\usepackage[dvipsnames,table]{xcolor}
\usepackage{graphicx}
\usepackage{fullpage}
\usepackage{tikz}
\usepackage{bbold}
\usepackage{caption,subcaption}
\usepackage{enumerate}
\usepackage{multirow}
\usepackage{enumitem}
\usepackage{afterpage}
\usepackage{blkarray}
\usepackage{caption}

\usepackage{arydshln}
\usepackage{adjustbox}
\usepackage{pdflscape}
\usepackage{cprotect}
\usepackage{lipsum}
\usepackage[overload]{empheq}

\usepackage{microtype}
\usepackage{wrapfig}

\usetikzlibrary{patterns}
\usetikzlibrary{decorations.pathreplacing}
\usetikzlibrary{shapes}

\usetikzlibrary{decorations.markings}
\tikzset{->-/.style={decoration={
  markings,
  mark=at position .5 with {\arrow{>}}},postaction={decorate}}}

\usepackage[ruled,vlined,linesnumbered]{algorithm2e} 

\usepackage[hang,flushmargin]{footmisc}

\makeatletter
\newcommand{\algorithmfootnote}[2][\footnotesize]{%
  \let\old@algocf@finish\@algocf@finish
  \def\@algocf@finish{\old@algocf@finish
    \leavevmode\rlap{\begin{minipage}{\linewidth}
    #1#2
    \end{minipage}}%
  }%
}
\makeatother

\usepackage{euler} 
\usepackage{palatino} 
\usepackage{fourier-orns} 
\usepackage{pifont} 
\newcommand{\xmark}{\ding{55}}%

\usepackage{hyperref}
\hypersetup{
colorlinks=true,
breaklinks=true,
urlcolor= OrangeRed,
linkcolor= OrangeRed,
citecolor=OrangeRed,
}

\def\xscale{1}
\def\yscale{1}
\def\nodescale{1}

 \usetikzlibrary{decorations.markings}
\tikzset{->-/.style={decoration={
  markings,
  mark=at position .5 with {\arrow{>}}},postaction={decorate}}}

\usepackage[colored]{shadethm}
\definecolor{shadethmcolor}{rgb}{1,0.95,0.96} 
\definecolor{shaderulecolor}{rgb}{1,1,1} 

\newshadetheorem{theorem}{Theorem}[section]
\newshadetheorem{env2}[theorem]{Lemma}
\newshadetheorem{env3}[theorem]{Proposition}
\newshadetheorem{env4}[theorem]{Corollary}
\newshadetheorem{env5}[theorem]{Remark}
\newshadetheorem{env6}[theorem]{Definition}
\newshadetheorem{env7}[theorem]{Problem}

\NewDocumentEnvironment{lemma}{o}
  {\IfNoValueTF{#1}
      {\begin{env2}}
      {\begin{env2}[#1]}
  }
  {\end{env2}}
\NewDocumentEnvironment{proposition}{o}
  {\IfNoValueTF{#1}
      {\begin{env3}}
      {\begin{env3}[#1]}
  }
  {\end{env3}}
  \NewDocumentEnvironment{corollary}{o}
  {\IfNoValueTF{#1}
      {\begin{env4}}
      {\begin{env4}[#1]}
  }
  {\end{env4}}
   \NewDocumentEnvironment{remark}{o}
  {\IfNoValueTF{#1}
      {\begin{env5}}
      {\begin{env5}[#1]}
  }
  {\end{env5}}
     \NewDocumentEnvironment{definition}{o}
  {\IfNoValueTF{#1}
      {\begin{env6}}
      {\begin{env6}[#1]}
  }
  {\end{env6}}

       \NewDocumentEnvironment{problem}{o}
  {\IfNoValueTF{#1}
      {\begin{env7}}
      {\begin{env7}[#1]}
  }
  {\end{env7}}
 \renewenvironment{proof}{{\noindent\bfseries Proof.}}{\hfill{\color{OrangeRed}{\leafNE}}}

\DeclareMathOperator*{\degr}{deg}
\DeclareMathOperator*{\height}{\mathcal{H}}

\DeclareMathOperator*{\child}{\mathcal{C}}
\DeclareMathOperator*{\leaves}{\mathcal{L}}
\DeclareMathOperator*{\des}{\mathcal{D}}
\DeclareMathOperator*{\subtrees}{\mathcal{S}}
\DeclareMathOperator*{\red}{\mathfrak{R}}

\newcommand{\alphabet}[2][]{\mathcal{A}_{#2}^{#1}}
\newcommand{\lang}[2][]{\Lambda_{#1}^{#2}}
\renewcommand{\Im}{Im\,}

\newcommand{\origin}{o} 
\DeclareMathOperator*{\stufforigin}{\mathbb{o}}

\DeclareMathOperator*{\presence}{\pi}

\DeclareMathOperator*{\topord}{\psi}
\DeclareMathOperator*{\childc}{\mathcal{C}_{\topord}}

\newcommand{\lex}[1]{#1_\text{lex.}}
\newcommand{\scut}{SC}

\newcommand\restr[2]{{
  \left.\kern-\nulldelimiterspace 
  #1 
  \vphantom{\big|} 
  \right|_{#2} 
  }}

   \DeclareMathOperator*{\tree}{\mathcal{T}}
  
 \DeclareMathOperator*{\forest}{\mathcal{F}}

\usetikzlibrary{calc}

\tikzset{
    ncbar angle/.initial=90,
    ncbar/.style={
        to path=(\tikztostart)
        -- ($(\tikztostart)!#1!\pgfkeysvalueof{/tikz/ncbar angle}:(\tikztotarget)$)
        -- ($(\tikztotarget)!($(\tikztostart)!#1!\pgfkeysvalueof{/tikz/ncbar angle}:(\tikztotarget)$)!\pgfkeysvalueof{/tikz/ncbar angle}:(\tikztostart)$)
        -- (\tikztotarget)
    },
    ncbar/.default=0.5cm,
}

\tikzset{round left paren/.style={ncbar=0.5cm,out=120,in=-120}}
\tikzset{round right paren/.style={ncbar=0.5cm,out=60,in=-60}}

\newcommand{\elongation}[1][1]{%
\raisebox{-1pt}{\begin{tikzpicture}[scale=#1]%

\draw[fill] (0,0) circle [radius=0.2ex] {};

\draw[->,thick] (0,0) -- (0,1.2ex);%

    \draw [thick] (-0.7ex,-0.2ex) to [round left paren ] (-0.7ex,1.2ex);
    \draw [thick] (0.7ex,-0.2ex) to [round right paren] (0.7ex,1.2ex);
\end{tikzpicture}}
}

\newcommand{\node}[1][white]{
\begin{tikzpicture}
\tikzstyle{noeud}=[draw,circle,fill=#1,scale=\nodescale*0.7]
\node[noeud] (0) at ({0},{0}) {};
\end{tikzpicture}
}

\newcommand{\widening}[1][1]{%
\raisebox{-1pt}{\begin{tikzpicture}[scale=#1]%

\draw[fill] (0,0) circle [radius=0.2ex] {};
\draw[fill] (0,1ex) circle [radius=0.2ex] {};

\draw[thick] (0,0)--(0,1ex);
\draw[->,thick] (0,0) -- (1.2ex,1.2ex);%

    \draw [thick] (-0.7ex,-0.2ex) to [round left paren ] (-0.7ex,1.2ex);
    \draw [thick] (1.7ex,-0.2ex) to [round right paren] (1.7ex,1.2ex);
\end{tikzpicture}}
}

\newcommand{\branching}[1][1]{%
\raisebox{-1pt}{\begin{tikzpicture}[scale=#1]%

\draw[fill] (0,1ex) circle [radius=0.2ex] {};

\draw[->,thick] (0,1ex) -- (0,-0.2ex);%
\draw[->,thick] (0,1ex) to [bend right=45](-1ex,-0.2ex);
\draw[->,thick] (0,1ex) to [bend left=45] (1ex,-0.2ex);

    \draw [thick] (-1.7ex,-0.2ex) to [round left paren ] (-1.7ex,1.2ex);
    \draw [thick] (1.7ex,-0.2ex) to [round right paren] (1.7ex,1.2ex);
\end{tikzpicture}}
}

\newcommand{\elong}[1][1]{%
\raisebox{-3pt}{\begin{tikzpicture}[scale=#1]%

\draw[fill] (0,0) circle [radius=0.2ex] {};

\draw[->,thick] (0,0) -- (0,1.2ex);%

\end{tikzpicture}}
}

\newcommand{\widen}[1][1]{%
\raisebox{-3pt}{\begin{tikzpicture}[scale=#1]%

\draw[fill] (0,0) circle [radius=0.2ex] {};
\draw[fill] (0,1ex) circle [radius=0.2ex] {};

\draw[-,thick] (0,0)--(0,1ex);
\draw[->,thick] (0,0) -- (1.2ex,1.2ex);%

\end{tikzpicture}}
}

\newcommand{\branch}[1][1]{%
\raisebox{-3pt}{\begin{tikzpicture}[scale=#1]%

\draw[fill] (0,1ex) circle [radius=0.2ex] {};

\draw[->,thick] (0,1ex) -- (0,-0.2ex);%
\draw[->,thick] (0,1ex) to [bend right=45](-1ex,-0.2ex);
\draw[->,thick] (0,1ex) to [bend left=45] (1ex,-0.2ex);

\end{tikzpicture}}
}

\title{\Large\textsc{Enumeration Of Irredundant Forests}}

\author{Florian Ingels \\ \href{mailto:florian.ingels@inria.fr}{florian.ingels@inria.fr}
\and
Romain Aza\"is \\ \href{mailto:romain.azais@inria.fr}{romain.azais@inria.fr}}

\date{\small Laboratoire Reproduction et D\'eveloppement des Plantes, Univ Lyon, ENS de Lyon, UCB Lyon 1, CNRS, INRAE, Inria, F-69342, Lyon, France}

\begin{document}

\maketitle

\vspace{0.25cm}


{\small
\begin{center}\textbf{Abstract}\end{center}

\vspace{-0.1cm}

\noindent
Reverse search is a convenient method for enumerating structured objects, that can be used both to address theoretical issues and to solve data mining problems. This method has already been successfully developed to handle unordered trees. If the literature proposes solutions to enumerate singletons of trees, we study in this article a more general problem, the enumeration of sets of trees -- forests. Specifically, we mainly study \emph{irredundant} forests, i.e., where no tree is a subtree of another. By compressing each such forest into a Directed Acyclic Graph (DAG), we develop a reverse search like method to enumerate DAGs compressing irredundant forests. Remarkably, we prove that these DAGs are in bijection with the row-Fishburn matrices, a well-studied class of combinatorial objects. In a second step, we derive our irredundant forest enumeration to provide algorithms for tackling related problems: (i) enumeration of forests in their classical sense (where redundancy is allowed); (ii) the enumeration of ``subforests'' of a forest, and (iii) the frequent ``subforest'' mining problem. All the methods presented in this article enumerate each item uniquely, up to isomorphism.

\smallskip

\noindent
\textbf{keywords:} Directed Acyclic Graph, Reverse Search, Unordered Trees, Enumeration, Forest
}


\section{Introduction}

\subsection{Context of the work}\label{ss:def:reverse}
Enumeration of trees is a long-term problem, where Cayley was the first to propose a formula for counting unordered trees in the mid-19th century \cite[I.5.2]{flajolet2009analytic}.
The exhaustive enumeration of ordered and unordered trees\footnote{A rooted tree is a connected graph without cycles such that there exists a special vertex called the root that has no parent, and the other vertices have exactly one parent. Trees are called ordered or unordered whether the order among siblings is important or not. See Subsection~\ref{ss:def:treedag}.} was successfully tackled in the early 00's by Nakano and Uno in \cite{nakano2002efficient,nakano2003efficient}. In the unordered case, an extension of the algorithm has been proposed to solve the problem of frequent substructure mining \cite{asai2003discovering}. Moreover, in the field of machine learning, we have recently demonstrated that exhaustive enumeration of the subtrees of a tree makes it possible to design classification algorithms significantly more efficient than their counterpart without such enumeration  \cite{azais2020weight}.

\medskip
\noindent
Our ambition in this article is to take these two problems of enumeration -- trees and subtrees -- to a higher order, i.e. to enumerate sets of trees instead of singletons. Specifically, we call an irredundant forest (shortened to forest in the sequel) a set of trees that contains no repetition -- in the sense that no tree is a subtree of another (see upcoming Subsection~\ref{ss:def:treedag} for a precise definition).  This condition of non-repetition is in line with a parsimonious enumeration approach, where the objects considered are all different and enumerated up to isomorphism. Besides, this condition is not restrictive since one can always introduce repetition afterwards. We are therefore interested in the problem of enumerating forests of unordered trees, and then, given a tree or forest, to enumerate all its ``subforests'' -- as forests of subtrees. The latter has already been discussed in the literature, but without consideration on isomorphism \cite{schwikowski2002enumerating}. We re-emphasize that we aim to enumerate these various items -- forests and subforests -- up to isomorphism.

\medskip
\noindent
Such an ambition immediately raises a number of obstacles. First of all, the trees are indeed unordered, but so are the sets of trees. For the former the literature has introduced the notion of the canonical form of a tree  \cite{nakano2003efficient,asai2003discovering}, which is a unique ordered representation of an unordered tree. The enumeration therefore focuses only on these canonical trees. Unfortunately, if it is possible to order a set of vertices, there is no total order on the set of trees, to the best of our knowledge. In addition, the condition of non-repetition filters the set of forests in a non-trivial way, making, a priori, the enumeration problem trickier.

\medskip
\noindent
Enumeration problems are recurrent in many fields, notably combinatorial optimization and data mining. They involve the exhaustive listing of a subset of the elements of a search set (possibly all of them), e.g. graphs, trees or vertices of a simplex. Given the possibly high combinatorial nature of these elements, it is essential to adopt clever exploration strategies as opposed to brute-force enumeration, typically to avoid areas of the search set not belonging to the objective subset.

\medskip
\noindent
One proven way of proceeding is to provide the search set with an enumeration tree structure; starting from the root, the branches of the tree are explored recursively, eliminating those that do not address the problem. Based on this principle, we can notably mention the well-known ``branch and bound'' method in combinatorial optimization \cite{land2010automatic} and the gSpan algorithm for frequent subgraph mining in data mining \cite{yan2002gspan}. Another of these methods is the so-called reverse search technique, which requires that the search set has a partial order structure, and which has solved a large number of enumeration problems since its introduction \cite{avis1996reverse} until very recently \cite{yamazaki2020enumeration}.  Actually, the algorithms previously introduced in the literature to enumerate trees are based on this technique \cite{nakano2002efficient,nakano2003efficient,asai2003discovering}.

\medskip
\noindent
In the present paper, we restrict ourselves to reverse search methods, for which the following formalism is adapted from the one that can be found in \cite[p. 45-51]{nowozin2009learning}, and slightly differ from the original definition by Avis and Fukuda \cite{avis1996reverse}. We refer the reader to these two references for further details.

\medskip
\noindent
Let $(\mathbb{S},\subseteq)$ be a partially ordered set, and $g:\mathbb{S} \to \lbrace \top,\bot \rbrace $ be a \emph{property}, satisfying anti-monotonicity
$$\forall s,t\in \mathbb{S} : ( s\subseteq t  )\land g(t)\implies g(s).$$
\noindent
The \emph{enumeration problem} for the property $g$ is the problem of listing all elements of $E_{\mathbb{S}}(g) = \lbrace s\in \mathbb{S} : g(s)=\top\rbrace$. An enumeration algorithm is an algorithm that returns $E_{\mathbb{S}}(g)$.

\medskip
\noindent
The reverse search technique relies on inverting a \emph{reduction rule} $f:\mathbb{S}\setminus \emptyset \to \mathbb{S}$, where $f$ satisfies the two properties of (i) \emph{covering}: $\forall s\in \mathbb{S}\setminus \emptyset, f(s) \subset s$ and (ii) \emph{finiteness}: $\forall s \in \mathbb{S}\setminus \emptyset, \exists k\in\mathbb{N}^\ast, f^k(s)=\emptyset$. Then, the \emph{expansion rule} is defined as $f^{-1}(t) = \lbrace s\in \mathbb{S} : f(s)=t\rbrace$. This defines an enumeration tree rooted in $\emptyset$, and repeated call to $f^{-1}$ can therefore enumerates all the elements of $\mathbb{S}$.

\medskip
\noindent
\begin{minipage}{0.45\textwidth}
The reverse search algorithm is shown in Algorithm~\ref{algo:reverse:def}. $E_{\mathbb{S}}(g)$ can be obtained from the call of \textsc{ReverseSearch}($(\mathbb{S},\subseteq),f^{-1},g,\emptyset$). As $g$ is anti-monotone, if $g(s)=\bot$, then all elements $s\subseteq t$ also have $g(t)=\bot$, and thefore pruning the enumeration tree in $s$ does not miss any element of $E_{\mathbb{S}}(g)$.
\end{minipage}
\hfill
\begin{minipage}{0.5\textwidth}
\begin{algorithm}[H]
\caption{\textsc{ReverseSearch}}\label{algo:reverse:def}
\KwIn{$(\mathbb{S},\subseteq)$, $f^{-1}$, $g$, $s_0\in \mathbb{S}$ -- s.t. $g(s_0)=\top$}
\textbf{output} $s_0$\\
\For{$t\in \lbrace s\in f^{-1}(s_0) | g(s)=\top\rbrace$}{
\textsc{ReverseSearch}($(\mathbb{S},\subseteq),f^{-1},g,t$)
}
\end{algorithm}
\end{minipage}

\medskip
\noindent
When successfully designed, a reverse search technique should yield polynomial output delay \cite{johnson1988generating,avis1996reverse}, i.e., the time between the output of one element and the next is bounded by a polynomial function in the size of the input.

\begin{remark}
It would have been possible to define directly the set $\mathbb{S}$ as the set of elements verifying the property $g$. Separating the two induces that the reduction rule $f$ formally depends only on $\mathbb{S}$, and not on $g$. This allows, once $f$ is constructed once and for all, to filter $\mathbb{S}$ according to various properties $g$ without additional work. In particular, this is useful in the case where $g$ depends on a tunable parameter  -- as in the \emph{frequent pattern mining problem} introduced in Section~\ref{sec:datamining}.
\end{remark}

\subsection{Precise formulation of the problem}\label{ss:def:treedag}

A rooted tree $T$ is a connected graph with no cycle such that there exists a unique vertex called the root, which has no parent, and any vertex different from the root has exactly one parent. Rooted trees are said unordered if the order between the sibling vertices of any vertex is not significant. As such, the set of children of a vertex $v$ is considered as a multiset and denoted by $\child(v)$. The leaves $\leaves(T)$ are all the vertices without children. The height of a vertex $v$ of a tree $T$ can be recursively defined as

\begin{equation}\label{eq:height}
\height(v)=\begin{cases} 0 & \text{if } v\in \leaves(T), \\ 1+\max_{u \in \child(v)}\height(u) & \text{otherwise.} \end{cases}
\end{equation}
\noindent
The height $\height(T)$ of the tree $T$ is defined as the height of its root. The outdegree of a vertex $v\in T$ is defined as $\degr(v)=\#\child(v)$\footnote{The notation $\#$ is used in this paper to denote both (i) the cardinality $\#S$ of any set $S$, and (ii) the number of vertices $\#G$ of any graph $G$.}; the outdegree of $T$ is then defined as $\degr(T)=\max_{v\in T}\degr(v)$. The depth of a vertex $v$ is the number of edges on the path from $v$ to the root of the tree.

\medskip
\noindent
Two trees $T_1$ and $T_2$ are isomorphic if there exists a one-to-one correspondance $\phi$ between the vertices of the trees such that (i) $u\in \child(v)$ in $T_1 \iff \phi(u) \in \child(\phi(v))$ in $T_2$ and (ii) the roots are mapped together. For any vertex $v$ of $T$, the subtree $T[v]$ rooted in $v$ is the tree composed of $v$ and all its descendants -- denoted by $\des(v)$. $\subtrees(T)$ denotes the set of all distinct subtrees of $T$, which is the quotient set of $\lbrace T[v] : v\in T\rbrace$ by the tree isomorphism relation. In this article, we consider only unordered rooted trees that will simply be called trees in the sequel. We denote by $\tree$ the set of all trees.

\medskip
\noindent
As mentioned before, we are interested in this paper in the enumeration of \emph{forests}. The literature acknowledges two definitions for a forest  \cite[p. 172]{bender2010lists}: (i) an undirected graph in which any two vertices are connected by at most one path or (ii) a disjoint union of trees. We adopt a variation of the latter one, that forbids repetitions inside the forest.

\begin{definition}
A set $\lbrace T_1,\dots, T_n \rbrace$ of trees is an irredundant forest if and only if
\begin{equation}\label{eq:forest}
\forall i\neq j, T_i \notin \subtrees(T_j).
\end{equation}
\end{definition}

\noindent
We denote by $\forest$ the set of all irredundant forests -- shortened to forests in the sequel of the paper. Our goal is to provide a reverse search method that outputs $\forest$. As already stated, this goal raises two major difficulties: firstly, the twofold unordered nature of forests (the set of trees and the trees themselves), and secondly, the non-trivial condition of non-repetition. While the latter problem is intrinsic, the main idea of this paper to address the former is to resort to the reduction of a forest into a Directed Acyclic Graph (DAG).

\begin{wrapfigure}[16]{R}{0.45\textwidth}
\vspace{-\baselineskip}
\centering
\def\xscale{0.6}
\def\yscale{0.6}
\def\nodescale{0.6}
\begin{tikzpicture}[xscale=\xscale,yscale=\yscale]
\tikzstyle{noeud}=[draw,circle,fill=blue,scale=\nodescale*1]
\node[noeud] (0) at (0,0) {};
\tikzstyle{noeud}=[draw,circle,fill=magenta,scale=\nodescale*1]
\node[noeud] (1) at (0,1) {};
\tikzstyle{noeud}=[draw,circle,fill=orange,scale=\nodescale*1]
\node[noeud] (4) at (0,2) {};

 \node at (0,3) {$T_1$};

\tikzstyle{fleche}=[-,>=latex]
\draw[fleche] (1)--(0) {};
\draw[fleche] (4)--(1) {};

\def\xshift{2}

\tikzstyle{noeud}=[draw,circle,fill=blue,scale=\nodescale*1]
\node[noeud] (0) at (0+\xshift,0) {};
\node[noeud] (0b) at (-1+\xshift,0) {};
\node[noeud] (0c) at (1+\xshift,0) {};
\node[noeud] (0d) at (2+\xshift,0) {};
\tikzstyle{noeud}=[draw,circle,fill=magenta,scale=\nodescale*1]
\node[noeud] (1) at (-1+\xshift,1) {};
\node[noeud] (1b) at (0+\xshift,1) {};
\tikzstyle{noeud}=[draw,circle,fill=cyan,scale=\nodescale*1]
\node[noeud] (2) at (1.5+\xshift,1) {};
\tikzstyle{noeud}=[draw,circle,fill=yellow,scale=\nodescale*1]
\node[noeud] (5) at (0+\xshift,2) {};
\tikzstyle{fleche}=[-,>=latex]
\draw[fleche] (1b)--(0) {};
\draw[fleche] (1)--(0b) {};
\draw[fleche] (2)--(0c) {};
\draw[fleche] (2)--(0d) {};
\draw[fleche] (5)--(2) {};
\draw[fleche] (5)--(1) {};
\draw[fleche] (5)--(1b) {};

\node at (0+\xshift,3) {$T_2$};

\def\xshift{6}

\tikzstyle{noeud}=[draw,circle,fill=blue,scale=\nodescale*1]
\node[noeud] (0) at (0+\xshift,0) {};
\node[noeud] (0b) at (-1+\xshift,0) {};
\node[noeud] (0c) at (1+\xshift,0) {};
\tikzstyle{noeud}=[draw,circle,fill=brown,scale=\nodescale*1]
\node[noeud] (3) at (0+\xshift,1) {};

\tikzstyle{fleche}=[-,>=latex]

\draw[fleche] (3)--(0) {};
\draw[fleche] (3)--(0b) {};
\draw[fleche] (3)--(0c) {};

\node at (0+\xshift,3) {$T_3$};

\draw [decorate,decoration={brace,amplitude=5pt,raise=2ex}] (-0.5,3) -- (7.5,3) node[midway,yshift=2em]{$F=\lbrace T_1,T_2,T_3\rbrace$};

\def\xshift{10}

\def\xscale{0.7}
\def\yscale{0.7}
\def\nodescale{0.7}

\tikzstyle{noeud}=[draw,circle,fill=blue,scale=\nodescale*1]
\node[noeud] (0) at (0+\xshift,0) {};
\tikzstyle{noeud}=[draw,circle,fill=magenta,scale=\nodescale*1]
\node[noeud] (1) at (-1+\xshift,1) {};
\tikzstyle{noeud}=[draw,circle,fill=cyan,scale=\nodescale*1]
\node[noeud] (2) at (0+\xshift,1) {};
\tikzstyle{noeud}=[draw,circle,fill=brown,scale=\nodescale*1]
\node[noeud] (3) at (1+\xshift,1) {};
\tikzstyle{noeud}=[draw,circle,fill=orange,scale=\nodescale*1]
\node[noeud] (4) at (-1+\xshift,2) {};
\tikzstyle{noeud}=[draw,circle,fill=yellow,scale=\nodescale*1]
\node[noeud] (5) at (0+\xshift,2) {};
\tikzstyle{fleche}=[->-,>=latex,black]
\draw[fleche] (1)--(0) {};
\draw[fleche] (2)--(0) {}node [right,midway,scale=\nodescale] {$2$};
\draw[fleche] (3)--(0) {}node [right,midway,scale=\nodescale] {$3$};
\draw[fleche] (4)--(1) {};
\draw[fleche] (5)--(2) {};
\draw[fleche] (5)--(1) {}node [left,midway,scale=\nodescale] {$2$};

\draw[->,red,ultra thick] (0+\xshift,3)--(5) {};
\draw[->,red,ultra thick] (-1+\xshift,3)--(4) {};
\draw[->,red,ultra thick] (1+\xshift,2)--(3) {};

\node[yshift=2em] at (0+\xshift,3) {$\red(F)$};

\end{tikzpicture}
\captionof{figure}{A forest $F$ (left) and its DAG reduction (right). Roots of isomorphic subtrees are identically colored, as well as the corresponding vertex of the DAG. The sources of the DAG (indicated with red arrows) correspond exactly to the roots of the trees in $F$. For the sake of clarity, arcs of multiplicity greater than one are drawn only once and their multiplicity is written next to the arc.}
\label{fig:dag:reduction}
\end{wrapfigure}

\medskip
\noindent
DAG reduction is a method meant to eliminate internal repetitions in the structure of trees and forests of trees. Beginning with \cite{Sutherland:1963:SMG:1461551.1461591}, DAG representations of trees are also much used in computer graphics where the process of condensing a tree into a graph is called object instancing \cite{Hart:1991:EAR:127719.122728}. A precise definition of DAG reduction of trees, together with algorithms to compute it, are provided in \cite{GODI}, whereas one technique to extend those algorithms to forests is presented in \cite[Section 3.2]{azais2020weight}.  DAG reduction can be interpreted as the construction of the quotient graph of a forest by the tree isomorphism relation. However, in this paper, we provide the general idea of DAG reduction as a vertex coloring procedure.

\medskip
\noindent
Consider a forest $F=\lbrace T_1,\dots, T_n\rbrace$ to reduce. Each vertex of each tree is given a color such that if two distinct vertices $u,v$ belonging respectively to $T_i, T_j$ (not necessarily distinct) have the same color, then $T_i[u]$ and $T_j[v]$ are isomorphic. Reciprocally, if two subtrees are isomorphic, their roots have to be identically colored. Let us denote $c(\cdot)$ the function that associates a color to any vertex. Then, we build a directed graph $D=(V,A)$ with as many vertices as colors used, i.e. $\#V = \# \Im(c)$. For any two vertices $u,v$ in the forest, if $u\in\child(v)$, then we create an arc $c(v)\to c(u)$ in $D$. Note that this definition implies that multiples arcs are possible in $D$, as if there exist $u,u' \in \child(v)$, for $v\in T$, such that $T[u]$ and $T[u']$ are isomorphic, then the arcs $c(v)\to c(u)$ and $c(v)\to c(u')$ are identical. The graph $D$ is a DAG \cite[Proposition 1]{GODI}, i.e. a connected directed (multi)graph without cycles. We refer to Figure~\ref{fig:dag:reduction} for an example of DAG reduction.

\medskip
\noindent
In this paper, $\red(F)$ denotes the DAG reduction of $F$. It is crucial to notice that the function $\red$ is a one-to-one correspondence \cite[Proposition~4]{GODI}, which means that DAG reduction is a lossless compression algorithm. Since $F$ fulfills condition (\ref{eq:forest}), no tree of $F$ is a subtree of another. If this were the case, say $T_i \in \subtrees(T_j)$, then $\red(T_i)$ would be a subDAG of $\red(T_j)$, and therefore the numbers of roots in $\red(F)$ would be strictly less than $\#F$. Since such a situation can not occur, there are exactly as many roots in $\red(F)$ as there are elements in $F$: no information is lost. In other words, $F$ can be reconstructed from $\red(F)$ and $\red^{-1}$ stands for the inverse function.

\medskip
\noindent
The DAG structure inherits of some properties of trees. For a vertex $v$ in a DAG $D$, we will denote by $\child(v)$ the set of children of $v$. $\height(v)$ and $\deg(v)$ are inherited as well. Similarly to trees, we denote by $D[v]$ the subDAG rooted in $v$ composed of $v$ and all its descendants $\des(v)$. Note that since $D[v]$ has a unique root $v$, it compresses a forest made of a single tree. For the simplicity of notation, we use $\red^{-1}(D[v])$ to designate the tree compressed by $D[v]$ -- instead of the singleton.

\medskip
\noindent
In the sequel, DAGs compressing forests are called \emph{FDAGs}, to distinguish them from general directed acyclic graphs.

\medskip
\noindent
Since DAG compression is lossless, and since a forest can be reconstructed from its DAG reduction, it should be clear that enumerating all forests is equivalent to enumerating all FDAGs. Yet, the latter approach has the merit of transforming set of trees into unique objects, which makes it possible, if able to design a canonical representation -- like the trees in \cite{nakano2003efficient,asai2003discovering}, to get rid of the twofold unordered nature of forests, as claimed earlier. Indeed, any ordering of the vertices of the DAG induces an order on the roots of the DAG, and therefore on the elements of the forest, as well on the vertices of the trees themselves.

\subsection{Aim of the paper}\label{ss:aim}

To the best of our knowledge, the enumeration of DAGs has never been considered in the literature. The aim of this article is twofold, i.e (i) to open the way by presenting a reverse search algorithm enumerating FDAGs, in Section~\ref{sec:enumeration}, and (ii) to derive from it an algorithm for enumerating substructures in Section~\ref{sec:mining}. The frequent pattern mining problem is a classical data mining problem  -- see \cite{han2007frequent} for a survey on that question -- and we provide in Section~\ref{sec:datamining} a slight variation of the algorithm of Section~\ref{sec:mining} to tackle this issue. In addition, Section~\ref{sec:analysis} analyses the growth of the enumeration tree defined in Section~\ref{sec:enumeration}, while Section~\ref{sec:variant} proposes two variations of it. In more detail, our outline is as follows:

\begin{itemize}
\item The first step is to introduce a canonical form for FDAG. For trees \cite[Section 3]{nakano2003efficient}, this consisted in associating an integer (its depth) to each vertex, and maximizing the sequence by choosing an appropriate ordering over the vertices. The notion of depth does not apply to FDAGs, which forces us to find another strategy. DAGs are characterized by the existence of a topological ordering \cite{kahn1962topological}, and we introduce in Subsection~\ref{ss:canonical} a topological ordering that is unique if and only if a DAG compresses a forest. This canonical ordering is defined so that the sequence of children of the vertices is strictly increasing, where the multisets of children are ordered by the lexicographical order. In fact, these ordered multisets of children are considered as formal words, which brings us to a detour through the theory of formal languages in Subsection~\ref{ss:words} to introduce useful results for the rest of the article. Compared to trees, we have here a first gain in complexity insofar as we maximize a sequence of words instead of a sequence of integers.

\item The expansion rule used for trees \cite[Section 4]{nakano2003efficient} is to add a new vertex in the tree as a child of some other vertex, so that the depth-sequence remains maximal. Consequently, a single arc is also added. On the other hand, for a FDAG, we want to be able to add either vertices or arcs independently. In Subsection~\ref{ss:rules}, we define three expansion rules, reflecting the full spectrum of possible operations, so that the DAG obtained afterward is still a FDAG. Specifically, the branching rule allows to add an arc, where the elongation and widening rules add vertices at different height. We show in Proposition~\ref{prop:canonical} that the rules preserve the canonicalness and in Proposition~\ref{prop:bij} that they are ``bijective'': any FDAG can be reached by applying the expansion rules to a unique FDAG. In Subsection~\ref{ss:enumtree} we derive from them an enumeration tree covering the set of FDAGs.

\item Notably, a bijection between FDAGs and row-Fishburn matrices, a class of combinatorial objects much exploited in the literature  \cite[Section 2]{hwang2019asymptotics}, is shown in Theorem~\ref{th:bijdagmatrix} -- which proof lies in Appendix~\ref{sec:dagmatrix}. The asymptotic behavior of these matrices being well known \cite{jelinek2012counting,bringmann2014asymptotics}, this allows us to derive from it the behavior of the enumeration tree. In return, since our bijection is constructive, the enumeration tree can be used to enumerate row-Fishburn matrices -- and all the objects they are in bijection with -- via the reverse search method. Remarkably, this bijection operates between two objects that, at first sight, have little in common.

\item For an enumeration algorithm to have any practical interest, it is necessary that the associated enumeration tree has a ``reasonable'' growth -- with regard to the size of the explored space. This is the case for our algorithm since we prove, in Subsection~\ref{ss:branching}, that a FDAG with $n$ vertices has a number of successors in the enumeration tree in the order of $\Theta(n)$ -- and that those successors can be computed in quasi-quadratic time. We also show, in Theorem~\ref{th:polynomial_delay}, that our algorithm runs with polynomial delay \cite{johnson1988generating}.

\item Subsection~\ref{ss:redundant} introduces a way of enumerating forests in their classical definition, i.e., with redundancy, where some trees may be equal to or subtrees of others. The proposed method takes a redundancy-free forest, as enumerated by our algorithm, and adds repetition in an extra enumeration step. Finally, Subsection~\ref{ss:constraint} concludes on enumeration by proposing sets of constraints that make the enumeration tree finite. Indeed, since the rules only allow to increase the height, degree or number of vertices, it is sufficient to set maximum values for some of these parameters to achieve this goal; however the combination of parameters has to be wisely chosen, as we show it.

\item Since the structures we enumerate are forests, it is natural that the substructures we are interested in are ``subforests''. A precise definition of the latter is given in Section~\ref{sec:mining}, i.e. forests of subtrees, and are referred to as subFDAGs. An algorithm to enumerate all subFDAGs appearing in a FDAG is also provided. The frequent subFDAG mining problem is finally addressed in Section~\ref{sec:datamining}.
\end{itemize}

\noindent
Concluding remarks concerning the implementation of our results in the Python library \verb+treex+ \cite{azais2019treex} are briefly mentioned at the end of the article. In Appendix~\ref{annex:notations}, the interested reader will find an index of frequent notations used throughout the paper.

\section{Exhaustive enumeration of FDAGs}\label{sec:enumeration}

In this section, we introduce our main result, that is, a reverse search algorithm for the enumeration of FDAGs. As we will consider the multisets of children of vertices as formal words built on the alphabet formed by the set of vertices, we introduce in Subsection~\ref{ss:words} some definitions and results on formal languages that will be useful for the sequel. We characterize unambiguously in Subsection~\ref{ss:canonical} our objects of study, through the lens of topological orderings, defining a canonical topological ordering for DAGs, that is unique if and only if a DAG compresses a forest of unordered trees, i.e. it is a FDAG -- see Theorem~\ref{prop:dag:cutf}. We then define three expansion rules that are meant to extend the structures of FDAGs in Subsection~\ref{ss:rules}, and we study their properties in Subsection~\ref{ss:analysis_rule}. In Subsection~\ref{ss:enumtree}, we show with Theorem~\ref{th:reduction} that these expansion rules define an enumeration tree on the set of FDAGs.

\subsection{Preliminary: a detour through formal languages}\label{ss:words}

We present in this subsection some definitions and results on formal languages that will be useful for the sequel of Section~\ref{sec:enumeration}.

\medskip
\noindent
Let $\alphabet[]{}$ be a totally ordered finite set, called alphabet, whose elements are called letters. A word is a finite sequence of letters of $\alphabet[]{}$. The length of a word $w$ is equal to its number of letters and is denoted by $\# w$. There is a unique word with no letter called the empty word and denoted  by $\epsilon$. The set of all words is denoted  by $\alphabet[*]{}$. Words can be concatenated to create a new word whose length is the sum of the lengths of the original words; $\epsilon$ is the neutral element of this concatenation operation.

\medskip
\noindent
The lexicographical order over $\alphabet[*]{}$, denoted by $\lex{<}$ is defined as follows. Let $w_1=a_0\cdots a_p$ and $w_2 = b_0\cdots b_q$ be two words, with $a_i,b_j\in \alphabet[]{}$. If $\# w_1 = \#w_2$, then $w_1\lex{<} w_2$ if and only if $\exists k \in [\![0,p]\!], a_i=b_i\; \forall i<k$ and $a_k<b_k$. Otherwise, let $m=\min(p,q)$; $w_1\lex{<}w_2$ if and only if either (i) $a_0\cdots a_m \lex{<} b_0\cdots b_m$ or (ii) $a_0\cdots a_m \lex{=} b_0\cdots b_m$ and $m<q$ -- that is, $p<q$. Note that, by convention, $\epsilon \lex{<} w$ for any word $w$.

\medskip
\noindent
Let $w\in \alphabet[*]{}$. We define the suffix-cut operator $\scut(w)$, which removes the last letter of $w$:

\begin{equation}\label{eq:scut}
\scut(w) = \begin{cases} w' & \text{if } w =w'a \text{ with } a\in \alphabet[]{} \text{ and } w' \in \alphabet[*]{},\\
\epsilon & \text{otherwise}.
\end{cases}
\end{equation}

\medskip
\noindent
A \emph{language} is a set of words satisfying some construction rules. We introduce hereafter two languages that will be useful in the sequel of the paper.

\begin{definition}\label{def:decreasingwords}
The language of decreasing words is defined as
$$\lang{} =\big\lbrace w=a_0\cdots a_m \in \alphabet[*]{}  :  a_i \lex{\geq} a_{i+1} \; \forall i\in [\![ 0,m-1]\!] \big\rbrace.$$
\end{definition}

\begin{definition}\label{def:boundedwords}
Let $\overline{w}\in \lang{}$. The language of decreasing words bounded by $\overline{w}$ is defined as
$$\lang{\overline{w}} = \left\lbrace w \in \lang{}  : w \lex{>} \overline{w}\right\rbrace.$$
Any word $w\in \lang{\overline{w}}$ is said to be minimal if and only if $w\in \lang{\overline{w}}$ but $\scut(w)\notin \lang{\overline{w}}$.
\end{definition}
\noindent
As an example, if $\alphabet[]{}=\lbrace 0,1,2,3\rbrace$, then  $\overline{w}=211\in\lang{}$, whereas $121\notin \lang{}$. In addition, $\lang{\overline{w}}$ contains words such as $31$, $22$, $21110$, etc. $22$ is a minimal word of $\lang{\overline{w}}$ as $22 \lex{>} 211$ but $\scut(22)=2 \lex{<}211$.

\medskip
\noindent
Our focus is now on the construction of the minimal words of $\lang{\overline{w}}$. Let $\overline{w} = a_0\cdots a_p$ and $w=b_0\cdots b_q\in \lang{\overline{w}}$. Taking into account that $w\lex{>} \overline{w}$ and that they both are decreasing words, there are only two possibles cases:
\begin{enumerate}[label=(\roman*)]
\item $w$ and $\overline{w}$ share a common prefix $a_0\cdots a_m$. Then $w=a_0\cdots a_m b_{m+1} \cdots b_q$, and the word $a_0\cdots a_m b_{m+1}$ is minimal by applying successive suffix-cut operations.
\item $w$ and $\overline{w}$ do not share a common prefix. Necessarily $b_0 \lex{>} a_0$, and then the word $b_0$ is minimal by applying several suffix-cut operations.
\end{enumerate}

\noindent
From the above, we deduce a method for constructing all minimal word of $\lang{\overline{w}}$. First, we partition $\alphabet[]{}$ into disjoint -- potentially empty -- subsets:
\begin{align*}
\alphabet[0]{} & =  \lbrace a \in \alphabet[]{} : a \lex{>} a_0 \rbrace,\\
\alphabet[i]{} &= \lbrace a \in \alphabet[]{} : a_{i-1} \lex{\geq} a \lex{>} a_{i}\rbrace \quad 1\leq i \leq p, \\
\alphabet[p+1]{} &=  \lbrace a \in \alphabet[]{} :  a_p \lex{\geq} a\rbrace.
\end{align*}

It then follows that -- empty $\alphabet[i]{}$'s not being considered,
\begin{itemize}
\item $\forall b\in \alphabet[0]{}$, the word $b$ is minimal,
\item $\forall b\in \alphabet[i]{}$ with $i \in \lbrace 1,\dots, p\rbrace$, the word $a_0\cdots a_{i-1}b$ is minimal,
\item $\forall b \in \alphabet[p+1]{}$, the word $\overline{w} b$ is minimal.
\end{itemize}

\medskip
\noindent
As we partitioned $\alphabet[]{}$, we have proved the following proposition.

\begin{proposition}\label{prop:widening}
The number of minimal words of $\lang{\overline{w}}$ is exactly $\#\alphabet[]{}$.
\end{proposition}

\noindent
\begin{minipage}{0.45\textwidth}
As a follow-up of the example some lines ago, with $\alphabet[]{}=\lbrace 0,1,2,3\rbrace$ and $\overline{w}=211$, we apply the proposed method to find the minimal elements of $\lang{\overline{w}}$. We partition $\alphabet[]{}$ into: $ \alphabet[0]{} = \lbrace 3 \rbrace$, $ \alphabet[1]{} = \lbrace 2 \rbrace$, $ \alphabet[2]{} =\emptyset$, $ \alphabet[3]{} = \lbrace 0,1 \rbrace$. The four minimal words are therefore $3$, $22$, $2111$ and $2110$.

\medskip
\noindent
Although the previous result is completely general, if we require that $\alphabet[]{}=\lbrace 0,\dots,n\rbrace$, then the partition method described above can be rewritten into Algorithm~\ref{algo:widening}. While this is not included in the pseudocode provided, note that the algorithm should return an empty list if $a_0>n$, as in this case there would be no minimal word to look for.

\end{minipage}\hfill
\begin{minipage}{0.5\textwidth}
\begin{algorithm}[H]
\caption{\textsc{MinimalWords}}\label{algo:widening}
\KwIn{$\overline{w} = a_0\cdots a_p$, $\alphabet[]{}=\lbrace 0,\dots, n\rbrace$}
\KwOut{All minimal words of $\lang{\overline{w}}$}
Set $L$ to the empty list\\
\If{$a_0< n$}{
\For{$i\in \lbrace a_0+1,\dots, n\rbrace $}{Add the word $i$ to $L$}
}
\For{$k\in \lbrace 1,\dots,p\rbrace$}{
\If{$a_k< a_{k-1}$}{
\For{$i\in \lbrace a_k+1,\dots, a_{k-1}\rbrace $}{Add the word $a_0\cdots a_{k-1}i$ to $L$}
}
}
\For{$i\in \lbrace 0,\dots, a_p\rbrace $}{Add the word $a_0\cdots a_{p}i$ to $L$}
\Return{$L$}
\end{algorithm}
\end{minipage}

\subsection{Canonical FDAGs}\label{ss:canonical}
FDAGs are unordered objects, like the trees they compress, and therefore their enumeration requires to reflect this nature. In practice, finding a systematic way to order them makes it possible to design a simpler reduction rule, as done for trees \cite{nakano2003efficient}, ignoring the combinatorics of permutations. The purpose of this subsection is to provide a unique way to order FDAGs. We show that such an order exists in Theorem~\ref{prop:dag:cutf}, unambiguously characterizing FDAGs. The approach chosen is based on the notion of topological order.

\begin{wrapfigure}[9]{R}{0.5\textwidth}
  \centering
 \begin{minipage}[c]{0.17\textwidth}
   \centering
 \def\xscale{1.5}
\def\yscale{1.5}
\def\nodescale{1}
\begin{tikzpicture}[xscale=\xscale,yscale=\yscale]
\tikzstyle{arc}=[->-,>=latex]
\tikzstyle{noeud}=[draw,circle,scale=\nodescale*1]
\node[noeud,fill=blue] (E) at ({0},{0}) {};
\tikzstyle{noeud}=[draw,circle,scale=\nodescale*1]
\node[noeud,fill=red] (G) at ({-0.5},{1}) {};
\tikzstyle{noeud}=[draw,circle,scale=\nodescale*1]
\node[noeud,fill=magenta] (H) at ({0.5},{1}) {};
\node[noeud,fill=cyan] (F) at (1,0) {};
\draw[arc] (G) to [bend right=30 ](E) {};
\draw[arc] (G) to [bend left=30 ](E) {};
\draw[arc] (H)--(E) {};
\draw[arc] (H)--(F) {};
\end{tikzpicture}
\end{minipage}~
\begin{minipage}[c]{0.32\textwidth}
  \centering
  \small
  \begin{tabular}{c|cccc}
  \node &  \node[red] & \node[magenta] & \node[blue] &\node[cyan]\\
  \hline
  $\topord_1(\node)$ & $3$ & $2$ & $1$ & $0$\\
  $\topord_2(\node)$  &$3$ & $2$ & $0$ & $1$\\
   $\topord_3(\node)$  &$2$ & $3$ & $0$ & $1$\\
    $\topord_4(\node)$  &$2$ & $3$ & $1$ & $0$\\
     $\topord_5(\node)$  &$1$ & $3$ & $0$ & $2$\\
  \end{tabular}
 \end{minipage}
   \caption{The DAG on the left admits five topological orderings, which are shown in the table.}\label{fig:topord:nonunique}
\end{wrapfigure}

\paragraph{Topological ordering}
Let $D$ be a directed graph, where multiple arcs are allowed. A topological ordering on $D$ is an ordering of the vertices of $D$ such that for every arc $uv$ from vertex $u$ to vertex $v$, $u$ comes after $v$ in the ordering. Formally, $\topord : D \to [\![ 0, \# D -1]\!]$ is a topological ordering if and only if $\topord$ is bijective and $\topord(u)>\topord(v)$ for all $u,v\in D$ such that there exists at least one arc $uv$ in $D$. A well known result establishes that $D$ is a DAG if and only if it admits a topological ordering \cite{kahn1962topological}. Nonetheless, when a topological ordering exists, it is in general not unique -- see Figure~\ref{fig:topord:nonunique}. A reverse search enumeration of topological orderings of a given DAG can actually be found in \cite[Section 3.5]{avis1996reverse}.

\paragraph{Constrained topological ordering} We aim to reduce the number of possible topological orderings of a DAG by constraining them. Let $D$ be a DAG and $\topord$ a topological ordering. Taking advantage of the vertical hierarchy of DAG, our first constraint is
\begin{equation}\label{eq:topord:h}
\forall (u,v) \in D^2, \quad \height(u) > \height(v) \implies \topord(u)>\topord(v).
\end{equation}
\noindent
Applying (\ref{eq:topord:h}) to the topological orderings presented in Figure~\ref{fig:topord:nonunique}, $\topord_5$ must be removed, as $\topord_5(\node[cyan])>\topord_5(\node[red])$ and $\height(\node[red]) > \height(\node[cyan])$.

\medskip
\noindent
For any vertex $v$, and any $u\in\child(v)$, by definition, $\height(v)>\height(u)$. Therefore, there can be no arcs between vertices at same height. Any arbitrary order on them leads to a different topological ordering. The next constraint we propose relies on the lexicographical order:
\begin{equation}\label{eq:topord:child}
\forall (u,v) \in D^2, \quad \height(u) = \height(v) \text{ and } {\child}_{\topord}(u) \lex{>} {\child}_{\topord}(v)\implies \topord(u)>\topord(v),
\end{equation}
where $\child_{\topord}(v)$ is the list $\left[ \topord(v_i) : v_i \in  \child(v)\right]$ sorted by decreasing order w.r.t. the lexicographical order. In other words, $\childc(v)$ is a decreasing word -- see Definition~\ref{def:decreasingwords} -- on the alphabet $\alphabet[]{}=[\![ 0, \#D-1]\!]$. Table~\ref{tab:topord:child} illustrates the behavior of (\ref{eq:topord:child}) on the followed example of Figure~\ref{fig:topord:nonunique}.

\begin{table}[h]
  \centering
\begin{minipage}[c]{0.3\textwidth}
  \begin{tabular}{c|cc||c}
  \node &  \node[red] & \node[magenta] & (\ref{eq:topord:child})\\
  \hline
  ${\child}_{\topord_1}(\node)$ & $11$ & $10$ &   \checkmark\\
  ${\child}_{\topord_2}(\node)$  &$00$  & $10$   & \xmark \\
  ${\child}_{\topord_3}(\node)$& $00$  & $10$   & \checkmark \\
    ${\child}_{\topord_4}(\node)$& $11$  & $10$  & \xmark \\
  \end{tabular}
 \end{minipage}~
 \begin{minipage}[c]{0.6\textwidth}
   \caption{Application of (\ref{eq:topord:child}) to the remaining topological orderings of Figure~\ref{fig:topord:nonunique} that satisfy (\ref{eq:topord:h}). As ${\child}_{\topord}(\node[blue])={\child}_{\topord}(\node[cyan])$, we only need to consider vertices \node[red] and \node[magenta]. As ${\topord_i}(\node[red])>{\topord_i}(\node[magenta]) \iff i\in \lbrace 1,2\rbrace$, the only orderings that are kept are $\topord_1$ and $\topord_3$.
   }\label{tab:topord:child}
 \end{minipage}
\end{table}

\noindent
The combination of those two constraints imposes uniqueness in all cases except when there exists $(u,v) \in D^2$ such that ${\child}_{\topord}(u)={\child}_{\topord}(v)$ and $u\neq v$. It should be clear that if we impose the upcoming condition (\ref{eq:dag:forest}), such a pathological case can not occur.
\begin{equation}\label{eq:dag:forest}
\forall (u, v) \in D^2, \quad u\neq v \implies \child(u)\neq \child(v)
\end{equation}

\medskip
\noindent
Upcoming Theorem~\ref{prop:dag:cutf} establishes that a DAG compresses a forest if and only if the topological order constrained by (\ref{eq:topord:h}) and (\ref{eq:topord:child}) is unique. In other words, an unambiguous characterization of FDAGs is exhibited.

\begin{theorem}\label{prop:dag:cutf}
The following statements are equivalent:
\begin{enumerate}[label=(\roman*)]
\item $D$ fulfills (\ref{eq:dag:forest}),
\item there exists a unique topological ordering $\topord$ of $D$ that satisfies both constraints (\ref{eq:topord:h}) and (\ref{eq:topord:child}),
\item there exists a unique forest $F\in\forest$ -- cf. (\ref{eq:forest}) -- such that $D=\red(F)$,
\end{enumerate}
where $\red$ is the DAG reduction operation defined in Subsection~\ref{ss:def:treedag}.
\end{theorem}
\begin{proof}
$(i)\iff (ii)$ follows from the above discussion. $(iii)\implies (i)$ follows from the definition of $\red$. Indeed, if there was two distinct vertices $(u,v)\in D^2$ with the same multiset of children, they would have been compressed as a unique vertex in the reduction. We now prove that $(i) \implies (iii)$.

\medskip
\noindent
In the first place, if $D$ fulfills (\ref{eq:dag:forest}), then $D$ must admit a unique leaf, denoted  by $\leaves(D)$. Indeed, if there were two leaves $l_1$ and $l_2$, we would have $\height(l_1) = \height(l_2) = 0$ but also $\child(l_1) = \child(l_2) = \emptyset$, which would violate (\ref{eq:dag:forest}). Let $r_1, \dots, r_k$ be the vertices in $D$ that have no parent. We define $D_1,\dots, D_k$ as the subDAG rooted respectively in $r_1,\dots, r_k$. Then, we define $T_i = \red^{-1}(D_i)$ and $F=\lbrace T_1,\dots, T_k\rbrace$. The $T_i$'s are well defined as all vertices in $D$ (consequently in $D_i$) have a different multiset of children, and therefore compress distinct subtrees -- i.e. $F$ fulfills (\ref{eq:forest}), therefore $F\in\forest$. Moreover, $D=\red(F)$.
\end{proof}

\begin{wrapfigure}[11]{R}{0.5\textwidth}
\vspace{-\baselineskip}
\begin{minipage}[c]{0.12\textwidth}
\centering
\def\xscale{0.7}
\def\yscale{0.7}
\def\nodescale{0.7}
\begin{tikzpicture}[xscale=\xscale,yscale=\yscale]
\tikzstyle{noeud}=[draw,circle,fill=blue,scale=\nodescale*1]
\node[noeud] (0) at (0,0) {};
\tikzstyle{noeud}=[draw,circle,fill=magenta,scale=\nodescale*1]
\node[noeud] (1) at (-1,1) {};
\tikzstyle{noeud}=[draw,circle,fill=cyan,scale=\nodescale*1]
\node[noeud] (2) at (0,1) {};
\tikzstyle{noeud}=[draw,circle,fill=brown,scale=\nodescale*1]
\node[noeud] (3) at (1,1) {};
\tikzstyle{noeud}=[draw,circle,fill=orange,scale=\nodescale*1]
\node[noeud] (4) at (-1,2) {};
\tikzstyle{noeud}=[draw,circle,fill=yellow,scale=\nodescale*1]
\node[noeud] (5) at (0,2) {};
\tikzstyle{fleche}=[->-,>=latex,black]
\draw[fleche] (1)--(0) {};
\draw[fleche] (2)--(0) {}node [right,midway,scale=\nodescale] {$2$};
\draw[fleche] (3)--(0) {}node [right,midway,scale=\nodescale] {$3$};
\draw[fleche] (4)--(1) {};
\draw[fleche] (5)--(2) {};
\draw[fleche] (5)--(1) {}node [left,midway,scale=\nodescale] {$2$};

\draw[->,red,ultra thick] (0,3)--(5) {};
\draw[->,red,ultra thick] (-1,3)--(4) {};
\draw[->,red,ultra thick] (1,2)--(3) {};
\end{tikzpicture}
\end{minipage}~
\begin{minipage}[c]{0.28\textwidth}
\centering
$\begin{array}{c|c:ccc:cc}
v & \node[blue] & \node[magenta] &\node[cyan]&\node[brown]&\node[orange]&\node[yellow]\\
\topord(v)&0&1&2&3&4&5 \\
\childc(v) & &0&00& 000 & 1 &211
\end{array}$
\end{minipage}
\caption{A FDAG $D$ (left) and its canonical ordering $\topord$ (right). Vertices that are at the same height are enclosed in the table between the dashed lines. Red arrows indicate the roots of the trees of the forest that is compressed by $D$.}
\label{fig:ordering:canonical}
\end{wrapfigure}

\medskip
\noindent
In the sequel of the article, we shall only consider FDAGs. Consequently, from Proposition~\ref{prop:dag:cutf}, they admit a unique topological ordering $\topord$ satisfying both constraints (\ref{eq:topord:h}) and (\ref{eq:topord:child}), called \emph{canonical ordering}. Thus, for any FDAG $D$, the associated canonical ordering $\topord$ will be implicitly defined. The vertices will be numbered accordingly to their ordering, i.e $D=(v_0,\dots,v_n)$ with $\topord(v_i)=i$. Finally, as a consequence of constraints (\ref{eq:topord:h}) and (\ref{eq:topord:child}), note that $D$ can be partitioned in subsets of vertices with same height, each of them containing only consecutive numbered vertices. Figure~\ref{fig:ordering:canonical} provides an example of a FDAG and its canonical ordering.

\subsection{Expansion rules}\label{ss:rules}
Reverse search techniques implies finding reduction rules, and then inverse them. Equally, we will define instead three expansion rules, of which inverse will be reduction rules. An expansion rule takes a FDAG and create a new DAG, that is ``expanded'' in the sense of having either more vertices or more arcs. Our rules are analysed at the end of the subsection, where notably we prove in Proposition~\ref{prop:canonical} that expansion rules preserve the canonicalness. Moreover, we show in Proposition~\ref{prop:bij} that they are ``bijective'': any FDAG is in the image of a unique FDAG through the expansion rules. We begin with a preliminary definition.

\begin{definition}
Let $D$ be a FDAG, with $D=(v_0,\dots,v_n)$. We define the two following alphabets
\begin{align*}
\alphabet{=} &= \lbrace \topord(v) : v\in D, \height(v) = \height(v_n) \rbrace =  \lbrace p+1,\dots,n\rbrace, \\
\alphabet{<} &= \lbrace \topord(v)  : v\in D, \height(v) < \height(v_n) \rbrace = \lbrace 0,\dots,p\rbrace,
\end{align*}
where $p\in[\![0,n-1]\!]$ and $\topord(\cdot)$ is the canonical ordering of $D$.
\end{definition}
\noindent
In other words, $\alphabet{=}$ contains the indices of all vertices that have the same height as the vertex with the highest index according to $\topord$, and $\alphabet{<}$ the indices of all vertices that have an inferior height. The FDAG presented in Figure~\ref{fig:ordering:canonical} will serve as a guideline example all along this subsection. Here, we have $\alphabet{=}=\lbrace 4,5\rbrace$ and $\alphabet{<} = \lbrace 0,1,2,3\rbrace$.

\medskip
\noindent
The three expansion rules are now introduced. Let $D=(v_0,\dots,v_n)$. Each of these rules is associated with an explicit symbol, which may be used, when necessary, to designate the rule afterward. It is worth noting that all of these rules will operate according to the vertex of highest index, $v_n$.

\paragraph{Branching rule \branching}
This rule adds an arc between $v_n$ and a vertex below. The end vertex of the new arc is chosen such that $\childc(v_n)$ remains a decreasing word. In Figure~\ref{fig:rules:branching}, \branching is applied on our guideline example.

\begin{definition}
\branching Let $\childc(v_n) = a_0\cdots a_m$. Choose $a_{m+1}\in \alphabet{<}$ such that $a_{m} \lex{\geq} a_{m+1}$ and add an arc between $\topord^{-1}(a_{m+1})$ and $v_n$.
\end{definition}

\begin{figure}[h]
\centering
\begin{subfigure}[t]{0.45\textwidth}
\begin{minipage}[c]{0.37\textwidth}
\centering
\def\xscale{0.7}
\def\yscale{0.7}
\def\nodescale{0.7}

\begin{tikzpicture}[xscale=\xscale,yscale=\yscale]
\tikzstyle{noeud}=[draw,circle,fill=blue,scale=\nodescale*1]
\node[noeud] (0) at (0,0) {};
\tikzstyle{noeud}=[draw,circle,fill=magenta,scale=\nodescale*1]
\node[noeud] (1) at (-1,1) {};
\tikzstyle{noeud}=[draw,circle,fill=cyan,scale=\nodescale*1]
\node[noeud] (2) at (0,1) {};
\tikzstyle{noeud}=[draw,circle,fill=brown,scale=\nodescale*1]
\node[noeud] (3) at (1,1) {};
\tikzstyle{noeud}=[draw,circle,fill=orange,scale=\nodescale*1]
\node[noeud] (4) at (-1,2) {};
\tikzstyle{noeud}=[draw,circle,fill=yellow,scale=\nodescale*1]
\node[noeud] (5) at (0,2) {};
\tikzstyle{fleche}=[->-,>=latex,black]
\draw[fleche] (1)--(0) {};
\draw[fleche] (2)--(0) {}node [right,midway,scale=\nodescale] {$2$};
\draw[fleche] (3)--(0) {}node [right,midway,scale=\nodescale] {$3$};
\draw[fleche] (4)--(1) {};
\draw[fleche] (5)--(2) {};
\draw[fleche] (5)--(1) {}node [left,midway,red,scale=\nodescale] {$3$};
\end{tikzpicture}
\end{minipage}~
\begin{minipage}[c]{0.52\textwidth}
$\begin{array}{c|cc}
v & \dots&\node[yellow]\\
\topord(v)&\dots&5 \\
\childc(v) & \dots&211\textcolor{red}{1}
\end{array}$
\end{minipage}
\caption{\label{fig:rules:branching:a}}
\end{subfigure}~
\begin{subfigure}[t]{0.45\textwidth}
\begin{minipage}[c]{0.37\textwidth}
\centering
\def\xscale{0.7}
\def\yscale{0.7}
\def\nodescale{0.7}

\begin{tikzpicture}[xscale=\xscale,yscale=\yscale]
\tikzstyle{noeud}=[draw,circle,fill=blue,scale=\nodescale*1]
\node[noeud] (0) at (0,0) {};
\tikzstyle{noeud}=[draw,circle,fill=magenta,scale=\nodescale*1]
\node[noeud] (1) at (-1,1) {};
\tikzstyle{noeud}=[draw,circle,fill=cyan,scale=\nodescale*1]
\node[noeud] (2) at (0,1) {};
\tikzstyle{noeud}=[draw,circle,fill=brown,scale=\nodescale*1]
\node[noeud] (3) at (1,1) {};
\tikzstyle{noeud}=[draw,circle,fill=orange,scale=\nodescale*1]
\node[noeud] (4) at (-1,2) {};
\tikzstyle{noeud}=[draw,circle,fill=yellow,scale=\nodescale*1]
\node[noeud] (5) at (0,2) {};
\tikzstyle{fleche}=[->-,>=latex,black]
\draw[fleche] (1)--(0) {};
\draw[fleche] (2)--(0) {}node [right,midway,scale=\nodescale] {$2$};
\draw[fleche] (3) to [bend left =30]node [right,midway,scale=\nodescale] {$3$} (0) {};
\draw[fleche] (4)--(1) {};
\draw[fleche] (5)--(2) {};
\draw[fleche] (5)--(1) {}node [left,midway,scale=\nodescale] {$2$};
\draw[fleche,red] (5) to [bend left=60](0) {};
\end{tikzpicture}
\end{minipage}~
\begin{minipage}[c]{0.52\textwidth}
$\begin{array}{c|cc}
v & \dots&\node[yellow]\\
\topord(v)&\dots&5 \\
\childc(v) & \dots&211\textcolor{red}{0}
\end{array}$
\end{minipage}
\caption{\label{fig:rules:branching:b}}
\end{subfigure}
\caption{Branching rule applied to the FDAG of Figure~\ref{fig:ordering:canonical}. As $\childc(v_5) = 211$, the only letters $a$ we can pick from $\alphabet{<}=\lbrace 0,1,2,3\rbrace$, satisfying $a\lex{\leq} 1$, are $0$ and $1$. The only two possibles outcomes of \branching are the words (\subref{fig:rules:branching:a}) $2111$ and (\subref{fig:rules:branching:b}) $2110$.}
\label{fig:rules:branching}
\vspace{-\baselineskip}
\end{figure}

\paragraph{Elongation rule \elongation}
This rule adds a new vertex $v_{n+1}$ such that $\height(v_{n+1})=\height(v_n)+1$. Consequently, the alphabets change and become $\alphabet{=} = \lbrace n+1 \rbrace$ and $\alphabet{<} = \lbrace 0, \dots, n \rbrace$. Note that after using this rule, it is not possible to ever add a new vertex at height $\height(v_n)$. See Figure~\ref{fig:rules:elongation} for an illustration of this rule on the guideline example.

\begin{definition}
\elongation Add new vertex $v_{n+1}$ such that $\childc(v_{n+1})=a_0 \in \alphabet{=}$.
\end{definition}

\begin{figure}[h]
\centering
\begin{subfigure}[t]{0.45\textwidth}
\begin{minipage}[c]{0.37\textwidth}
\centering
\def\xscale{0.7}
\def\yscale{0.7}
\def\nodescale{0.7}

\begin{tikzpicture}[xscale=\xscale,yscale=\yscale]
\tikzstyle{noeud}=[draw,circle,fill=blue,scale=\nodescale*1]
\node[noeud] (0) at (0,0) {};
\tikzstyle{noeud}=[draw,circle,fill=magenta,scale=\nodescale*1]
\node[noeud] (1) at (-1,1) {};
\tikzstyle{noeud}=[draw,circle,fill=cyan,scale=\nodescale*1]
\node[noeud] (2) at (0,1) {};
\tikzstyle{noeud}=[draw,circle,fill=brown,scale=\nodescale*1]
\node[noeud] (3) at (1,1) {};
\tikzstyle{noeud}=[draw,circle,fill=orange,scale=\nodescale*1]
\node[noeud] (4) at (-1,2) {};
\tikzstyle{noeud}=[draw,circle,fill=yellow,scale=\nodescale*1]
\node[noeud] (5) at (0,2) {};
\tikzstyle{noeud}=[draw,circle,red,fill=red,scale=\nodescale*1]
\node[noeud] (6) at (0,3) {};
\tikzstyle{fleche}=[->-,>=latex,black]
\draw[fleche] (1)--(0) {};
\draw[fleche] (2)--(0) {}node [right,midway,scale=\nodescale] {$2$};
\draw[fleche] (3)--(0) {}node [right,midway,scale=\nodescale] {$3$};
\draw[fleche] (4)--(1) {};
\draw[fleche] (5)--(2) {};
\draw[fleche] (5)--(1) {}node [left,midway,scale=\nodescale] {$2$};
\draw[fleche,red] (6)--(4) {};
\end{tikzpicture}
\end{minipage}~
\begin{minipage}[c]{0.52\textwidth}
$\begin{array}{c|ccc}
v & \dots&\node[yellow] & \node[red]\\
\topord(v)&\dots&5 & \textcolor{red}{6}\\
\childc(v) & \dots&211& \textcolor{red}{4}
\end{array}$
\end{minipage}
\caption{\label{fig:rules:elongation:a}}
\end{subfigure}~
\begin{subfigure}[t]{0.45\textwidth}
\begin{minipage}[c]{0.37\textwidth}
\centering
\def\xscale{0.7}
\def\yscale{0.7}
\def\nodescale{0.7}

\begin{tikzpicture}[xscale=\xscale,yscale=\yscale]
\tikzstyle{noeud}=[draw,circle,fill=blue,scale=\nodescale*1]
\node[noeud] (0) at (0,0) {};
\tikzstyle{noeud}=[draw,circle,fill=magenta,scale=\nodescale*1]
\node[noeud] (1) at (-1,1) {};
\tikzstyle{noeud}=[draw,circle,fill=cyan,scale=\nodescale*1]
\node[noeud] (2) at (0,1) {};
\tikzstyle{noeud}=[draw,circle,fill=brown,scale=\nodescale*1]
\node[noeud] (3) at (1,1) {};
\tikzstyle{noeud}=[draw,circle,fill=orange,scale=\nodescale*1]
\node[noeud] (4) at (-1,2) {};
\tikzstyle{noeud}=[draw,circle,fill=yellow,scale=\nodescale*1]
\node[noeud] (5) at (0,2) {};
\tikzstyle{noeud}=[draw,circle,red,fill=red,scale=\nodescale*1]
\node[noeud] (6) at (0,3) {};
\tikzstyle{fleche}=[->-,>=latex,black]
\draw[fleche] (1)--(0) {};
\draw[fleche] (2)--(0) {}node [right,midway,scale=\nodescale] {$2$};
\draw[fleche] (3)--(0) node [right,midway,scale=\nodescale] {$3$}  {};
\draw[fleche] (4)--(1) {};
\draw[fleche] (5)--(2) {};
\draw[fleche] (5)--(1) {}node [left,midway,scale=\nodescale] {$2$};
\draw[fleche,red] (6)--(5) {};
\end{tikzpicture}
\end{minipage}~
\begin{minipage}[c]{0.52\textwidth}
$\begin{array}{c|ccc}
v & \dots&\node[yellow]& \node[red]\\
\topord(v)&\dots&5 &\textcolor{red}{6} \\
\childc(v) & \dots&211 & \textcolor{red}{5}
\end{array}$
\end{minipage}
\caption{\label{fig:rules:elongation:b}}
\end{subfigure}
\caption{Elongation rule applied to the FDAG of Figure~\ref{fig:ordering:canonical}. As $\alphabet{=}=\lbrace 4,5\rbrace$, there are only two choices leading to (\subref{fig:rules:branching:a}) $\childc(v_6) = 4$ and (\subref{fig:rules:branching:b}) $\childc(v_6) = 5$. The alphabets become $\alphabet{<}=\lbrace 0,\dots,5 \rbrace$ and $\alphabet{=}=\lbrace 6 \rbrace$.}
\label{fig:rules:elongation}
\vspace{-\baselineskip}
\end{figure}

\paragraph{Widening rule \widening}
This rule adds a new vertex $v_{n+1}$ at height $\height(v_n)$. The vertex is added with children that respects the canonicalness of the DAG, that is, such that $\childc(v_{n+1})\lex{>} \childc(v_n)$ -- as in condition (\ref{eq:topord:child}). In other terms, denoting $\lang[<]{}$ the language of decreasing words on alphabet $\alphabet[]{<}$, and with $\overline{w}=\childc(v_n)$, $\childc(v_{n+1})$ must be chosen in $\lang[<]{\overline{w}}$ -- see Definition~\ref{def:boundedwords}. However, this set is infinite, so we restrict $\childc(v_{n+1})$ to be chosen among the minimal words of $\lang[<]{\overline{w}}$. It follows from the definition of suffix-cut operator $\scut(\cdot)$ that, by inverting the said operator, the other words in $\lang[<]{\overline{w}}$ can be obtained by performing repeated \branching operations. Finally, this new vertex is added to $\alphabet{=}$.

\begin{definition}
\widening Add new vertex $v_{n+1}$ such that $$\childc(v_{n+1})\in \left\lbrace w \in \lang[<]{\overline{w}} : w \text{ is a minimal word of } \lang[<]{\overline{w}}\right\rbrace$$
with $\overline{w}=\childc(v_n)$.
\end{definition}

\noindent
From Proposition~\ref{prop:widening} we now that such minimal words exist. We prove in the upcoming lemma that, as claimed, $\height(v_{n+1}) = \height(v_n)$.

\begin{lemma}
Any element of $\lang[<]{\overline{w}}$ defines a new vertex $v_{n+1}$ such that $\height(v_{n+1})=\height(v_n)$.
\end{lemma}
\begin{proof}
From the definition of $\height(\cdot)$ -- (\ref{eq:height}), it suffices to prove that $v_{n+1}$ admits at least one child at height $h=\height(v_n)-1$. Let us denote $b_0$ and $a_0$ the first letter of, respectively, $\childc(v_{n+1})$ and $\childc(v_n)$. Denoting $v=\topord^{-1}(b_0)$ and $u=\topord^{-1}(a_0)$, we already know that $\height(u)=h$ -- as $\topord$ respects (\ref{eq:topord:child}) and $\childc(v_n)$ is a decreasing word. Therefore, as by construction $\childc(v_{n+1})\lex{>}\childc(v_n)$, either (i) $b_0=a_0$ and therefore $v=u$, either (ii) $b_0\lex{>}a_0$. In the latter, as $\topord$ respects (\ref{eq:topord:h}) and (\ref{eq:topord:child}), $\height(v)\geq \height(u)=h$. But, as $b_0\in \alphabet{<}$, $\height(v)<\height(v_n)=h+1$. In both cases, $\height(v)=h$.
\end{proof}

\medskip
\noindent
Figure~\ref{fig:rules:widening} illustrates the use of the widening rule on the followed example. It should be noted that the possible outcomes of \widening are obtained by using Algorithm~\ref{algo:widening}, applied to $w=\childc(v_n)$ and $p$ -- with $\alphabet{<}=\lbrace 0,\dots,p\rbrace$.

\begin{figure}[h]
\centering
\begin{subfigure}[t]{0.45\textwidth}
\begin{minipage}[c]{0.37\textwidth}
\centering
\def\xscale{0.7}
\def\yscale{0.7}
\def\nodescale{0.7}

\begin{tikzpicture}[xscale=\xscale,yscale=\yscale]
\tikzstyle{noeud}=[draw,circle,fill=blue,scale=\nodescale*1]
\node[noeud] (0) at (0,0) {};
\tikzstyle{noeud}=[draw,circle,fill=magenta,scale=\nodescale*1]
\node[noeud] (1) at (-1,1) {};
\tikzstyle{noeud}=[draw,circle,fill=cyan,scale=\nodescale*1]
\node[noeud] (2) at (0,1) {};
\tikzstyle{noeud}=[draw,circle,fill=brown,scale=\nodescale*1]
\node[noeud] (3) at (1,1) {};
\tikzstyle{noeud}=[draw,circle,fill=orange,scale=\nodescale*1]
\node[noeud] (4) at (-1,2) {};
\tikzstyle{noeud}=[draw,circle,fill=yellow,scale=\nodescale*1]
\node[noeud] (5) at (0,2) {};
\tikzstyle{noeud}=[draw,circle,red,fill=red,scale=\nodescale*1]
\node[noeud] (6) at (1,2) {};
\tikzstyle{fleche}=[->-,>=latex,black]
\draw[fleche] (1)--(0) {};
\draw[fleche] (2)--(0) {}node [right,midway,scale=\nodescale] {$2$};
\draw[fleche] (3)--(0) {}node [right,midway,scale=\nodescale] {$3$};
\draw[fleche] (4)--(1) {};
\draw[fleche] (5)--(2) {};
\draw[fleche] (5)--(1) {}node [left,midway,scale=\nodescale] {$2$};
\draw[fleche,red] (6)--(3) {};
\end{tikzpicture}
\end{minipage}~
\begin{minipage}[c]{0.52\textwidth}
$\begin{array}{c|ccc}
v & \dots&\node[yellow] & \node[red]\\
\topord(v)&\dots&5 & \textcolor{red}{6}\\
\childc(v) & \dots&211& \textcolor{red}{3}
\end{array}$
\end{minipage}
\caption{\label{fig:rules:widening:a}}
\end{subfigure}~
\begin{subfigure}[t]{0.45\textwidth}
\begin{minipage}[c]{0.37\textwidth}
\centering
\def\xscale{0.7}
\def\yscale{0.7}
\def\nodescale{0.7}

\begin{tikzpicture}[xscale=\xscale,yscale=\yscale]
\tikzstyle{noeud}=[draw,circle,fill=blue,scale=\nodescale*1]
\node[noeud] (0) at (0,0) {};
\tikzstyle{noeud}=[draw,circle,fill=magenta,scale=\nodescale*1]
\node[noeud] (1) at (-1,1) {};
\tikzstyle{noeud}=[draw,circle,fill=cyan,scale=\nodescale*1]
\node[noeud] (2) at (0,1) {};
\tikzstyle{noeud}=[draw,circle,fill=brown,scale=\nodescale*1]
\node[noeud] (3) at (1,1) {};
\tikzstyle{noeud}=[draw,circle,fill=orange,scale=\nodescale*1]
\node[noeud] (4) at (-1,2) {};
\tikzstyle{noeud}=[draw,circle,fill=yellow,scale=\nodescale*1]
\node[noeud] (5) at (0,2) {};
\tikzstyle{noeud}=[draw,circle,red,fill=red,scale=\nodescale*1]
\node[noeud] (6) at (1,2) {};
\tikzstyle{fleche}=[->-,>=latex,black]
\draw[fleche] (1)--(0) {};
\draw[fleche] (2)--(0) {}node [right,midway,scale=\nodescale] {$2$};
\draw[fleche] (3)--(0) node [right,midway,scale=\nodescale] {$3$}  {};
\draw[fleche] (4)--(1) {};
\draw[fleche] (5)--(2) {};
\draw[fleche] (5)--(1) {}node [left,midway,scale=\nodescale] {$2$};
\draw[fleche,red] (6)--(2) {}node[left,midway,scale=\nodescale] {$2$};
\end{tikzpicture}
\end{minipage}~
\begin{minipage}[c]{0.52\textwidth}
$\begin{array}{c|ccc}
v & \dots&\node[yellow]& \node[red]\\
\topord(v)&\dots&5 &\textcolor{red}{6} \\
\childc(v) & \dots&211 & \textcolor{red}{22}
\end{array}$
\end{minipage}
\caption{\label{fig:rules:widening:b}}
\end{subfigure}
\begin{subfigure}[t]{0.45\textwidth}
\begin{minipage}[c]{0.37\textwidth}
\centering
\def\xscale{0.7}
\def\yscale{0.7}
\def\nodescale{0.7}

\begin{tikzpicture}[xscale=\xscale,yscale=\yscale]
\tikzstyle{noeud}=[draw,circle,fill=blue,scale=\nodescale*1]
\node[noeud] (0) at (0,0) {};
\tikzstyle{noeud}=[draw,circle,fill=magenta,scale=\nodescale*1]
\node[noeud] (1) at (-1,1) {};
\tikzstyle{noeud}=[draw,circle,fill=cyan,scale=\nodescale*1]
\node[noeud] (2) at (0,1) {};
\tikzstyle{noeud}=[draw,circle,fill=brown,scale=\nodescale*1]
\node[noeud] (3) at (1,1) {};
\tikzstyle{noeud}=[draw,circle,fill=orange,scale=\nodescale*1]
\node[noeud] (4) at (-1,2) {};
\tikzstyle{noeud}=[draw,circle,fill=yellow,scale=\nodescale*1]
\node[noeud] (5) at (0,2) {};
\tikzstyle{noeud}=[draw,circle,red,fill=red,scale=\nodescale*1]
\node[noeud] (6) at (1,2) {};
\tikzstyle{fleche}=[->-,>=latex,black]
\draw[fleche] (1)--(0) {};
\draw[fleche] (2)--(0) {}node [right,midway,scale=\nodescale] {$2$};
\draw[fleche] (3)--(0) {}node [right,midway,scale=\nodescale] {$3$};
\draw[fleche] (4)--(1) {};
\draw[fleche] (5)--(2) {};
\draw[fleche] (5)--(1) {}node [left,midway,scale=\nodescale] {$2$};
\draw[fleche,red] (6)--(2) {};
\draw[fleche,red] (6)--(1) {}node [left,midway,scale=\nodescale] {$3$};
\end{tikzpicture}
\end{minipage}~
\begin{minipage}[c]{0.52\textwidth}
$\begin{array}{c|ccc}
v & \dots&\node[yellow] & \node[red]\\
\topord(v)&\dots&5 & \textcolor{red}{6}\\
\childc(v) & \dots&211& \textcolor{red}{2111}
\end{array}$
\end{minipage}
\caption{\label{fig:rules:widening:c}}
\end{subfigure}~
\begin{subfigure}[t]{0.45\textwidth}
\begin{minipage}[c]{0.37\textwidth}
\centering
\def\xscale{0.7}
\def\yscale{0.7}
\def\nodescale{0.7}

\begin{tikzpicture}[xscale=\xscale,yscale=\yscale]
\tikzstyle{noeud}=[draw,circle,fill=blue,scale=\nodescale*1]
\node[noeud] (0) at (0,0) {};
\tikzstyle{noeud}=[draw,circle,fill=magenta,scale=\nodescale*1]
\node[noeud] (1) at (-1,1) {};
\tikzstyle{noeud}=[draw,circle,fill=cyan,scale=\nodescale*1]
\node[noeud] (2) at (0,1) {};
\tikzstyle{noeud}=[draw,circle,fill=brown,scale=\nodescale*1]
\node[noeud] (3) at (1,1) {};
\tikzstyle{noeud}=[draw,circle,fill=orange,scale=\nodescale*1]
\node[noeud] (4) at (-1,2) {};
\tikzstyle{noeud}=[draw,circle,fill=yellow,scale=\nodescale*1]
\node[noeud] (5) at (0,2) {};
\tikzstyle{noeud}=[draw,circle,red,fill=red,scale=\nodescale*1]
\node[noeud] (6) at (1,2) {};
\tikzstyle{fleche}=[->-,>=latex,black]
\draw[fleche] (1)--(0) {};
\draw[fleche] (2)--(0) {}node [right,midway,scale=\nodescale] {$2$};
\draw[fleche] (3)--(0) node [right,midway,scale=\nodescale] {$3$}  {};
\draw[fleche] (4)--(1) {};
\draw[fleche] (5)--(2) {};
\draw[fleche] (5)--(1) {}node [left,midway,scale=\nodescale] {$2$};
\draw[fleche,red] (6)--(2) {};
\draw[fleche,red] (6)--(1) {}node [left,midway,scale=\nodescale] {$2$};
\draw[fleche,red] (6)--(0) {};
\end{tikzpicture}
\end{minipage}~
\begin{minipage}[c]{0.52\textwidth}
$\begin{array}{c|ccc}
v & \dots&\node[yellow]& \node[red]\\
\topord(v)&\dots&5 &\textcolor{red}{6} \\
\childc(v) & \dots&211 & \textcolor{red}{2110}
\end{array}$
\end{minipage}
\caption{\label{fig:rules:widening:d}}
\end{subfigure}
\caption{We apply \widening to the FDAG of Figure~\ref{fig:ordering:canonical}. Here, $\alphabet{<}=\lbrace 0,1,2,3\rbrace$ and $\overline{w}=211$. As seen in Subsection~\ref{ss:words}, the minimal words of $\lang[<]{\overline{w}}$ are $3$, $22$, $2111$ and $2110$. Therefore, there are $4$ ways to add a new vertex $v_{6}$ via the widening rule, that are such that (\subref{fig:rules:widening:a}) $\childc(v_6)=3$, (\subref{fig:rules:widening:b}) $\childc(v_6)=22$, (\subref{fig:rules:widening:c}) $\childc(v_6)=2111$ or (\subref{fig:rules:widening:d}) $\childc(v_6)=2110$. Finally, we update $\alphabet{=}$ to be equal to $\lbrace 4,5, 6\rbrace$.}
\label{fig:rules:widening}
\vspace{-\baselineskip}
\end{figure}

\subsection{Analysis of the rules}\label{ss:analysis_rule}

Since our goal is to enumerate FDAGs, it is required that the expansion rules indeed construct FDAGs. This is achieved by virtue of the following proposition.

\begin{proposition}\label{prop:canonical}
The expansion rules preserve the canonicalness property.
\end{proposition}
\begin{proof}
Let $D=(v_0,\dots, v_n)$ be a FDAG. The proposition follows naturally from the definitions:
\begin{description}[leftmargin=!,labelwidth=2.5em,align=right]
\item[\branching] Let $a$ be the letter added to ${w}=\childc(v_n)$. As ${w} a\lex{>} {w} \lex{>} \childc(v_{n-1})$, the ordering is unchanged.
\item[\elongation] The new vertex $v_{n+1}$ is such that $\height(v_{n+1})>\height(v_n)$, so condition (\ref{eq:topord:h}) is still met.
\item[\widening] The new vertex $v_{n+1}$ is chosen so that $\height(v_{n+1})=\height(v_n)$ and $\childc(v_{n+1})\lex{>}\childc(v_n)$, so condition (\ref{eq:topord:child}) is also still met.
\end{description}
Therefore, any DAG obtained from $D$ is still a FDAG.
\end{proof}

\medskip
\noindent
Secondly, since our goal is to provide the FDAGs space with an enumeration tree, which will be explored via the expansion rules, it is important that these expansion rules are ``bijective'' in the following sense: for any FDAG $D$, there exists a unique FDAG $D'$ such that $D$ is obtained from $D'$ via one of the three rules \branching, \elongation or \widening.

\medskip
\noindent
\begin{minipage}{0.35\textwidth}
Such $D'$ can be constructed via Algorithm~\ref{algo:reduction} as shown in upcoming Proposition~\ref{prop:bij}. Conditional expressions applied to $D$ are used to determine which modification should be applied to construct $D'$. The gray symbol (in the algorithm) next to these modifications indicates which expansion rule allows to retrieve $D$ from $D'$.

\begin{proposition}\label{prop:bij}
Algorithm~\ref{algo:reduction} applied to any FDAG constructs the unique antecedent of this FDAG.
\end{proposition}

\end{minipage}\hfill
\begin{minipage}{0.6\textwidth}
\begin{algorithm}[H]
\algorithmfootnote{$\scut(\cdot)$ is the suffix-cut operator defined in (\ref{eq:scut}).}
\caption{\textsc{Antecedent}}\label{algo:reduction}
\KwIn{$D=(v_0,\dots,v_n)$; $w=\childc(v_n)$; $w' = \childc(v_{n-1})$}
\eIf{$v_n$ \textup{\textbf{is the only vertex of height}} $\height(v_n)$}{
	\eIf{$\#w=1$}{\textcolor{gray}{\elongation} Delete vertex $v_n$}{\textcolor{gray}{\branching} $w \gets \scut(w)$}
}{
\eIf{$w$ \textup{\textbf{is a minimal word of}} $\lang[<]{w'}$}{\textcolor{gray}{\widening} Delete vertex $v_n$}{\textcolor{gray}{\branching} $w \gets \scut(w)$}
}
\end{algorithm}
\end{minipage}

\medskip
\begin{proof}
Let $D=(v_0,\dots,v_n)$ be a FDAG. Let $w=\childc(v_n)$ and $w'=\childc(v_{n-1})$. Two cases can occur: (i) either $v_n$ is the only vertex at height $\height(v_n)$, (ii) or it is not.
\begin{enumerate}[label=(\roman*)]
\item It is clear in this case that $D$ can not be obtained from any FDAG via the rule \widening -- otherwise $v_n$ would not be alone at its height. Concerning \branching and \elongation, let us look at the number of children of $v_n$.
\begin{enumerate}
\item  If $v_n$ admits only one child, it must come from an \elongation step, since \branching would imply that $\#w\geq 2$. Therefore, in this case, $D$ can be retrieved among the outcomes of rule \elongation applied to $D'=(v_0,\dots,v_{n-1})$.
\item Otherwise, when $\#w>1$, $D$ can not come from an \elongation step, and must therefore come from \branching. Denoting $v'_n$ the vertex with list of children $\scut(w)$ -- see (\ref{eq:scut}), $D$ is one of the outcomes of $D'=(v_0,\dots,v_{n-1},v'_n)$ via \branching.
\end{enumerate}
\item following the same logic as (i), $D$ can not be obtained via \elongation. We discrimine between rules \widening and \branching be comparing $w$ and $w'$. If $w$ is a minimal word of $\lang[<]{w'}$, then $D$ can not be obtained from \branching -- this would break the canonical order. Therefore, in this case, $D$ is an outcome of rule \widening applied to $D'=(v_0,\dots,v_{n-1})$. Otherwise, if $w$ is not a minimal word, then it can not be obtained from \widening, and must come from a \branching step, applied to $D'=(v_0,\dots,v_{n-1},v'_n)$ where $\childc(v'_n)=\scut(w)$.
\end{enumerate}
Whatever the case among those evoked, they correspond exactly to the conditional expressions of the Algorithm~\ref{algo:reduction}, which therefore constructs the correct antecedent of $D$, which is unique by virtue of the previous discussion.
\end{proof}

\subsection{Enumeration tree}\label{ss:enumtree}
In this subsection, we construct the enumeration tree of FDAGs derived from the expansion rules of Subsection~\ref{ss:rules}. As aimed, their inverse is indeed a reduction rule.

\begin{theorem}\label{th:reduction}
Algorithm~\ref{algo:reduction} is a reduction rule, as defined in Subsection~\ref{ss:def:reverse}.
\end{theorem}
\begin{proof}
Let us denote $f(D)$ the output of Algorithm~\ref{algo:reduction} applied to a FDAG $D$. We need to prove that: (i) $f(D)$ is a subgraph of $D$ and (ii) for any $D\neq D_0$, there exists an integer $k$ such that $f^k(D)=D_0$, where $D_0$ is the FDAG with one vertex and no arcs.
\begin{enumerate}[label=(\roman*)]
\item Since Algorithm~\ref{algo:reduction} deletes either one vertex and its leaving arcs, or just one arc, $f(D)$ is indeed a subgraph of $D$.
\item The sequence of general term $f^k(D)$ is made of discrete objects whose size is strictly decreasing, therefore the sequence is finite and reaches $D_0$.
\end{enumerate}
\vspace{-\baselineskip}
\end{proof}

\medskip
\noindent
The associated expansion rule is exactly, in light of Proposition~\ref{prop:bij}, the union of the three expansion rules \elongation, \branching and \widening. Since $D_0$, the DAG with one vertex and no arcs, is a FDAG, by virtue of what precedes and with Algorithm~\ref{algo:reverse:def} -- here with $g(\cdot)=\top$, we just defined an enumeration tree covering the whole set of FDAGs, whose root is $D_0$. A fraction of this enumeration tree is shown in Figure~\ref{fig:enumtree}, illustrating the path from the root $D_0$ to the FDAG of Figure~\ref{fig:ordering:canonical}. Unexplored branches are ignored, but are still indicated by their respective root.

\begin{figure}
\centering

\def\xscale{0.53}
\def\yscale{0.7}
\def\nodescale{0.7}

\def\xdim{5*\xscale}
\def\ydim{5*\yscale}
\def\opacity{0.5}
\begin{tikzpicture}[remember picture]
    \node (1) at (0,0){\begin{tikzpicture}[xscale=\xscale,yscale=\yscale]
\tikzstyle{noeud}=[draw,circle,fill=blue,scale=\nodescale*1]
\node[noeud] (0a) at (0,0) {};
\end{tikzpicture}};
    \node (2) at (1*\xdim,0){\begin{tikzpicture}[xscale=\xscale,yscale=\yscale]
\tikzstyle{noeud}=[draw,circle,fill=blue,scale=\nodescale*1]
\node[noeud] (0a) at (0,0) {};
\tikzstyle{noeud}=[draw,circle,fill=magenta,scale=\nodescale*1]
\node[noeud] (1a) at (0,1) {};
\tikzstyle{fleche}=[->-,>=latex,black]
\draw[fleche] (1a)--(0a) {};\end{tikzpicture}};
    \node[opacity=\opacity] (3) at (\xdim,0.7*\ydim) {\begin{tikzpicture}[xscale=\xscale,yscale=\yscale]
\tikzstyle{noeud}=[draw,circle,fill=blue,scale=\nodescale*1]
\node[noeud] (0a) at (0,0) {};
\tikzstyle{noeud}=[draw,circle,fill=magenta,scale=\nodescale*1]
\node[noeud] (1a) at (0,1) {};
\tikzstyle{noeud}=[draw,circle,fill=cyan,scale=\nodescale*1]
\node[noeud] (2a) at (0,2) {};
\tikzstyle{fleche}=[->-,>=latex,black]
\draw[fleche] (1a)--(0a) {};
\draw[fleche] (2a)--(1a) {};
\end{tikzpicture}};
       \node[opacity=\opacity] (4) at (\xdim,-0.7*\ydim) {\begin{tikzpicture}[xscale=\xscale,yscale=\yscale]
\tikzstyle{noeud}=[draw,circle,fill=blue,scale=\nodescale*1]
\node[noeud] (0a) at (0,0) {};
\tikzstyle{noeud}=[draw,circle,fill=magenta,scale=\nodescale*1]
\node[noeud] (1a) at (0,1) {};
\tikzstyle{fleche}=[->-,>=latex,black]
\draw[fleche] (1a)--(0a) {}node [right,midway,scale=\nodescale] {$2$};
\end{tikzpicture}};
    \node (5) at (2*\xdim,0){\begin{tikzpicture}[xscale=\xscale,yscale=\yscale]
\tikzstyle{noeud}=[draw,circle,fill=blue,scale=\nodescale*1]
\node[noeud] (0a) at (0,0) {};
\tikzstyle{noeud}=[draw,circle,fill=magenta,scale=\nodescale*1]
\node[noeud] (1a) at (-0.5,1) {};
\tikzstyle{noeud}=[draw,circle,fill=cyan,scale=\nodescale*1]
\node[noeud] (2a) at (0.5,1) {};
\tikzstyle{fleche}=[->-,>=latex,black]
\draw[fleche] (1a)--(0a) {};
\draw[fleche] (2a)--(0a) {}node [right,midway,scale=\nodescale] {$2$};
\end{tikzpicture}};
       \node[red] (a) at (2*\xdim,0.4*\ydim){\elong};
       \node[opacity=\opacity] (6) at (1.7*\xdim,0.7*\ydim){\begin{tikzpicture}[xscale=\xscale,yscale=\yscale]
\tikzstyle{noeud}=[draw,circle,fill=blue,scale=\nodescale*1]
\node[noeud] (0a) at (0,0) {};
\tikzstyle{noeud}=[draw,circle,fill=magenta,scale=\nodescale*1]
\node[noeud] (1a) at (-0.5,1) {};
\tikzstyle{noeud}=[draw,circle,fill=cyan,scale=\nodescale*1]
\node[noeud] (2a) at (0.5,1) {};
\tikzstyle{noeud}=[draw,circle,fill=brown,scale=\nodescale*1]
\node[noeud] (3a) at (-0.5,2) {};
\tikzstyle{fleche}=[->-,>=latex,black]
\draw[fleche] (1a)--(0a) {};
\draw[fleche] (2a)--(0a) {}node [right,midway,scale=\nodescale] {$2$};
\draw[fleche] (3a)--(1a)  {};
\end{tikzpicture}};
        \node[opacity=\opacity] (7) at (2.3*\xdim,0.7*\ydim){\begin{tikzpicture}[xscale=\xscale,yscale=\yscale]
\tikzstyle{noeud}=[draw,circle,fill=blue,scale=\nodescale*1]
\node[noeud] (0a) at (0,0) {};
\tikzstyle{noeud}=[draw,circle,fill=magenta,scale=\nodescale*1]
\node[noeud] (1a) at (-0.5,1) {};
\tikzstyle{noeud}=[draw,circle,fill=cyan,scale=\nodescale*1]
\node[noeud] (2a) at (0.5,1) {};
\tikzstyle{noeud}=[draw,circle,fill=brown,scale=\nodescale*1]
\node[noeud] (3a) at (0.5,2) {};
\tikzstyle{fleche}=[->-,>=latex,black]
\draw[fleche] (1a)--(0a) {};
\draw[fleche] (2a)--(0a) {}node [right,midway,scale=\nodescale] {$2$};
\draw[fleche] (3a)--(2a) {};
\end{tikzpicture}};
   \node[opacity=\opacity] (8) at (2*\xdim,-0.7*\ydim){\begin{tikzpicture}[xscale=\xscale,yscale=\yscale]
\tikzstyle{noeud}=[draw,circle,fill=blue,scale=\nodescale*1]
\node[noeud] (0a) at (0,0) {};
\tikzstyle{noeud}=[draw,circle,fill=magenta,scale=\nodescale*1]
\node[noeud] (1a) at (-0.5,1) {};
\tikzstyle{noeud}=[draw,circle,fill=cyan,scale=\nodescale*1]
\node[noeud] (2a) at (0.5,1) {};
\tikzstyle{fleche}=[->-,>=latex,black]
\draw[fleche] (1a)--(0a) {};
\draw[fleche] (2a)--(0a) {}node [right,midway,scale=\nodescale] {$3$};
\end{tikzpicture}};
    \node (9) at (3*\xdim,0){\begin{tikzpicture}[xscale=\xscale,yscale=\yscale]
\tikzstyle{noeud}=[draw,circle,fill=blue,scale=\nodescale*1]
\node[noeud] (0a) at (0,0) {};
\tikzstyle{noeud}=[draw,circle,fill=magenta,scale=\nodescale*1]
\node[noeud] (1a) at (-1,1) {};
\tikzstyle{noeud}=[draw,circle,fill=cyan,scale=\nodescale*1]
\node[noeud] (2a) at (0,1) {};
\tikzstyle{noeud}=[draw,circle,fill=brown,scale=\nodescale*1]
\node[noeud] (3a) at (1,1) {};
\tikzstyle{fleche}=[->-,>=latex,black]
\draw[fleche] (1a)--(0a) {};
\draw[fleche] (2a)--(0a) {}node [right,midway,scale=\nodescale] {$2$};
\draw[fleche] (3a)--(0a) node [right,midway,scale=\nodescale] {$3$}  {};

\end{tikzpicture}};
    \node[opacity=\opacity] (10) at (3*\xdim,0.7*\ydim){\begin{tikzpicture}[xscale=\xscale,yscale=\yscale]
\tikzstyle{noeud}=[draw,circle,fill=blue,scale=\nodescale*1]
\node[noeud] (0a) at (0,0) {};
\tikzstyle{noeud}=[draw,circle,fill=magenta,scale=\nodescale*1]
\node[noeud] (1a) at (-1.5,1) {};
\tikzstyle{noeud}=[draw,circle,fill=cyan,scale=\nodescale*1]
\node[noeud] (2a) at (-0.5,1) {};
\tikzstyle{noeud}=[draw,circle,fill=brown,scale=\nodescale*1]
\node[noeud] (3a) at (0.5,1) {};
\tikzstyle{noeud}=[draw,circle,fill=orange,scale=\nodescale*1]
\node[noeud] (4a) at (1.5,1) {};
\tikzstyle{fleche}=[->-,>=latex,black]
\draw[fleche] (1a)--(0a) {};
\draw[fleche] (2a)--(0a) {}node [right,midway,scale=\nodescale] {$2$};
\draw[fleche] (3a)--(0a) node [right,midway,scale=\nodescale] {$3$}  {};
\draw[fleche] (4a)--(0a) node [right,midway,scale=\nodescale] {$4$}{};
\end{tikzpicture}};
    \node[opacity=\opacity] (11) at  (3*\xdim,-0.7*\ydim){\begin{tikzpicture}[xscale=\xscale,yscale=\yscale]
\tikzstyle{noeud}=[draw,circle,fill=blue,scale=\nodescale*1]
\node[noeud] (0a) at (0,0) {};
\tikzstyle{noeud}=[draw,circle,fill=magenta,scale=\nodescale*1]
\node[noeud] (1a) at (-1,1) {};
\tikzstyle{noeud}=[draw,circle,fill=cyan,scale=\nodescale*1]
\node[noeud] (2a) at (0,1) {};
\tikzstyle{noeud}=[draw,circle,fill=brown,scale=\nodescale*1]
\node[noeud] (3a) at (1,1) {};
\tikzstyle{fleche}=[->-,>=latex,black]
\draw[fleche] (1a)--(0a) {};
\draw[fleche] (2a)--(0a) {}node [right,midway,scale=\nodescale] {$2$};
\draw[fleche] (3a)--(0a) node [right,midway,scale=\nodescale] {$4$}  {};
\end{tikzpicture}};
\node[red] (b) at (3.5*\xdim, 0){\elong};
    \node[opacity=\opacity] (12) at (4*\xdim,0.7*\ydim){\begin{tikzpicture}[xscale=\xscale,yscale=\yscale]
\tikzstyle{noeud}=[draw,circle,fill=blue,scale=\nodescale*1]
\node[noeud] (0a) at (0,0) {};
\tikzstyle{noeud}=[draw,circle,fill=magenta,scale=\nodescale*1]
\node[noeud] (1a) at (-1,1) {};
\tikzstyle{noeud}=[draw,circle,fill=cyan,scale=\nodescale*1]
\node[noeud] (2a) at (0,1) {};
\tikzstyle{noeud}=[draw,circle,fill=brown,scale=\nodescale*1]
\node[noeud] (3a) at (1,1) {};
\tikzstyle{noeud}=[draw,circle,fill=orange,scale=\nodescale*1]
\node[noeud] (4a) at (1,2) {};
\tikzstyle{fleche}=[->-,>=latex,black]
\draw[fleche] (1a)--(0a) {};
\draw[fleche] (2a)--(0a) {}node [right,midway,scale=\nodescale] {$2$};
\draw[fleche] (3a)--(0a) node [right,midway,scale=\nodescale] {$3$}  {};
\draw[fleche] (4a)--(3a) {};
\end{tikzpicture}};
    \node[opacity=\opacity] (13) at (4*\xdim,0){\begin{tikzpicture}[xscale=\xscale,yscale=\yscale]
\tikzstyle{noeud}=[draw,circle,fill=blue,scale=\nodescale*1]
\node[noeud] (0a) at (0,0) {};
\tikzstyle{noeud}=[draw,circle,fill=magenta,scale=\nodescale*1]
\node[noeud] (1a) at (-1,1) {};
\tikzstyle{noeud}=[draw,circle,fill=cyan,scale=\nodescale*1]
\node[noeud] (2a) at (0,1) {};
\tikzstyle{noeud}=[draw,circle,fill=brown,scale=\nodescale*1]
\node[noeud] (3a) at (1,1) {};
\tikzstyle{noeud}=[draw,circle,fill=orange,scale=\nodescale*1]
\node[noeud] (4a) at (0,2) {};
\tikzstyle{fleche}=[->-,>=latex,black]
\draw[fleche] (1a)--(0a) {};
\draw[fleche] (2a)--(0a) {}node [right,midway,scale=\nodescale] {$2$};
\draw[fleche] (3a)--(0a) node [right,midway,scale=\nodescale] {$3$}  {};
\draw[fleche] (4a)--(2a) {};

\end{tikzpicture}};
    \node (14) at (4*\xdim,-0.7*\ydim){\begin{tikzpicture}[xscale=\xscale,yscale=\yscale]
\tikzstyle{noeud}=[draw,circle,fill=blue,scale=\nodescale*1]
\node[noeud] (0a) at (0,0) {};
\tikzstyle{noeud}=[draw,circle,fill=magenta,scale=\nodescale*1]
\node[noeud] (1a) at (-1,1) {};
\tikzstyle{noeud}=[draw,circle,fill=cyan,scale=\nodescale*1]
\node[noeud] (2a) at (0,1) {};
\tikzstyle{noeud}=[draw,circle,fill=brown,scale=\nodescale*1]
\node[noeud] (3a) at (1,1) {};
\tikzstyle{noeud}=[draw,circle,fill=orange,scale=\nodescale*1]
\node[noeud] (4a) at (-1,2) {};
\tikzstyle{fleche}=[->-,>=latex,black]
\draw[fleche] (1a)--(0a) {};
\draw[fleche] (2a)--(0a) {}node [right,midway,scale=\nodescale] {$2$};
\draw[fleche] (3a)--(0a) node [right,midway,scale=\nodescale] {$3$}  {};
\draw[fleche] (4a)--(1a) {};
\end{tikzpicture}};

    \node[opacity=\opacity] (15) at (4*\xdim,-1.5*\ydim){\begin{tikzpicture}[xscale=\xscale,yscale=\yscale]
\tikzstyle{noeud}=[draw,circle,fill=blue,scale=\nodescale*1]
\node[noeud] (0a) at (0,0) {};
\tikzstyle{noeud}=[draw,circle,fill=magenta,scale=\nodescale*1]
\node[noeud] (1a) at (-1,1) {};
\tikzstyle{noeud}=[draw,circle,fill=cyan,scale=\nodescale*1]
\node[noeud] (2a) at (0,1) {};
\tikzstyle{noeud}=[draw,circle,fill=brown,scale=\nodescale*1]
\node[noeud] (3a) at (1,1) {};
\tikzstyle{noeud}=[draw,circle,fill=orange,scale=\nodescale*1]
\node[noeud] (4a) at (-1,2) {};
\tikzstyle{noeud}=[draw,circle,fill=yellow,scale=\nodescale*1]
\node[noeud] (5a) at (-1,3) {};
\tikzstyle{fleche}=[->-,>=latex,black]
\draw[fleche] (1a)--(0a) {};
\draw[fleche] (2a)--(0a) {}node [right,midway,scale=\nodescale] {$2$};
\draw[fleche] (3a)--(0a) node [right,midway,scale=\nodescale] {$3$}  {};
\draw[fleche] (4a)--(1a) {};
\draw[fleche] (5a)--(4a) {};
\end{tikzpicture}};
    \node[red] (j) at (4.8*\xdim,-0.7*\ydim){\branch};
     \node[opacity=\opacity] (16) at (5.6*\xdim,0.2*\ydim){\begin{tikzpicture}[xscale=\xscale,yscale=\yscale]
\tikzstyle{noeud}=[draw,circle,fill=blue,scale=\nodescale*1]
\node[noeud] (0a) at (0,0) {};
\tikzstyle{noeud}=[draw,circle,fill=magenta,scale=\nodescale*1]
\node[noeud] (1a) at (-1,1) {};
\tikzstyle{noeud}=[draw,circle,fill=cyan,scale=\nodescale*1]
\node[noeud] (2a) at (0,1) {};
\tikzstyle{noeud}=[draw,circle,fill=brown,scale=\nodescale*1]
\node[noeud] (3a) at (1,1) {};
\tikzstyle{noeud}=[draw,circle,fill=orange,scale=\nodescale*1]
\node[noeud] (4a) at (-1,2) {};
\tikzstyle{fleche}=[->-,>=latex,black]
\draw[fleche] (1a)--(0a) {};
\draw[fleche] (2a)--(0a) {}node [right,midway,scale=\nodescale] {$2$};
\draw[fleche] (3a)--(0a) node [right,midway,scale=\nodescale] {$3$}  {};
\draw[fleche] (4a)--(1a) node [right,midway,scale=\nodescale]{$2$}{};
\end{tikzpicture}};
     \node[opacity=\opacity] (38) at (5.6*\xdim,-0.4*\ydim){\begin{tikzpicture}[xscale=\xscale,yscale=\yscale]
\tikzstyle{noeud}=[draw,circle,fill=blue,scale=\nodescale*1]
\node[noeud] (0a) at (0,0) {};
\tikzstyle{noeud}=[draw,circle,fill=magenta,scale=\nodescale*1]
\node[noeud] (1a) at (-1,1) {};
\tikzstyle{noeud}=[draw,circle,fill=cyan,scale=\nodescale*1]
\node[noeud] (2a) at (0,1) {};
\tikzstyle{noeud}=[draw,circle,fill=brown,scale=\nodescale*1]
\node[noeud] (3a) at (1,1) {};
\tikzstyle{noeud}=[draw,circle,fill=orange,scale=\nodescale*1]
\node[noeud] (4a) at (-1,2) {};
\tikzstyle{fleche}=[->-,>=latex,black]
\draw[fleche] (1a)--(0a) {};
\draw[fleche] (2a)--(0a) {}node [right,midway,scale=\nodescale] {$2$};
\draw[fleche] (3a)--(0a) node [right,midway,scale=\nodescale] {$3$}  {};
\draw[fleche] (4a)--(1a) {};
\draw[fleche] (4a)to [bend right=90] (0a) {};

\end{tikzpicture}};
    \node[red] (c) at (4.8*\xdim, -1.4*\ydim){\widen};
         \node[opacity=\opacity] (17) at (5.6*\xdim,-1.1*\ydim){\begin{tikzpicture}[xscale=\xscale,yscale=\yscale]
\tikzstyle{noeud}=[draw,circle,fill=blue,scale=\nodescale*1]
\node[noeud] (0a) at (0,0) {};
\tikzstyle{noeud}=[draw,circle,fill=magenta,scale=\nodescale*1]
\node[noeud] (1a) at (-1,1) {};
\tikzstyle{noeud}=[draw,circle,fill=cyan,scale=\nodescale*1]
\node[noeud] (2a) at (0,1) {};
\tikzstyle{noeud}=[draw,circle,fill=brown,scale=\nodescale*1]
\node[noeud] (3a) at (1,1) {};
\tikzstyle{noeud}=[draw,circle,fill=orange,scale=\nodescale*1]
\node[noeud] (4a) at (-1,2) {};
\tikzstyle{noeud}=[draw,circle,fill=yellow,scale=\nodescale*1]
\node[noeud] (5a) at (0,2) {};
\tikzstyle{fleche}=[->-,>=latex,black]
\draw[fleche] (1a)--(0a) {};
\draw[fleche] (2a)--(0a) {}node [right,midway,scale=\nodescale] {$2$};
\draw[fleche] (3a)--(0a) node [right,midway,scale=\nodescale] {$3$}  {};
\draw[fleche] (4a)--(1a) {};
\draw[fleche] (5a)--(1a) node[right,midway,scale=\nodescale] {$2$}{};
\end{tikzpicture}};
              \node[opacity=\opacity] (18) at (5.6*\xdim,-1.7*\ydim){\begin{tikzpicture}[xscale=\xscale,yscale=\yscale]
\tikzstyle{noeud}=[draw,circle,fill=blue,scale=\nodescale*1]
\node[noeud] (0a) at (0,0) {};
\tikzstyle{noeud}=[draw,circle,fill=magenta,scale=\nodescale*1]
\node[noeud] (1a) at (-1,1) {};
\tikzstyle{noeud}=[draw,circle,fill=cyan,scale=\nodescale*1]
\node[noeud] (2a) at (0,1) {};
\tikzstyle{noeud}=[draw,circle,fill=brown,scale=\nodescale*1]
\node[noeud] (3a) at (1,1) {};
\tikzstyle{noeud}=[draw,circle,fill=orange,scale=\nodescale*1]
\node[noeud] (4a) at (-1,2) {};
\tikzstyle{noeud}=[draw,circle,fill=yellow,scale=\nodescale*1]
\node[noeud] (5a) at (0,2) {};
\tikzstyle{fleche}=[->-,>=latex,black]
\draw[fleche] (1a)--(0a) {};
\draw[fleche] (2a)--(0a) {}node [right,midway,scale=\nodescale] {$2$};
\draw[fleche] (3a)--(0a) node [right,midway,scale=\nodescale] {$3$}  {};
\draw[fleche] (4a)--(1a) {};
\draw[fleche] (5a)--(1a) {};
\draw[fleche] (5a) to [bend left=100,looseness=2] (0a) {};
\end{tikzpicture}};
                   \node (20) at (4.5*\xdim,-2.1*\ydim){\begin{tikzpicture}[xscale=\xscale,yscale=\yscale]
\tikzstyle{noeud}=[draw,circle,fill=blue,scale=\nodescale*1]
\node[noeud] (0a) at (0,0) {};
\tikzstyle{noeud}=[draw,circle,fill=magenta,scale=\nodescale*1]
\node[noeud] (1a) at (-1,1) {};
\tikzstyle{noeud}=[draw,circle,fill=cyan,scale=\nodescale*1]
\node[noeud] (2a) at (0,1) {};
\tikzstyle{noeud}=[draw,circle,fill=brown,scale=\nodescale*1]
\node[noeud] (3a) at (1,1) {};
\tikzstyle{noeud}=[draw,circle,fill=orange,scale=\nodescale*1]
\node[noeud] (4a) at (-1,2) {};
\tikzstyle{noeud}=[draw,circle,fill=yellow,scale=\nodescale*1]
\node[noeud] (5a) at (0,2) {};
\tikzstyle{fleche}=[->-,>=latex,black]
\draw[fleche] (1a)--(0a) {};
\draw[fleche] (2a)--(0a) {}node [right,midway,scale=\nodescale] {$2$};
\draw[fleche] (3a)--(0a) node [right,midway,scale=\nodescale] {$3$}  {};
\draw[fleche] (4a)--(1a) {};
\draw[fleche] (5a)--(2a) {};
\end{tikzpicture}};
                        \node[opacity=\opacity] (19) at (5.1*\xdim,-2.1*\ydim){\begin{tikzpicture}[xscale=\xscale,yscale=\yscale]
\tikzstyle{noeud}=[draw,circle,fill=blue,scale=\nodescale*1]
\node[noeud] (0a) at (0,0) {};
\tikzstyle{noeud}=[draw,circle,fill=magenta,scale=\nodescale*1]
\node[noeud] (1a) at (-1,1) {};
\tikzstyle{noeud}=[draw,circle,fill=cyan,scale=\nodescale*1]
\node[noeud] (2a) at (0,1) {};
\tikzstyle{noeud}=[draw,circle,fill=brown,scale=\nodescale*1]
\node[noeud] (3a) at (1,1) {};
\tikzstyle{noeud}=[draw,circle,fill=orange,scale=\nodescale*1]
\node[noeud] (4a) at (-1,2) {};
\tikzstyle{noeud}=[draw,circle,fill=yellow,scale=\nodescale*1]
\node[noeud] (5a) at (1,2) {};
\tikzstyle{fleche}=[->-,>=latex,black]
\draw[fleche] (1a)--(0a) {};
\draw[fleche] (2a)--(0a) {}node [right,midway,scale=\nodescale] {$2$};
\draw[fleche] (3a)--(0a) node [right,midway,scale=\nodescale] {$3$}  {};
\draw[fleche] (4a)--(1a) {};
\draw[fleche] (5a)--(3a) {};
\end{tikzpicture}};

\node[red] (d) at (4.5*\xdim,-2.6*\ydim){\elong};
\node[red] (e) at (2.7*\xdim,-2.1*\ydim){\widen};
\node[red] (f) at (3.3*\xdim,-2.2*\ydim){\branch};

\node[opacity=\opacity] (21) at (4.5*\xdim,-3.3*\ydim){\begin{tikzpicture}[xscale=\xscale,yscale=\yscale]
\tikzstyle{noeud}=[draw,circle,fill=blue,scale=\nodescale*1]
\node[noeud] (0a) at (0,0) {};
\tikzstyle{noeud}=[draw,circle,fill=magenta,scale=\nodescale*1]
\node[noeud] (1a) at (-1,1) {};
\tikzstyle{noeud}=[draw,circle,fill=cyan,scale=\nodescale*1]
\node[noeud] (2a) at (0,1) {};
\tikzstyle{noeud}=[draw,circle,fill=brown,scale=\nodescale*1]
\node[noeud] (3a) at (1,1) {};
\tikzstyle{noeud}=[draw,circle,fill=orange,scale=\nodescale*1]
\node[noeud] (4a) at (-1,2) {};
\tikzstyle{noeud}=[draw,circle,fill=yellow,scale=\nodescale*1]
\node[noeud] (5a) at (0,2) {};
\tikzstyle{noeud}=[draw,circle,fill=red,scale=\nodescale*1]
\node[noeud] (6a) at (-1,3) {};
\tikzstyle{fleche}=[->-,>=latex,black]
\draw[fleche] (1a)--(0a) {};
\draw[fleche] (2a)--(0a) {}node [right,midway,scale=\nodescale] {$2$};
\draw[fleche] (3a)--(0a) node [right,midway,scale=\nodescale] {$3$}  {};
\draw[fleche] (4a)--(1a) {};
\draw[fleche] (5a)--(2a) {};
\draw[fleche] (6a)--(4a) {};
\end{tikzpicture}};
\node[opacity=\opacity] (22) at (5.1*\xdim,-2.8*\ydim){\begin{tikzpicture}[xscale=\xscale,yscale=\yscale]
\tikzstyle{noeud}=[draw,circle,fill=blue,scale=\nodescale*1]
\node[noeud] (0a) at (0,0) {};
\tikzstyle{noeud}=[draw,circle,fill=magenta,scale=\nodescale*1]
\node[noeud] (1a) at (-1,1) {};
\tikzstyle{noeud}=[draw,circle,fill=cyan,scale=\nodescale*1]
\node[noeud] (2a) at (0,1) {};
\tikzstyle{noeud}=[draw,circle,fill=brown,scale=\nodescale*1]
\node[noeud] (3a) at (1,1) {};
\tikzstyle{noeud}=[draw,circle,fill=orange,scale=\nodescale*1]
\node[noeud] (4a) at (-1,2) {};
\tikzstyle{noeud}=[draw,circle,fill=yellow,scale=\nodescale*1]
\node[noeud] (5a) at (0,2) {};
\tikzstyle{noeud}=[draw,circle,fill=red,scale=\nodescale*1]
\node[noeud] (6a) at (0,3) {};
\tikzstyle{fleche}=[->-,>=latex,black]
\draw[fleche] (1a)--(0a) {};
\draw[fleche] (2a)--(0a) {}node [right,midway,scale=\nodescale] {$2$};
\draw[fleche] (3a)--(0a) node [right,midway,scale=\nodescale] {$3$}  {};
\draw[fleche] (4a)--(1a) {};
\draw[fleche] (5a)--(2a) {};
\draw[fleche] (6a)--(5a) {};
\end{tikzpicture}};
\node[opacity=\opacity] (23) at (3.9*\xdim,-2.8*\ydim){\begin{tikzpicture}[xscale=\xscale,yscale=\yscale]
\tikzstyle{noeud}=[draw,circle,fill=blue,scale=\nodescale*1]
\node[noeud] (0a) at (0,0) {};
\tikzstyle{noeud}=[draw,circle,fill=magenta,scale=\nodescale*1]
\node[noeud] (1a) at (-1,1) {};
\tikzstyle{noeud}=[draw,circle,fill=cyan,scale=\nodescale*1]
\node[noeud] (2a) at (0,1) {};
\tikzstyle{noeud}=[draw,circle,fill=brown,scale=\nodescale*1]
\node[noeud] (3a) at (1,1) {};
\tikzstyle{noeud}=[draw,circle,fill=orange,scale=\nodescale*1]
\node[noeud] (4a) at (-1,2) {};
\tikzstyle{noeud}=[draw,circle,fill=yellow,scale=\nodescale*1]
\node[noeud] (5a) at (0,2) {};
\tikzstyle{fleche}=[->-,>=latex,black]
\draw[fleche] (1a)--(0a) {};
\draw[fleche] (2a)--(0a) {}node [right,midway,scale=\nodescale] {$2$};
\draw[fleche] (3a)--(0a) node [right,midway,scale=\nodescale] {$3$}  {};
\draw[fleche] (4a)--(1a) {};
\draw[fleche] (5a)--(2a) node[right,midway,scale=\nodescale]{$2$}{};
\end{tikzpicture}};
\node (24) at (2.5*\xdim,-2.8*\ydim){\begin{tikzpicture}[xscale=\xscale,yscale=\yscale]
\tikzstyle{noeud}=[draw,circle,fill=blue,scale=\nodescale*1]
\node[noeud] (0a) at (0,0) {};
\tikzstyle{noeud}=[draw,circle,fill=magenta,scale=\nodescale*1]
\node[noeud] (1a) at (-1,1) {};
\tikzstyle{noeud}=[draw,circle,fill=cyan,scale=\nodescale*1]
\node[noeud] (2a) at (0,1) {};
\tikzstyle{noeud}=[draw,circle,fill=brown,scale=\nodescale*1]
\node[noeud] (3a) at (1,1) {};
\tikzstyle{noeud}=[draw,circle,fill=orange,scale=\nodescale*1]
\node[noeud] (4a) at (-1,2) {};
\tikzstyle{noeud}=[draw,circle,fill=yellow,scale=\nodescale*1]
\node[noeud] (5a) at (0,2) {};
\tikzstyle{fleche}=[->-,>=latex,black]
\draw[fleche] (1a)--(0a) {};
\draw[fleche] (2a)--(0a) {}node [right,midway,scale=\nodescale] {$2$};
\draw[fleche] (3a)--(0a) node [right,midway,scale=\nodescale] {$3$}  {};
\draw[fleche] (4a)--(1a) {};
\draw[fleche] (5a)--(2a) {};
\draw[fleche] (5a)--(1a) {};
\end{tikzpicture}};
\node[opacity=\opacity] (37) at (3.3*\xdim,-2.8*\ydim){\begin{tikzpicture}[xscale=\xscale,yscale=\yscale]
\tikzstyle{noeud}=[draw,circle,fill=blue,scale=\nodescale*1]
\node[noeud] (0a) at (0,0) {};
\tikzstyle{noeud}=[draw,circle,fill=magenta,scale=\nodescale*1]
\node[noeud] (1a) at (-1,1) {};
\tikzstyle{noeud}=[draw,circle,fill=cyan,scale=\nodescale*1]
\node[noeud] (2a) at (0,1) {};
\tikzstyle{noeud}=[draw,circle,fill=brown,scale=\nodescale*1]
\node[noeud] (3a) at (1,1) {};
\tikzstyle{noeud}=[draw,circle,fill=orange,scale=\nodescale*1]
\node[noeud] (4a) at (-1,2) {};
\tikzstyle{noeud}=[draw,circle,fill=yellow,scale=\nodescale*1]
\node[noeud] (5a) at (0,2) {};
\tikzstyle{fleche}=[->-,>=latex,black]
\draw[fleche] (1a)--(0a) {};
\draw[fleche] (2a)--(0a) {}node [right,midway,scale=\nodescale] {$2$};
\draw[fleche] (3a)--(0a) node [right,midway,scale=\nodescale] {$3$}  {};
\draw[fleche] (4a)--(1a) {};
\draw[fleche] (5a)--(2a) {};
\draw[fleche] (5a) to [bend left=100,looseness=2] (0a) {};
\end{tikzpicture}};
\node[opacity=\opacity] (25) at (3.3*\xdim,-1.5*\ydim){\begin{tikzpicture}[xscale=\xscale,yscale=\yscale]
\tikzstyle{noeud}=[draw,circle,fill=blue,scale=\nodescale*1]
\node[noeud] (0a) at (0,0) {};
\tikzstyle{noeud}=[draw,circle,fill=magenta,scale=\nodescale*1]
\node[noeud] (1a) at (-1,1) {};
\tikzstyle{noeud}=[draw,circle,fill=cyan,scale=\nodescale*1]
\node[noeud] (2a) at (0,1) {};
\tikzstyle{noeud}=[draw,circle,fill=brown,scale=\nodescale*1]
\node[noeud] (3a) at (1,1) {};
\tikzstyle{noeud}=[draw,circle,fill=orange,scale=\nodescale*1]
\node[noeud] (4a) at (-1,2) {};
\tikzstyle{noeud}=[draw,circle,fill=yellow,scale=\nodescale*1]
\node[noeud] (5a) at (0,2) {};
\tikzstyle{noeud}=[draw,circle,fill=red,scale=\nodescale*1]
\node[noeud] (6a) at (1,2) {};
\tikzstyle{fleche}=[->-,>=latex,black]
\draw[fleche] (1a)--(0a) {};
\draw[fleche] (2a)--(0a) {}node [right,midway,scale=\nodescale] {$2$};
\draw[fleche] (3a)--(0a) node [right,midway,scale=\nodescale] {$3$}  {};
\draw[fleche] (4a)--(1a) {};
\draw[fleche] (5a)--(2a) {};
\draw[fleche] (6a)--(2a) node[right,midway,scale=\nodescale]{$2$}{};
\end{tikzpicture}};
\node[opacity=\opacity] (26) at (2.7*\xdim,-1.5*\ydim){\begin{tikzpicture}[xscale=\xscale,yscale=\yscale]
\tikzstyle{noeud}=[draw,circle,fill=blue,scale=\nodescale*1]
\node[noeud] (0a) at (0,0) {};
\tikzstyle{noeud}=[draw,circle,fill=magenta,scale=\nodescale*1]
\node[noeud] (1a) at (-1,1) {};
\tikzstyle{noeud}=[draw,circle,fill=cyan,scale=\nodescale*1]
\node[noeud] (2a) at (0,1) {};
\tikzstyle{noeud}=[draw,circle,fill=brown,scale=\nodescale*1]
\node[noeud] (3a) at (1,1) {};
\tikzstyle{noeud}=[draw,circle,fill=orange,scale=\nodescale*1]
\node[noeud] (4a) at (-1,2) {};
\tikzstyle{noeud}=[draw,circle,fill=yellow,scale=\nodescale*1]
\node[noeud] (5a) at (0,2) {};
\tikzstyle{noeud}=[draw,circle,fill=red,scale=\nodescale*1]
\node[noeud] (6a) at (1,2) {};
\tikzstyle{fleche}=[->-,>=latex,black]
\draw[fleche] (1a)--(0a) {};
\draw[fleche] (2a)--(0a) {}node [right,midway,scale=\nodescale] {$2$};
\draw[fleche] (3a)--(0a) node [right,midway,scale=\nodescale] {$3$}  {};
\draw[fleche] (4a)--(1a) {};
\draw[fleche] (5a)--(2a) {};
\draw[fleche] (5a)--(1a) {}node [left,midway,scale=\nodescale] {$2$};
\draw[fleche] (6a)--(3a) {};

\end{tikzpicture}};
\node[opacity=\opacity] (27) at (2.1*\xdim,-1.5*\ydim){\begin{tikzpicture}[xscale=\xscale,yscale=\yscale]
\tikzstyle{noeud}=[draw,circle,fill=blue,scale=\nodescale*1]
\node[noeud] (0a) at (0,0) {};
\tikzstyle{noeud}=[draw,circle,fill=magenta,scale=\nodescale*1]
\node[noeud] (1a) at (-1,1) {};
\tikzstyle{noeud}=[draw,circle,fill=cyan,scale=\nodescale*1]
\node[noeud] (2a) at (0,1) {};
\tikzstyle{noeud}=[draw,circle,fill=brown,scale=\nodescale*1]
\node[noeud] (3a) at (1,1) {};
\tikzstyle{noeud}=[draw,circle,fill=orange,scale=\nodescale*1]
\node[noeud] (4a) at (-1,2) {};
\tikzstyle{noeud}=[draw,circle,fill=yellow,scale=\nodescale*1]
\node[noeud] (5a) at (0,2) {};
\tikzstyle{noeud}=[draw,circle,fill=red,scale=\nodescale*1]
\node[noeud] (6a) at (1,2) {};
\tikzstyle{fleche}=[->-,>=latex,black]
\draw[fleche] (1a)--(0a) {};
\draw[fleche] (2a)--(0a) {}node [right,midway,scale=\nodescale] {$2$};
\draw[fleche] (3a)--(0a) node [right,midway,scale=\nodescale] {$3$}  {};
\draw[fleche] (4a)--(1a) {};
\draw[fleche] (5a)--(2a) {};
\draw[fleche] (5a)--(1a) {}node [left,midway,scale=\nodescale] {$2$};
\draw[fleche] (6a)--(2a) {};
\draw[fleche] (6a)--(1a) {};
\end{tikzpicture}};
\node[opacity=\opacity] (28) at (2.1*\xdim,-2.1*\ydim){\begin{tikzpicture}[xscale=\xscale,yscale=\yscale]
\tikzstyle{noeud}=[draw,circle,fill=blue,scale=\nodescale*1]
\node[noeud] (0a) at (0,0) {};
\tikzstyle{noeud}=[draw,circle,fill=magenta,scale=\nodescale*1]
\node[noeud] (1a) at (-1,1) {};
\tikzstyle{noeud}=[draw,circle,fill=cyan,scale=\nodescale*1]
\node[noeud] (2a) at (0,1) {};
\tikzstyle{noeud}=[draw,circle,fill=brown,scale=\nodescale*1]
\node[noeud] (3a) at (1,1) {};
\tikzstyle{noeud}=[draw,circle,fill=orange,scale=\nodescale*1]
\node[noeud] (4a) at (-1,2) {};
\tikzstyle{noeud}=[draw,circle,fill=yellow,scale=\nodescale*1]
\node[noeud] (5a) at (0,2) {};
\tikzstyle{noeud}=[draw,circle,fill=red,scale=\nodescale*1]
\node[noeud] (6a) at (1,2) {};
\tikzstyle{fleche}=[->-,>=latex,black]
\draw[fleche] (1a)--(0a) {};
\draw[fleche] (2a)--(0a) {}node [right,midway,scale=\nodescale] {$2$};
\draw[fleche] (3a)--(0a) node [right,midway,scale=\nodescale] {$3$}  {};
\draw[fleche] (4a)--(1a) {};
\draw[fleche] (5a)--(2a) {};
\draw[fleche] (5a)--(1a) {}node [left,midway,scale=\nodescale] {$2$};
\draw[fleche] (6a)--(2a) {};
\draw[fleche] (6a) to [bend left=100](0a){};
\end{tikzpicture}};

\node[red] (g) at (3*\xdim,-3.2*\ydim){\elong};
\node[red] (h) at (0.9*\xdim,-2.8*\ydim){\widen};
\node[red] (i) at (1.8*\xdim,-3.2*\ydim){\branch};

\node[opacity=\opacity] (29) at (2.7*\xdim,-3.8*\ydim){\begin{tikzpicture}[xscale=\xscale,yscale=\yscale]
\tikzstyle{noeud}=[draw,circle,fill=blue,scale=\nodescale*1]
\node[noeud] (0a) at (0,0) {};
\tikzstyle{noeud}=[draw,circle,fill=magenta,scale=\nodescale*1]
\node[noeud] (1a) at (-1,1) {};
\tikzstyle{noeud}=[draw,circle,fill=cyan,scale=\nodescale*1]
\node[noeud] (2a) at (0,1) {};
\tikzstyle{noeud}=[draw,circle,fill=brown,scale=\nodescale*1]
\node[noeud] (3a) at (1,1) {};
\tikzstyle{noeud}=[draw,circle,fill=orange,scale=\nodescale*1]
\node[noeud] (4a) at (-1,2) {};
\tikzstyle{noeud}=[draw,circle,fill=yellow,scale=\nodescale*1]
\node[noeud] (5a) at (0,2) {};
\tikzstyle{noeud}=[draw,circle,fill=red,scale=\nodescale*1]
\node[noeud] (6a) at (0,3) {};
\tikzstyle{fleche}=[->-,>=latex,black]
\draw[fleche] (1a)--(0a) {};
\draw[fleche] (2a)--(0a) {}node [right,midway,scale=\nodescale] {$2$};
\draw[fleche] (3a)--(0a) node [right,midway,scale=\nodescale] {$3$}  {};
\draw[fleche] (4a)--(1a) {};
\draw[fleche] (5a)--(2a) {};
\draw[fleche] (5a)--(1a) {};
\draw[fleche] (6a)--(5a) {};
\end{tikzpicture}};
\node[opacity=\opacity] (30) at (3.3*\xdim,-3.8*\ydim){\begin{tikzpicture}[xscale=\xscale,yscale=\yscale]
\tikzstyle{noeud}=[draw,circle,fill=blue,scale=\nodescale*1]
\node[noeud] (0a) at (0,0) {};
\tikzstyle{noeud}=[draw,circle,fill=magenta,scale=\nodescale*1]
\node[noeud] (1a) at (-1,1) {};
\tikzstyle{noeud}=[draw,circle,fill=cyan,scale=\nodescale*1]
\node[noeud] (2a) at (0,1) {};
\tikzstyle{noeud}=[draw,circle,fill=brown,scale=\nodescale*1]
\node[noeud] (3a) at (1,1) {};
\tikzstyle{noeud}=[draw,circle,fill=orange,scale=\nodescale*1]
\node[noeud] (4a) at (-1,2) {};
\tikzstyle{noeud}=[draw,circle,fill=yellow,scale=\nodescale*1]
\node[noeud] (5a) at (0,2) {};
\tikzstyle{noeud}=[draw,circle,fill=red,scale=\nodescale*1]
\node[noeud] (6a) at (-1,3) {};
\tikzstyle{fleche}=[->-,>=latex,black]
\draw[fleche] (1a)--(0a) {};
\draw[fleche] (2a)--(0a) {}node [right,midway,scale=\nodescale] {$2$};
\draw[fleche] (3a)--(0a) node [right,midway,scale=\nodescale] {$3$}  {};
\draw[fleche] (4a)--(1a) {};
\draw[fleche] (5a)--(2a) {};
\draw[fleche] (5a)--(1a) {};
\draw[fleche] (6a)--(4a) {};
\end{tikzpicture}};
\node (31) at (1.5*\xdim,-3.8*\ydim){\begin{tikzpicture}[xscale=\xscale,yscale=\yscale]
\tikzstyle{noeud}=[draw,circle,fill=blue,scale=\nodescale*1]
\node[noeud] (0a) at (0,0) {};
\tikzstyle{noeud}=[draw,circle,fill=magenta,scale=\nodescale*1]
\node[noeud] (1a) at (-1,1) {};
\tikzstyle{noeud}=[draw,circle,fill=cyan,scale=\nodescale*1]
\node[noeud] (2a) at (0,1) {};
\tikzstyle{noeud}=[draw,circle,fill=brown,scale=\nodescale*1]
\node[noeud] (3a) at (1,1) {};
\tikzstyle{noeud}=[draw,circle,fill=orange,scale=\nodescale*1]
\node[noeud] (4a) at (-1,2) {};
\tikzstyle{noeud}=[draw,circle,fill=yellow,scale=\nodescale*1]
\node[noeud] (5a) at (0,2) {};
\tikzstyle{fleche}=[->-,>=latex,black]
\draw[fleche] (1a)--(0a) {};
\draw[fleche] (2a)--(0a) {}node [right,midway,scale=\nodescale] {$2$};
\draw[fleche] (3a)--(0a) node [right,midway,scale=\nodescale] {$3$}  {};
\draw[fleche] (4a)--(1a) {};
\draw[fleche] (5a)--(2a) {};
\draw[fleche] (5a)--(1a) {}node [left,midway,scale=\nodescale] {$2$};
\end{tikzpicture}};
\node[opacity=\opacity] (32) at (2.1*\xdim,-3.8*\ydim){\begin{tikzpicture}[xscale=\xscale,yscale=\yscale]
\tikzstyle{noeud}=[draw,circle,fill=blue,scale=\nodescale*1]
\node[noeud] (0a) at (0,0) {};
\tikzstyle{noeud}=[draw,circle,fill=magenta,scale=\nodescale*1]
\node[noeud] (1a) at (-1,1) {};
\tikzstyle{noeud}=[draw,circle,fill=cyan,scale=\nodescale*1]
\node[noeud] (2a) at (0,1) {};
\tikzstyle{noeud}=[draw,circle,fill=brown,scale=\nodescale*1]
\node[noeud] (3a) at (1,1) {};
\tikzstyle{noeud}=[draw,circle,fill=orange,scale=\nodescale*1]
\node[noeud] (4a) at (-1,2) {};
\tikzstyle{noeud}=[draw,circle,fill=yellow,scale=\nodescale*1]
\node[noeud] (5a) at (0,2) {};

\tikzstyle{fleche}=[->-,>=latex,black]
\draw[fleche] (1a)--(0a) {};
\draw[fleche] (2a)--(0a) {}node [right,midway,scale=\nodescale] {$2$};
\draw[fleche] (3a)--(0a) node [right,midway,scale=\nodescale] {$3$}  {};
\draw[fleche] (4a)--(1a) {};
\draw[fleche] (5a)--(2a) {};
\draw[fleche] (5a)--(1a) {};
\draw[fleche] (5a) to [bend left=100,looseness=2](0a){};
\end{tikzpicture}};

\node[opacity=\opacity] (33) at (0.3*\xdim,-3.4*\ydim){\begin{tikzpicture}[xscale=\xscale,yscale=\yscale]
\tikzstyle{noeud}=[draw,circle,fill=blue,scale=\nodescale*1]
\node[noeud] (0a) at (0,0) {};
\tikzstyle{noeud}=[draw,circle,fill=magenta,scale=\nodescale*1]
\node[noeud] (1a) at (-1,1) {};
\tikzstyle{noeud}=[draw,circle,fill=cyan,scale=\nodescale*1]
\node[noeud] (2a) at (0,1) {};
\tikzstyle{noeud}=[draw,circle,fill=brown,scale=\nodescale*1]
\node[noeud] (3a) at (1,1) {};
\tikzstyle{noeud}=[draw,circle,fill=orange,scale=\nodescale*1]
\node[noeud] (4a) at (-1,2) {};
\tikzstyle{noeud}=[draw,circle,fill=yellow,scale=\nodescale*1]
\node[noeud] (5a) at (0,2) {};
\tikzstyle{noeud}=[draw,circle,fill=red,scale=\nodescale*1]
\node[noeud] (6a) at (1,2) {};
\tikzstyle{fleche}=[->-,>=latex,black]
\draw[fleche] (1a)--(0a) {};
\draw[fleche] (2a)--(0a) {}node [right,midway,scale=\nodescale] {$2$};
\draw[fleche] (3a)--(0a) node [right,midway,scale=\nodescale] {$3$}  {};
\draw[fleche] (4a)--(1a) {};
\draw[fleche] (5a)--(2a) {};
\draw[fleche] (5a)--(1a) {};
\draw[fleche] (6a)--(3a) {};

\end{tikzpicture}};
\node[opacity=\opacity] (34) at (0.9*\xdim,-2.2*\ydim){\begin{tikzpicture}[xscale=\xscale,yscale=\yscale]
\tikzstyle{noeud}=[draw,circle,fill=blue,scale=\nodescale*1]
\node[noeud] (0a) at (0,0) {};
\tikzstyle{noeud}=[draw,circle,fill=magenta,scale=\nodescale*1]
\node[noeud] (1a) at (-1,1) {};
\tikzstyle{noeud}=[draw,circle,fill=cyan,scale=\nodescale*1]
\node[noeud] (2a) at (0,1) {};
\tikzstyle{noeud}=[draw,circle,fill=brown,scale=\nodescale*1]
\node[noeud] (3a) at (1,1) {};
\tikzstyle{noeud}=[draw,circle,fill=orange,scale=\nodescale*1]
\node[noeud] (4a) at (-1,2) {};
\tikzstyle{noeud}=[draw,circle,fill=yellow,scale=\nodescale*1]
\node[noeud] (5a) at (0,2) {};
\tikzstyle{noeud}=[draw,circle,fill=red,scale=\nodescale*1]
\node[noeud] (6a) at (1,2) {};
\tikzstyle{fleche}=[->-,>=latex,black]
\draw[fleche] (1a)--(0a) {};
\draw[fleche] (2a)--(0a) {}node [right,midway,scale=\nodescale] {$2$};
\draw[fleche] (3a)--(0a) node [right,midway,scale=\nodescale] {$3$}  {};
\draw[fleche] (4a)--(1a) {};
\draw[fleche] (5a)--(2a) {};
\draw[fleche] (5a)--(1a) {};
\draw[fleche] (6a)--(2a) {} node[left,midway,scale=\nodescale]{$2$};

\end{tikzpicture}};
\node[opacity=\opacity] (35) at (0.3*\xdim,-2.2*\ydim){\begin{tikzpicture}[xscale=\xscale,yscale=\yscale]
\tikzstyle{noeud}=[draw,circle,fill=blue,scale=\nodescale*1]
\node[noeud] (0a) at (0,0) {};
\tikzstyle{noeud}=[draw,circle,fill=magenta,scale=\nodescale*1]
\node[noeud] (1a) at (-1,1) {};
\tikzstyle{noeud}=[draw,circle,fill=cyan,scale=\nodescale*1]
\node[noeud] (2a) at (0,1) {};
\tikzstyle{noeud}=[draw,circle,fill=brown,scale=\nodescale*1]
\node[noeud] (3a) at (1,1) {};
\tikzstyle{noeud}=[draw,circle,fill=orange,scale=\nodescale*1]
\node[noeud] (4a) at (-1,2) {};
\tikzstyle{noeud}=[draw,circle,fill=yellow,scale=\nodescale*1]
\node[noeud] (5a) at (0,2) {};
\tikzstyle{noeud}=[draw,circle,fill=red,scale=\nodescale*1]
\node[noeud] (6a) at (1,2) {};
\tikzstyle{fleche}=[->-,>=latex,black]
\draw[fleche] (1a)--(0a) {};
\draw[fleche] (2a)--(0a) {}node [right,midway,scale=\nodescale] {$2$};
\draw[fleche] (3a)--(0a) node [right,midway,scale=\nodescale] {$3$}  {};
\draw[fleche] (4a)--(1a) {};
\draw[fleche] (5a)--(2a) {};
\draw[fleche] (5a)--(1a) {};
\draw[fleche] (6a)--(2a) {};
\draw[fleche] (6a)--(1a) {}node[left,midway,scale=\nodescale]{$2$};

\end{tikzpicture}};
\node[opacity=\opacity] (36) at (0.3*\xdim,-2.8*\ydim){\begin{tikzpicture}[xscale=\xscale,yscale=\yscale]
\tikzstyle{noeud}=[draw,circle,fill=blue,scale=\nodescale*1]
\node[noeud] (0a) at (0,0) {};
\tikzstyle{noeud}=[draw,circle,fill=magenta,scale=\nodescale*1]
\node[noeud] (1a) at (-1,1) {};
\tikzstyle{noeud}=[draw,circle,fill=cyan,scale=\nodescale*1]
\node[noeud] (2a) at (0,1) {};
\tikzstyle{noeud}=[draw,circle,fill=brown,scale=\nodescale*1]
\node[noeud] (3a) at (1,1) {};
\tikzstyle{noeud}=[draw,circle,fill=orange,scale=\nodescale*1]
\node[noeud] (4a) at (-1,2) {};
\tikzstyle{noeud}=[draw,circle,fill=yellow,scale=\nodescale*1]
\node[noeud] (5a) at (0,2) {};
\tikzstyle{noeud}=[draw,circle,fill=red,scale=\nodescale*1]
\node[noeud] (6a) at (1,2) {};
\tikzstyle{fleche}=[->-,>=latex,black]
\draw[fleche] (1a)--(0a) {};
\draw[fleche] (2a)--(0a) {}node [right,midway,scale=\nodescale] {$2$};
\draw[fleche] (3a)--(0a) node [right,midway,scale=\nodescale] {$3$}  {};
\draw[fleche] (4a)--(1a) {};
\draw[fleche] (5a)--(2a) {};
\draw[fleche] (5a)--(1a) {}node [left,midway,scale=\nodescale] {$2$};
\draw[fleche] (6a)--(2a) {};
\draw[fleche] (6a)--(1a) {};
\draw[fleche] (6a)--(0a) {};
\end{tikzpicture}};

          \draw[->,color=red,ultra thick] (1)--(2) {} node[midway,above] {\elong};
   \draw[->,color=red] (2)--(3) {} node[midway,right] {\elong};
       \draw[->,color=red] (2)--(4) {} node[midway,right] {\branch};
       \draw[->,color=red,ultra thick] (2)--(5) {} node[midway,above] {\widen};
          \draw[-,color=red] (5)--(a) {} node[midway,right] {};
              \draw[->,color=red] (a)--(6) {} node[midway,right] {};
                  \draw[->,color=red] (a)--(7) {} node[midway,right] {};
         \draw[->,color=red] (5)--(8) {} node[midway,right] {\branch};
          \draw[->,color=red,ultra thick] (5)--(9) {} node[midway,above] {\widen};
             \draw[->,color=red] (9)--(11) {} node[midway,right] {\branch};
       \draw[->,color=red] (9)--(10) {} node[midway,right] {\widen};
        \draw[-,color=red,ultra thick] (9)--(b) {} node[midway,right] {};
          \draw[->,color=red] (b)--(12) {} node[midway,right] {};
           \draw[->,color=red] (b)--(13) {} node[midway,right] {};
            \draw[->,color=red,ultra thick] (b)--(14) {} node[midway,right] {};
                         \draw[->,color=red] (14)--(15) {} node[midway,right] {\elong};
       \draw[-,color=red] (14)--(j) {} node[midway,above] {};
        \draw[->,color=red] (j)--(16) {} node[midway,right] {};
         \draw[->,color=red] (j)--(38) {} node[midway,right] {};
           \draw[-,color=red,ultra thick] (14)--(c) {} node[midway,right] {};
       \draw[->,color=red] (c)--(17) {} node[midway,right] {};
          \draw[->,color=red] (c)--(18) {} node[midway,right] {};
             \draw[->,color=red] (c)--(19) {} node[midway,right] {};
                \draw[->,color=red,ultra thick] (c)--(20) {} node[midway,right] {};
                           \draw[-,color=red] (20)--(d) {} node[midway,right] {};
                                      \draw[-,color=red] (20)--(e) {} node[midway,right] {};
                                                 \draw[-,color=red,ultra thick] (20)--(f) {} node[midway,right] {};
\draw[->,color=red] (d)--(21) {} node[midway,right] {};
\draw[->,color=red] (d)--(22) {} node[midway,right] {};
\draw[->,color=red] (f)--(23) {} node[midway,right] {};
\draw[->,color=red,ultra thick] (f)--(24) {} node[midway,right] {};
\draw[->,color=red] (f)--(37) {} node[midway,right] {};
\draw[->,color=red] (e)--(25) {} node[midway,right] {};
\draw[->,color=red] (e)--(26) {} node[midway,right] {};
\draw[->,color=red] (e)--(27) {} node[midway,right] {};
\draw[->,color=red] (e)--(28) {} node[midway,right] {};

\draw[-,color=red] (24)--(g) {} node[midway,right] {};
\draw[-,color=red] (24)--(h) {} node[midway,right] {};
\draw[-,color=red,ultra thick] (24)--(i) {} node[midway,right] {};

\draw[->,color=red] (g)--(29) {} node[midway,right] {};
\draw[->,color=red] (g)--(30) {} node[midway,right] {};
\draw[->,color=red,ultra thick] (i)--(31) {} node[midway,right] {};
\draw[->,color=red] (i)--(32) {} node[midway,right] {};

\draw[->,color=red] (h)--(33) {} node[midway,right] {};
\draw[->,color=red] (h)--(34) {} node[midway,right] {};
\draw[->,color=red] (h)--(35) {} node[midway,right] {};
\draw[->,color=red] (h)--(36) {} node[midway,right] {};

\end{tikzpicture}

\def\nodescale{1}

\caption{The path (in bold) in the FDAGs enumeration tree leading to the FDAG of Figure~\ref{fig:ordering:canonical}. The unexplored branches are only displayed by their root, which are shown partially transparent. The order of insertion of the vertices of each FDAG is always the same, and follows the color code (in the order of insertion): \node[blue], \node[magenta], \node[cyan], \node[brown], \node[orange], \node[yellow] and \node[red]. With respect to the canonical ordering, they are numbered $0$ to $6$ in the same order.}
\label{fig:enumtree}
\end{figure}

\section{Growth of the tree}\label{sec:analysis}

In this section, we analyse the enumeration tree defined in Section~\ref{sec:enumeration}. In Subsection~\ref{ss:asymptotic}, we exhibit a bijection -- Theorem~\ref{th:bijdagmatrix} -- between FDAGs and a class of combinatorial objects from the literature, allowing us to obtain an asymptotic expansion of the growth of the tree. In Subsection~\ref{ss:branching}, we show that any FDAG has a linear number of children in that tree in Theorem~\ref{th:successors}, and that the time complexity to construct those children is quadratic -- see Proposition~\ref{prop:time}. Finally, Theorem~\ref{th:polynomial_delay} states that our algorithm runs with polynomial delay \cite{johnson1988generating}.

\subsection{Asymptotic growth}\label{ss:asymptotic}
\noindent
In this subsection, we show that FDAGs are in bijection with a set of particular matrices, whose combinatorial properties are known and give us access to an asymptotic expansion of the enumeration tree growth.

\medskip
\noindent
\begin{minipage}{0.45\textwidth}
Let us denote $E_k$ the set of all FDAGs that are accessible from $D_0$ in exactly $k$ steps in the enumeration tree -- with $E_0 = \lbrace D_0 \rbrace$; then Table~\ref{tab:enumtree} depicts the values of $\#E_k$ for the first nine values of $k$\footnotemark.
\end{minipage}\hfill
\begin{minipage}{0.5\textwidth}
\centering
\small
\setlength\tabcolsep{4pt}
\begin{tabular}{c|ccccccccc}
$k$ & 0 & 1 & 2&3&4&5&6&7&8\\
$\#E_k$&1&1&3&12&61&380&2,815&24,213& 237,348
\end{tabular}
\captionof{table}{Number of FDAGs accessible from $D_0$ in $k$ steps in the enumeration tree.}\label{tab:enumtree}
\end{minipage}
\footnotetext{These numbers were obtained numerically (cf. ``Implementation'' at the end of the article).}

\medskip
\noindent
Actually, the terms of Table~\ref{tab:enumtree} coincide with the first terms of OEIS sequence A158691\footnote{OEIS Foundation Inc. (2021), The On-Line Encyclopedia of Integer Sequences, \url{http://oeis.org/A158691}.}, which counts the number of \emph{row-Fishburn matrices}, that are upper-triangular matrices with at least one nonzero entry in each row. The \emph{size} of such a matrix is equal to the sum of its entries.

\begin{theorem}\label{th:bijdagmatrix}
There exists a bijection $\Phi$ between the set of FDAGs and the set of row-Fishburn matrices, such that if $D$ is a FDAG and $M=\Phi(D)$, then
$$D \in E_k \iff \text{size}(M)=k.$$
\end{theorem}
\begin{proof}
The proof lies in Appendix~\ref{sec:dagmatrix}.
\end{proof}

\medskip
\noindent
This connection is to our advantage since Fishburn matrices (in general) are combinatorial objects widely explored in the literature as they are in bijection with many others -- see \cite[Section 2]{hwang2019asymptotics} for a general overview. Notably, the asymptotic expansion of the number of row-Fishburn matrices has been conjectured first by Jel\'inek \cite{jelinek2012counting} and then proved by Bringmann et al. \cite{bringmann2014asymptotics}.

\begin{proposition}[Jel\'inek, Bringmann et al.]\label{prop:asymptotic}
As $k\to \infty$,
$$\#E_k = k!\left(\frac{12}{\pi^2}\right)^k \left(\beta+O\left(\frac{1}{k}\right)\right)$$
with $\beta = \frac{6\sqrt{2}}{\pi^2}e^{\pi^2/24} = 1.29706861206\dots$.
\end{proposition}

\subsection{Branching factor}\label{ss:branching}
\noindent
Given the overall structure of FDAGs, it is no surprise that the enumeration tree grows extremely fast. However, despite this combinatorial explosion, we show in this subsection that the branching factor, i.e., the outdegree of the nodes in the enumeration tree, is controlled. Actually, we prove that any FDAG has a linear number of successors\footnote{``successor'' in the sense of ``children in the enumeration tree''. We make the distinction to avoid confusion with the children denoted  by $\child(\cdot)$.} in the enumeration tree.

\begin{theorem}\label{th:successors}
Any FDAG $D$ has $\Theta(\#D)$ successors in the FDAG enumeration tree.
\end{theorem}
\begin{proof}
Let $D=(v_0,\dots, v_n)$ be a FDAG. We denote $\childc(v_n) = a_0\cdots a_m$. Depending on the rule chosen:
\begin{description}[leftmargin=!,labelwidth=2.5em,align=right]
\item[\branching] $a_{m+1}$ belongs to $\alphabet{<} = \lbrace 0, \dots, p \rbrace$, so the maximum number of successors is at most $p+1$, and at least $1$, depending on the condition $a_{m}\lex{\geq} a_{m+1}$.
\item[\elongation] The child of the new vertex is taken from $\alphabet{=}=\lbrace p+1,\dots, n\rbrace$ so the number of successors is exactly $n-p$.
\item[\widening] Following Proposition~\ref{prop:widening}, the number of successors is exactly $\#\alphabet{<} = p+1$.
\end{description}
Combining everything, the number of successors is at least $n+2$ and at most $n+p+2\leq 2n+1$ (as $p\leq n-1$, with equality for FDAGs obtained just after using \elongation rule).
\end{proof}

\medskip
\noindent
In the previous proof, we have shown that the number of successors of a FDAG with $n$ vertices is between $n+1$ and $2n-1$. Figure~\ref{fig:successors} illustrates that these boundaries are tight, on 1~000 randomly generated FDAGs. A random FDAG is constructed as follows.

\begin{definition}[Random FDAG]\label{def:random}
Let $k\geq 0$. Starting from $D_0$ -- the root, construct iteratively $D_i$ as a successor of $D_{i-1}$ in the enumeration tree, picked uniformly at random. We stop after $k$ steps, and keep $D_k$.
\end{definition}

\noindent
\begin{minipage}[c]{0.48\textwidth}
In Figure~\ref{fig:successors}, we have generated 10 random FDAGs for each $k\in\lbrace 1,\dots,100\rbrace$.

\medskip
\noindent
It is indeed a suitable property that any FDAG admits a linear number of successors; but it would be of little use if the time required to compute those successors is too important. We demonstrate in the following proposition that temporal complexity is manageable. There are two possible strategies: (i) one can keep the enumeration tree in memory, and store on each node only the increment allowing to construct a FDAG from its predecessor; or (ii) one can explicitly build the successors by copying the starting FDAG, so that the tree can be forgotten. Depending on whether one wants to build the tree itself or only the FDAGs that compose it, one will choose either strategy.
\end{minipage}\hfill
\begin{minipage}[c]{0.48\textwidth}
\includegraphics[width=\textwidth]{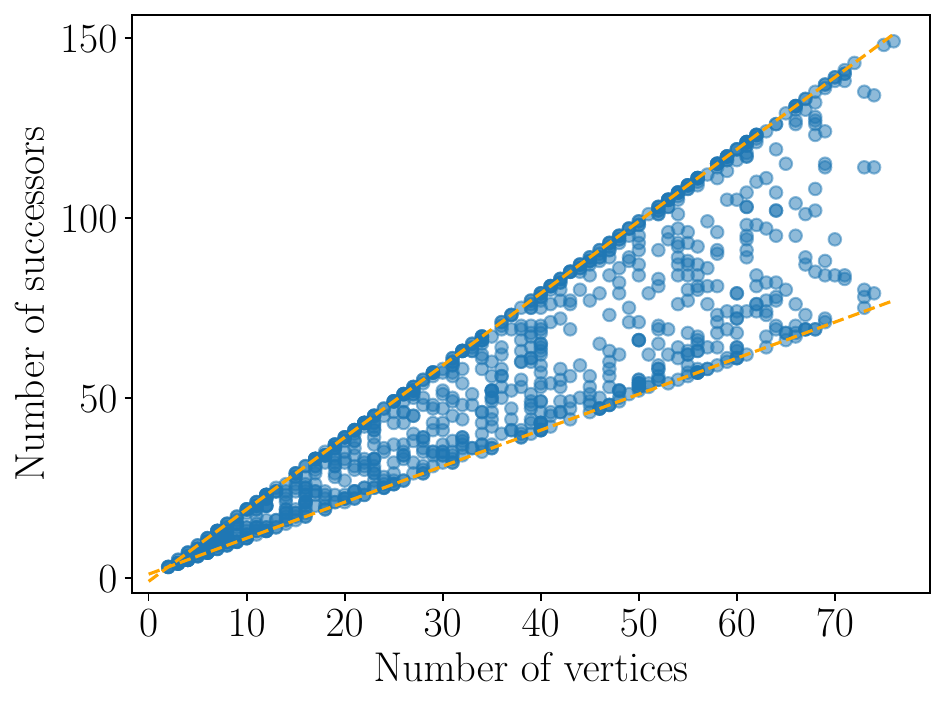}
\captionof{figure}{Numbers of successors of 1,000 random FDAGs in the enumeration tree, according to their number of vertices. Orange lines have equations $y=n+1$ and $y=2n-1$.}
\label{fig:successors}
\end{minipage}

\begin{proposition}\label{prop:time}
Computing the successors of any FDAG $D$ has complexity:
\begin{enumerate}[label=(\roman*)]
\item $O(\#D\deg(D))$ if the construction is incremental from $D$;
\item $O\left((\#D\deg(D))^2\right)$ if the construction involves copying $D$.
\end{enumerate}
\end{proposition}
\begin{proof}
Let $D=(v_0,\dots,v_n)$ be a FDAG with $n+1$ vertices, with $\childc(v_n) = a_0\cdots a_m$, $\alphabet{<}=\lbrace 0,\dots,p\rbrace$ and $\alphabet{=}=\lbrace p+1,\dots,n\rbrace$. Although the alphabets $\alphabet{<}$ and $\alphabet{=}$ can be retrieved in linear time, it is more efficient to maintain the pair $(n,p)$ during enumeration; how to update these indices has already been presented in Subsection~\ref{ss:rules}, when introducing each expansion rule.

\medskip
\noindent
The explicit construction of the successors in case $(ii)$ requires to copy the vertices of $D$ and their children, leading to a complexity in the order of $\sum_{i=0}^n (1+\deg(v_i))$, which can be roughly bounded by $(n+1)(\deg(D)+1)$.

\medskip
\noindent
Depending on the expansion rule, the complexity for computing the new vertex or new arc varies:
\begin{description}[leftmargin=!,labelwidth=2.5em,align=right]
\item[\branching] The last letter of $\childc(v_n)$ determines the number of successors -- but it is no more than $p+1$. In case $(i)$, although we could just store the information of the new letter, it is better to copy $\childc(v_n)$ and add the new letter and store the result. Indeed, this allows to always have the knowledge of $\childc(v_n)$ in the enumeration tree. The complexity for case $(i)$ is therefore bounded by $\deg(v_n)(p+1)$; whereas it is $p+1$ in case $(ii)$ since $\childc(v_n)$ is already copied.
\item[\elongation] Each successor is obtained by picking one element of $\alphabet{=}=\lbrace p+1,\dots, n\rbrace$. The complexity is exactly (up to a constant) $n-p$ in both cases.
\item[\widening] The successors are obtained by Algorithm~\ref{algo:widening}, involving copying subwords of $\childc(v_n)$ -- the overall complexity is bounded by $(p+1)\deg(v_n)$.
\end{description}
The overall complexity is therefore of the order of $2(p+1)\deg(v_n)+n-p$ in case $(i)$ and of $(n+1)(\deg(D)+1)\left[(p+1)(\deg(v_n)+1)+n-p\right]$ in case $(ii)$. Using rough bounds, with $\deg(v_n)\leq \deg(D)$ and $p\leq n$, we end up with the stated complexity.
\end{proof}

\medskip
\noindent
\begin{minipage}{0.48\textwidth}
\let\thempfootnote\thefootnote
Whereas Figure~\ref{fig:successors} shows the number of successors of 1,000 random FDAGs, we measured the time needed to compute \emph{explicitly} -- i.e., implying copy, which is case $(ii)$ in the previous Proposition~\ref{prop:time} -- these successors. The results are depicted in Figure~\ref{fig:successorstime}, where we plotted (in blue) the total time $t_D$ for computing all successors of a given FDAG $D$, and (in red) what we call \emph{amortized time}, i.e. $t_D / (\#D \deg(D))^2$. As expected from Proposition~\ref{prop:time}, one can observe an asymptotic quadratic behaviour for the total time (in blue); concerning amortized time (in red), despite some variability, the upper bound seems to be constant.

\medskip
\noindent
The computations have been made on a MacBook Pro (2014) with an Intel Core i7 2.8 GHz processor and 16 GB of RAM.

\end{minipage}\hfill
\begin{minipage}{0.48\textwidth}
\centering
\includegraphics[width=\textwidth]{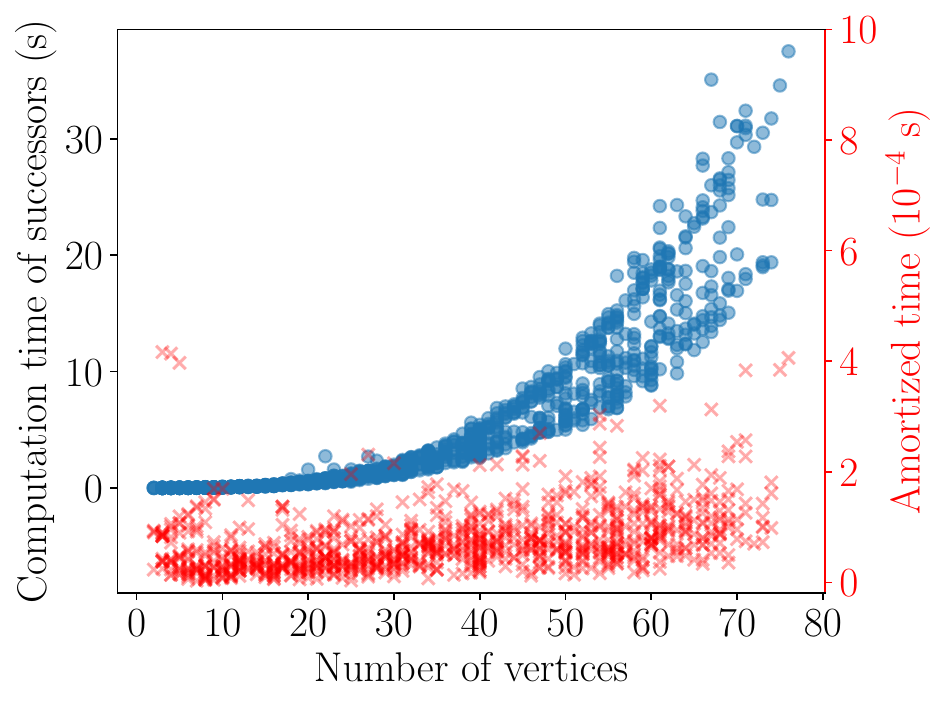}
\captionof{figure}{Total computation time (in blue) and amortized time (in red) for the explicit construction of the successors of the 1,000 random FDAGs of Figure~\ref{fig:successors}, according to their number of vertices.}
\label{fig:successorstime}
\end{minipage}

\subsection{Polynomial delay}\label{ss:poly_delay}

Let $E_{\leq K+1}=\bigcup_{k\leq K+1} E_k$ be the set of all FDAGs reachable in at most $K+1$ steps from the root of the FDAG enumeration tree. In this subsection, we show that the time complexity for enumerating $E_{\leq K+1}$ can be expressed as a function of the cardinality of $E_{\leq K}$ and has polynomial delay.

\begin{theorem}\label{th:polynomial_delay}
Enumerating $E_{\leq K+1}$ has time complexity $O(K^2 \#E_{\leq K})$.
\end{theorem}
\begin{proof}
We adopt the configuration where we keep the enumeration tree in memory and where each node contains the incremental information to construct a FDAG from its predecessor.

\medskip
\noindent
We first observe the following: as $D_0$, the root, has one vertex and no arcs, and since the rules of expansion can only add one vertex and/or increase the degree of the last vertex by one, for any $D\in E_k$, it follows naturally that $\#D\leq k+1$ and $\deg(D)\leq k$. It implies that the time complexity for generating the successors of a FDAG in $E_k$ is $O(k^2)$, according to case $(i)$ of Proposition~\ref{prop:time}.

\medskip
\noindent
Thus, enumerating all the elements of $E_{k+1}$ requires a time complexity of $O(k^2 \#E_k)$.
It follows that we have a complexity of $O\left(\sum_{k\leq K} k^2 \#E_k\right)$ for enumerating $E_{\leq K+1}$. Since $k\leq K$ and $\sum_{k\leq K} \#E_K=\# E_{\leq K}$, we end up with the announced complexity.
\end{proof}

\medskip
\noindent
As such, our algorithm has a polynomial delay, which is desirable for this kind of enumeration \cite{johnson1988generating}.

\section{Variations on the enumeration tree}\label{sec:variant}

In this section, two variants of the enumeration tree presented in Section~\ref{sec:enumeration} are introduced. Subsection~\ref{ss:redundant} proposes a way to enumerate forests in their classical sense, i.e., where redundancies within the forest are accepted, by adding an extra step following the previous enumeration. Finally, options to constrain the enumeration tree -- on maximum number of vertices, height or outdegree -- and making it finite are proposed in Subsection~\ref{ss:constraint}.

\subsection{Extension to forests with repetitions}\label{ss:redundant}

The enumeration tree constructed in Section~\ref{sec:enumeration} only allows to enumerate, in their compressed form, irredundant forests, where no tree can be a subtree of (or equal to) another. In this subsection, we propose a method to enumerate forests in the usual sense, without this non-redundancy restriction.

\medskip
\noindent
Let $F$ be a forest in the classical sense, i.e., where some trees may be identical to or subtrees of other trees. If we compute $D=\red(F)$, by definition, all these redundancies will be lost: the trees which are subtrees of another will be compressed with these subtrees in the obtained DAG, and those which are identical will be compressed in the same source of the DAG. This is why we specified in Subsection~\ref{ss:def:treedag} that DAG compression is lossless if and only if the forest is irredundant.

\medskip
\noindent
We can preserve the information lost by the compression if we keep, in addition, a \emph{presence vector}. Let us rewrite $D=(v_0,\dots,v_n)$ according to the canonical order. Each tree $T\in F$ is associated with an index $i\in\lbrace 0,\dots,n\rbrace$ such that $T=\red^{-1}(D[v_i])$. The presence vector $\presence_F : \lbrace 0,\dots,n\rbrace \to \mathbb{N}$ is constructed such that $\presence_F(i)$ counts how many times the tree $\red^{-1}(D[v_i])$ appears in $F$. Thus, the couple $(D,\presence_f)$ completely characterizes the forest $F$. To enumerate all (redundant) forests, it is therefore sufficient to enumerate both all FDAGs (corresponding to irredundant forests) and the presence vectors that may be associated with them.

\medskip
\noindent
Let $D$ be an FDAG constructed in the FDAG enumeration tree. We define $\presence_D$ as the presence vector associated to the (irredundant) forest $\red^{-1}(D)$. This vector can be computed in a linear traversal of $D$, where the sources of $D$ are assigned a value of 1 and the other vertices are assigned a value of 0. Adding redundancies in a forest means incrementing the presence vector, each $+1$ resulting in a new tree, whether it is equal to an existing tree or a subtree of it.

\medskip
\noindent
Our strategy is to enumerate, from $\presence_D$, all presence vectors corresponding to forests whose DAG reduction would be exactly $D$. To do so, we use a reverse search structure, with the following expansion rule $(E)$. Let $j$ be the index of the last increment, initialized to $j=0$.

\begin{definition}[$E$]\label{rule:repetition}
Choose any index $j'\geq j$. Increase $\presence(j')$ by one and set $j\gets j'$.
\end{definition}

\noindent
This rule allows to get any presence vector from $\presence_D$ in a unique way, i.e., each index must be increased to its desired final value before moving to the next index. This defines an enumeration tree of presence vectors. If we implement this tree in such a way that each node contains only the new index $j'$, we obtain an algorithm that enumerates each presence vector from its parent in constant time and space. The growth of this (infinite) new enumeration tree is given by the following proposition.

\begin{proposition}
The number of redundant forests that can be constructed in at most $k\geq 1$ steps from $\red^{-1}(D)$ -- following expansion rule $(E)$ -- is given by $\binom{n+1+k}{k}-1.$
\end{proposition}
\begin{proof}
We first notice that the expected number is exactly the same as the number of presence vectors constructible in at most $k$ steps from $\presence_D$. We then notice that if the current node (in the presence vector enumeration tree) has index $j$, then it has $n+1-j$ successors by the expansion rule $(E)$ of Definition~\ref{rule:repetition}. For instance, since the starting index is $0$, for $k=1$, we obtain the indices $0,\dots,n$ in one copy each. We denote by $n_p(j)$ the number of times the index $j$ appears in the nodes obtained in exactly $p$ steps from the origin. Thus, $n_1(j)=1$ by the above. Each index $j'\leq j$ existing at step $p-1$ will induce a successor with index $j$ at step $p$, so that $n_p(j) = \sum_{j'=0}^j n_{p-1}(j')$.

\medskip
\noindent
We establish by induction on $p$ that $n_p(j) = \binom{k-1+j}{j}$, using the so called hockey-stick identity $\sum_{r=0}^m \binom{n+r}{r}=\binom{n+1+m}{m}$ \cite{jones1994generalized}. Since the number of presence vectors that can be constructed in at most $k\geq 1$ steps from $\presence_D$ is given by $\sum_{p=1}^k \sum_{j=0}^n n_p(j)$, we obtain the expected result after applying twice the hockey-stick identity.
\end{proof}

\medskip
\noindent
We can merge the enumeration tree of repetitions with the enumeration tree of FDAGs, to form a single enumeration tree, which enumerates forests in the classical sense (and in compressed form), as follows: the nodes of the enumeration tree carry a couple (FDAG, presence vector), and the available expansion rules are \branching, \elongation, \widening and $(E)$. However, successors created with the last rule produce branches where it becomes the only rule available. In other words, once one chooses repetition, one can not modify any longer the topology of the FDAG -- this is to ensure that each forest can only be enumerated in a unique way.

\subsection{Constraining the enumeration}\label{ss:constraint}
In \cite{nakano2003efficient}, the authors propose an algorithm to enumerate all trees with at most $n$ vertices. They simply check whether the current tree has $n$ vertices or not, and as their expansion rule adds one vertex at a time, they decide to cut a branch in the enumeration tree once they have reached $n$ vertices. Similarly, adding a vertex to a tree can only increase its height or outdegree, so we can proceed in the same way to enumerate all trees with maximal height $H$ and maximal outdegree $d$. Indeed, the number of trees satisfying those constraints is finite \cite[Appendix D.2]{azais:hal-01294013}.

\medskip
\noindent
This property also holds with the approach presented in Section~\ref{sec:enumeration}: following one of the three expansion rules, we can only increase the height, outdegree or number of vertices of the FDAG. So, it makes sense to define similar constraints on the enumeration. However, for this constrained enumeration to generate a finite number of FDAGs, constraints must be chosen wisely, as shown in the following proposition.

\begin{proposition}\label{prop:enum:constraints}
The enumeration tree of FDAGs is finite if at least one of those set of constraints is chosen:
\begin{enumerate}[label=(\roman*)]
\item maximum number of vertices $n$ and maximum outdegree $d$,
\item maximum height $H$ and maximum outdegree $d$.
\end{enumerate}
\end{proposition}
\begin{proof}
As \branching allows to add arcs indefinitely without changing the numbers of vertices, constraining on the maximum outdegree is mandatory in both cases. As the two others rules add vertices, constraining by the number of vertices leads to a finite enumeration tree -- (i) is proved. To conclude, we only need to prove that \widening can not be repeated an infinite number of times, i.e. there is only a finite number of new vertices that can be added at a given height, up to the maximum outdegree. This is achieved by virtue of the upcoming lemma.

\medskip
\noindent
Let $H>2$ and $d\geq 1$. Let $D$ be the FDAG constructed so that for each $0\leq h \leq H$, $D$ has the maximum possible number $n_h$ of vertices of height $h$ and with maximum outdegree $d$. Initial values are $n_0=1$ and $n_1=d$.

\begin{lemma}
$\forall 2\leq h\leq H,\quad n_h = \displaystyle\sum_{k=1}^d \binom{k+n_{h-1}-1}{k}\binom{d-k+n_0+\dots + n_{h-2}}{d-k}$.
\end{lemma}
Let $h\geq 2$ be fixed. To lighten the notation, let $n=n_{h-1}$ and $m=n_0+\dots + n_{h-2}$. Let $v$ be a vertex to be added at height $h$. For any vertex $v_i$ at height $h-1$, let $x_i$ be the multiplicity of $v_i$ in $\child(v)$ -- $0$ if $v_i\notin \child(v)$. Similarly, for any vertex $v_h$ with $\height(v_j)\leq h-2$, $y_j$ is the multiplicity of $v_j$ in $\child(v)$ -- possibly $0$. By definition of $\height(\cdot)$ -- see (\ref{eq:height}), at least one $x_i$ is non-zero. Therefore, there exist $k\in [\![1,d]\!]$ such that:
$$\begin{array}{cl}
&x_1+\dots +x_n =k\\
& y_1 + \dots + y_m \leq d-k \\
\end{array}$$
\noindent
By virtue of the stars and bars theorem, for a fixed $k$, there are $\binom{k+n-1}{k}$ choices for variables $x_i$, and $\binom{d-k+m}{d-k}$ for variables $y_j$. Summing upon all values for $k$ proves the claim.
\end{proof}

\begin{remark}
In the constrained enumeration proposed in \cite{nakano2003efficient}, all the trees with $n$ vertices are the leaves of the enumeration tree. To get all trees with $n+1$ vertices, it suffices to add to the enumeration all children of these leaves, i.e. trees obtained by adding a single vertex to them. This property -- moving from one parameter value to the next by enumerating just one step further -- does not hold anymore as soon as our set of constraints involve the maximum outdegree $d$, both for trees and FDAGs. For instance, from a FDAG of height $H$, one can obtain FDAG of height $H+1$ by using \elongation once and repeating \branching up to $d-1$ times.
\end{remark}

\section{Enumeration of forests of subtrees}\label{sec:mining}
Once the reverse search scheme has been set up to enumerate a certain type of structure, it is natural to move to a finer scale by using the same scheme to enumerate substructures. However, the notion of ``substructure'' is not obvious to derive from the main structure, as several choices are possible -- e.g. for trees one can think of subtrees \cite{vishwanathan2002fast, azais2020weight}, subset trees \cite{collins2002convolution}, etc. From a practical point of view, the enumeration of substructures permits to solve the frequent pattern mining problem -- which will be tackled in Section~\ref{sec:datamining}.

\medskip
\noindent
In this section we define forests of subtrees, which will be our substructures. Compressed as FDAGs, these objects will be called subFDAGs. We then address the problem of enumerating all subFDAGs appearing in an FDAG $D$ -- similar as the one of enumerating all subtrees of a tree.

\paragraph{Forests of subtrees}
Similarly to forest being tuple of trees, \emph{forests of subtrees} are tuple of subtrees, satisfying (\ref{eq:forest}). Formally:

\begin{definition}\label{def:stuff}
Let $F$ and $f$ be two forests. $f$ is a \emph{forest of subtrees} of $F$ if and only if
$$\forall t\in f,\; \exists T\in F,\; t\in \subtrees(T).$$
\end{definition}

\noindent
Forests of subtrees can be directly constructed from FDAGs, as shown by the upcoming proposition. Let $D$ be a FDAG, and $V$ be a subset of vertices of $D$.

\begin{proposition}\label{prop:fostdag}
If $\forall v\in V$, $\child(v)\subseteq V$, then $V$ defines a FDAG $\Delta$, such that $\red^{-1}(\Delta)$ is a forest of subtrees of $\red^{-1}(D)$.
\end{proposition}
\begin{proof}
We recall from Subsection~\ref{ss:def:treedag} that the notation $D[v]$ stands for the subDAG of $D$ rooted in $v$ composed of $v$ and all its descendants $\des(v)$. The notation $\red^{-1}(D[v])$ stands for the tree compressed by $D[v]$.
The demonstration is in two steps. (i) Remove from $D$ the vertices that does not belong to $V$; as there are no arcs that leave $V$ by hypothesis, end up with a FDAG. Let us call $\Delta$ this FDAG. (ii) Let $\rho$ be a root of $\Delta$. By construction, $\rho$ is also a vertex in $D$. Among all roots of $D$, there exists a root $r$ such that $\rho\in \des(r)$. Therefore, $D[\rho]$ is a subDAG of $D[r]$, and then $t=\red^{-1}(D[\rho])$ is a subtree of $T=\red^{-1}(D[r])$ -- with $T\in F=\red^{-1}(D)$. As $\Delta[\rho]$ and $D[\rho]$ are isomorphic, $t\in f=\red^{-1}(\Delta)$. Therefore we have proved that $\forall t\in f,\; \exists T\in F, \; t\in \subtrees(T)$.
\end{proof}

\medskip
\noindent
We say that the FDAG $\Delta$ is a \emph{subFDAG}\footnote{Not to be confused with \emph{subDAG}, introduced in Subsection~\ref{ss:def:treedag}. A subDAG admits a single root and therefore compresses a single tree, whereas a subFDAG admits several roots and compresses a forest.} of $D$. Figure~\ref{fig:FoST} provides an example of such a construction.

\begin{figure}[h]
\centering
\def\xscale{0.7}
\def\yscale{0.7}
\def\nodescale{0.7}
\begin{subfigure}[t]{0.11\textwidth}
\begin{tikzpicture}[xscale=\xscale,yscale=\yscale]
\tikzstyle{noeud}=[draw,circle,fill=blue,scale=\nodescale*1]
\node[noeud] (0) at (0,0) {};
\tikzstyle{noeud}=[draw,circle,fill=magenta,scale=\nodescale*1]
\node[noeud] (1) at (-1,1) {};
\tikzstyle{noeud}=[draw,circle,fill=cyan,scale=\nodescale*1]
\node[noeud] (2) at (0,1) {};
\tikzstyle{noeud}=[draw,circle,fill=brown,scale=\nodescale*1]
\node[noeud] (3) at (1,1) {};
\tikzstyle{noeud}=[draw,circle,fill=orange,scale=\nodescale*1]
\node[noeud] (4) at (-1,2) {};
\tikzstyle{noeud}=[draw,circle,fill=yellow,scale=\nodescale*1]
\node[noeud] (5) at (0,2) {};
\tikzstyle{fleche}=[->-,>=latex,black]
\draw[fleche] (1)--(0) {};
\draw[fleche] (2)--(0) {}node [right,midway,scale=\nodescale] {$2$};
\draw[fleche] (3)--(0) {}node [right,midway,scale=\nodescale] {$3$};
\draw[fleche] (4)--(1) {};
\draw[fleche] (5)--(2) {};
\draw[fleche] (5)--(1) {}node [left,midway,scale=\nodescale] {$2$};

\tikzstyle{noeud}=[draw,circle,dashed,red,scale=1.8*\nodescale]
\node[noeud] (40) at (-1,2) {};
\node[noeud] (41) at (0,1) {};
\node[noeud] (42) at (-1,1) {};
\node[noeud] (43) at (0,0) {};

\end{tikzpicture}
\caption{\label{fig:fost:a}}
\end{subfigure}\hfill
\begin{subfigure}[t]{0.08\textwidth}
\begin{tikzpicture}[xscale=\xscale,yscale=\yscale]
\tikzstyle{noeud}=[draw,circle,fill=blue,scale=\nodescale*1]
\node[noeud] (0) at (0,0) {};
\tikzstyle{noeud}=[draw,circle,fill=magenta,scale=\nodescale*1]
\node[noeud] (1) at (-1,1) {};
\tikzstyle{noeud}=[draw,circle,fill=cyan,scale=\nodescale*1]
\node[noeud] (2) at (0,1) {};
\tikzstyle{noeud}=[draw,circle,fill=orange,scale=\nodescale*1]
\node[noeud] (4) at (-1,2) {};
\tikzstyle{fleche}=[->-,>=latex,black]
\draw[fleche] (1)--(0) {};
\draw[fleche] (2)--(0) {}node [right,midway,scale=\nodescale] {$2$};
\draw[fleche] (4)--(1) {};
\end{tikzpicture}
\caption{\label{fig:fost:b}}
\end{subfigure}\hfill
\begin{subfigure}[t]{0.15\textwidth}
\begin{tikzpicture}[xscale=\xscale,yscale=\yscale]
\tikzstyle{noeud}=[draw,circle,fill=blue,scale=\nodescale*1]
\node[noeud] (0) at (0,0) {};
\tikzstyle{noeud}=[draw,circle,fill=magenta,scale=\nodescale*1]
\node[noeud] (1) at (0,1) {};
\tikzstyle{noeud}=[draw,circle,fill=orange,scale=\nodescale*1]
\node[noeud] (4) at (0,2) {};

\node at (0,3) {$t_1$};

\tikzstyle{fleche}=[-,>=latex]
\draw[fleche] (1)--(0) {};
\draw[fleche] (4)--(1) {};

\def\xshift{2}

\tikzstyle{noeud}=[draw,circle,fill=blue,scale=\nodescale*1]
\node[noeud] (0c) at (-0.5+\xshift,0) {};
\node[noeud] (0d) at (0.5+\xshift,0) {};
\tikzstyle{noeud}=[draw,circle,fill=cyan,scale=\nodescale*1]
\node[noeud] (2) at (0+\xshift,1) {};
\tikzstyle{fleche}=[-,>=latex]
\draw[fleche] (2)--(0c) {};
\draw[fleche] (2)--(0d) {};

\node at (0+\xshift,3) {$t_2$};

\draw [decorate,decoration={brace,amplitude=5pt,raise=2ex}] (-0.5,3) -- (2.5,3) node[midway,yshift=2em]{$f=\lbrace t_1,t_2\rbrace$};

\end{tikzpicture}
\caption{\label{fig:fost:c}}
\end{subfigure}\hfill
\begin{subfigure}[t]{0.35\textwidth}
\begin{tikzpicture}[xscale=\xscale,yscale=\yscale]
\tikzstyle{noeud}=[draw,circle,fill=blue,scale=\nodescale*1]
\node[noeud] (0) at (0,0) {};
\tikzstyle{noeud}=[draw,circle,fill=magenta,scale=\nodescale*1]
\node[noeud] (1) at (0,1) {};
\tikzstyle{noeud}=[draw,circle,fill=orange,scale=\nodescale*1]
\node[noeud] (4) at (0,2) {};

\node at (0,3) {$T_1$};

\tikzstyle{fleche}=[-,>=latex]
\draw[fleche] (1)--(0) {};
\draw[fleche] (4)--(1) {};

\def\xshift{2}

\tikzstyle{noeud}=[draw,circle,fill=blue,scale=\nodescale*1]
\node[noeud] (0) at (0+\xshift,0) {};
\node[noeud] (0b) at (-1+\xshift,0) {};
\node[noeud] (0c) at (1+\xshift,0) {};
\node[noeud] (0d) at (2+\xshift,0) {};
\tikzstyle{noeud}=[draw,circle,fill=magenta,scale=\nodescale*1]
\node[noeud] (1) at (-1+\xshift,1) {};
\node[noeud] (1b) at (0+\xshift,1) {};
\tikzstyle{noeud}=[draw,circle,fill=cyan,scale=\nodescale*1]
\node[noeud] (2) at (1.5+\xshift,1) {};
\tikzstyle{noeud}=[draw,circle,fill=yellow,scale=\nodescale*1]
\node[noeud] (5) at (0+\xshift,2) {};
\tikzstyle{fleche}=[-,>=latex]
\draw[fleche] (1b)--(0) {};
\draw[fleche] (1)--(0b) {};
\draw[fleche] (2)--(0c) {};
\draw[fleche] (2)--(0d) {};
\draw[fleche] (5)--(2) {};
\draw[fleche] (5)--(1) {};
\draw[fleche] (5)--(1b) {};

\node at (0+\xshift,3) {$T_2$};

\def\xshift{6}

\tikzstyle{noeud}=[draw,circle,fill=blue,scale=\nodescale*1]
\node[noeud] (0) at (0+\xshift,0) {};
\node[noeud] (0b) at (-1+\xshift,0) {};
\node[noeud] (0c) at (1+\xshift,0) {};
\tikzstyle{noeud}=[draw,circle,fill=brown,scale=\nodescale*1]
\node[noeud] (3) at (0+\xshift,1) {};

\tikzstyle{fleche}=[-,>=latex]

\draw[fleche] (3)--(0) {};
\draw[fleche] (3)--(0b) {};
\draw[fleche] (3)--(0c) {};

\node at (0+\xshift,3) {$T_3$};

\draw [decorate,decoration={brace,amplitude=5pt,raise=2ex}] (-0.5,3) -- (7.5,3) node[midway,yshift=2em]{$F=\lbrace T_1,T_2,T_3\rbrace$};

\end{tikzpicture}
\caption{\label{fig:fost:d}}
\end{subfigure}~
\caption{Construction of a forest of subtrees from FDAG. (\subref{fig:fost:a}) A FDAG $D$. The set $V$ is circled in red. (\subref{fig:fost:b}) The FDAG $\Delta$ (\subref{fig:fost:c}) The forest $f$ compressed by $\Delta$. (\subref{fig:fost:d}) The forest $F$ compressed by $D$. One can spot that $t_1\in \subtrees(T_1)$, $t_2\in\subtrees(T_2)$ so $f$ is a forest of subtrees of $F$, and $\Delta$ a subFDAG of $D$.}
\label{fig:FoST}
\end{figure}

\paragraph{Enumeration of subFDAGs} We now solve the following enumeration problem: given a forest $F$, find all forests of subtrees of $F$. Equally, given a FDAG $D$, find all subFDAGs of $D$. To address this, we make extensive use of the reverse search technique, adapting the one presented in Section~\ref{sec:enumeration}.

\medskip
\noindent
Since a subFDAG is also a FDAG, it admits successors in the enumeration tree defined in Section~\ref{sec:enumeration}. We are interested in those of these successors who are also subFDAGs (if any). In fact, since a subFDAG can be defined from a set of vertices, all one has to do is determine which new vertex can be chosen to expand an existing subFDAG -- corresponding to a \elongation or \widening step.The covering of all added new arcs is implicit in this construction and corresponds to some steps of \branching.

\medskip
\noindent
Let $\Delta$ be a subFDAG of $D$ and $v$ its last inserted vertex --  it is also the vertex with the largest ordering number in $\Delta$. We denote by $S(\Delta)$ the set of all vertices $v'\in D$ that can be added to $\Delta$ to expand it to a new subFDAG. Let us call $S(\Delta)$ the \emph{set of candidate vertices} of $\Delta$. More precisely:

\begin{lemma}
$S(\Delta)$ is the set of vertices $v'\in D$ that satisfies both:
\begin{enumerate}[label=(\roman*)]
\item $\child(v') \subseteq \Delta$
\item $\topord(v') > \topord(v)$
\end{enumerate}
where $\topord(\cdot)$ is the canonical ordering of $D$.
\end{lemma}
\begin{proof}
(i) This condition is necessary so that $\Delta'=\Delta \cup \lbrace v'\rbrace$ fulfill the requirements for Proposition~\ref{prop:fostdag}. (ii) This condition is necessary so that $\Delta'$ remains a FDAG. As $\topord(v') > \topord(v)$, either $\height(v')=\height(v)+1$ -- then it is a \elongation step -- or $\height(v')=\height(v)$ and $\childc(v')\lex{>}\childc(v)$ -- for a \widening step.
\end{proof}

\begin{wrapfigure}[10]{R}{0.5\textwidth}
\vspace{-\baselineskip}
\begin{algorithm}[H]
\caption{\textsc{Heirs}}\label{algo:extendpattern}
\KwIn{$D, \big[\Delta, S(\Delta)\big]$}
Set $L$ to the empty list\\
\For{$s\in S(\Delta)$}{
Let $S'$ be a copy of $S(\Delta)$\\
$S'\gets S'\setminus \lbrace v' \in S' : \childc(v')\lex{\leq} \childc(s)\rbrace$\\
$S' \gets S'\cup \big\lbrace v'\in D : s\in {\child}(v') \subseteq \Delta\cup \lbrace s \rbrace  \big\rbrace$\\
Add $\big[\Delta\cup \lbrace s\rbrace , S'\big]$ to $L$}
\Return{$L$}
\end{algorithm}
\end{wrapfigure}

\medskip
\noindent
When $S(\Delta)$ is not empty, picking $s\in S(\Delta)$ ensure that $\Delta'=\Delta \cup \lbrace s\rbrace$ is a subFDAG of $D$. With respect to the enumeration tree of Section~\ref{sec:enumeration}, $\Delta$ is an ancestor of $\Delta'$ -- but not necessarily its parent, since the steps of \branching are implicit. $\Delta'$ is called an \emph{heir} of $\Delta$. We can in turn calculate $S(\Delta')$, by updating $S(\Delta)$: (i) remove from $S(\Delta)$ all vertices $v'$ such that $\topord(s)>\topord(v')$; (ii) in $D$, look only after the vertices $v'$ such that $s\in \child(v')\subseteq \Delta\cup\lbrace s\rbrace$ and add them to $S(\Delta')$.

\medskip
\noindent
If $\Delta'$ is an heir of $\Delta$, then by removing the last inserted vertex of $\Delta'$, one can retrieve $\Delta$. This define a reduction rule $f$, and therefore an enumeration tree. Algorithm~\ref{algo:extendpattern} is meant to construct the set $f^{-1}(\Delta)$. Applying Algorithm~\ref{algo:reverse:def} together with it, and starting from $\Delta= \leaves(D)$\footnote{where $\leaves(D)$ designates the leaf of $D$, i.e. the only vertex without children.} -- in this case, $S(\Delta)$ is the set of parents of $\leaves(D)$ of height $1$ -- permits to enumerate all subFDAGs of $D$. Figure~\ref{fig:stuff:enum} provide an example by enumerating all subFDAGs of the FDAG of Figure~\ref{fig:ordering:canonical}.


\begin{figure}
\centering
\def\scale{0.7}
\def\xdim{2*\scale}
\def\ydim{1.6*\scale}
\begin{tikzpicture}[xscale=\xscale,yscale=\yscale]

\tikzstyle{data}=[rectangle split,rectangle split parts=2,draw,text centered]

\node[data,scale=\scale] (0) at (0,0) {$0$ \nodepart{second} $1,2,3$};
\node[data,scale=\scale] (01) at (-2*\xdim,-\ydim) {$0,\textcolor{red}{1}$ \nodepart{second} $2,3,\textcolor{red}{4}$};
\node[data,scale=\scale] (02) at (0,-\ydim) {$0,\textcolor{red}{2}$ \nodepart{second} $3$};
\node[data,scale=\scale] (023) at (0,-2*\ydim) {$0,2,\textcolor{red}{3}$ \nodepart{second} $\emptyset$};
\node[data,scale=\scale] (03) at (\xdim,-\ydim) {$0,\textcolor{red}{3}$ \nodepart{second} $\emptyset$};

\node[data,scale=\scale] (012) at (-4*\xdim,-2*\ydim) {$0,1,\textcolor{red}{2}$ \nodepart{second} $3,4,\textcolor{red}{5}$};
\node[data,scale=\scale] (013) at (-2*\xdim,-2*\ydim) {$0,1,\textcolor{red}{3}$ \nodepart{second} $4$};
\node[data,scale=\scale] (0134) at (-2*\xdim,-3*\ydim) {$0,1,3,\textcolor{red}{4}$ \nodepart{second} $\emptyset$};
\node[data,scale=\scale] (014) at (-1*\xdim,-2*\ydim) {$0,1,\textcolor{red}{4}$ \nodepart{second} $\emptyset$};

\node[data,scale=\scale] (0123) at (-5*\xdim,-3*\ydim) {$0,1,2,\textcolor{red}{3}$ \nodepart{second} $4,5$};
\node[data,scale=\scale] (0124) at (-4*\xdim,-3*\ydim) {$0,1,2,\textcolor{red}{4}$ \nodepart{second} $5$};
\node[data,scale=\scale] (01245) at (-4*\xdim,-4*\ydim) {$0,1,2,4,\textcolor{red}{5}$ \nodepart{second} $\emptyset$};
\node[data,scale=\scale] (0125) at (-3*\xdim,-3*\ydim) {$0,1,2,\textcolor{red}{5}$ \nodepart{second} $\emptyset$};

\node[data,scale=\scale] (01234) at (-6*\xdim,-4*\ydim) {$0,1,2,3,\textcolor{red}{4}$ \nodepart{second} $5$};
\node[data,scale=\scale] (012345) at (-6*\xdim,-5*\ydim) {$0,1,2,3,4,\textcolor{red}{5}$ \nodepart{second} $\emptyset$};
\node[data,scale=\scale] (01235) at (-5*\xdim,-4*\ydim) {$0,1,2,3,\textcolor{red}{5}$ \nodepart{second} $\emptyset$};

\draw[->] (0)--(01) {};
\draw[->] (0)--(02) {};
\draw[->] (0)--(03) {};
\draw[->] (02)--(023) {};
\draw[->] (01)--(012) {};
\draw[->] (01)--(013) {};
\draw[->] (01)--(014) {};
\draw[->] (013)--(0134) {};
\draw[->] (012)--(0123) {};
\draw[->] (012)--(0124) {};
\draw[->] (012)--(0125) {};
\draw[->] (0124)--(01245) {};

\draw[->] (0123)--(01234) {};
\draw[->] (0123)--(01235) {};
\draw[->] (01234)--(012345) {};

\end{tikzpicture}

\caption{Enumeration tree of the subFDAGs of the FDAG of Figure~\ref{fig:ordering:canonical}, using both Algorithm~\ref{algo:reverse:def} and Algorithm~\ref{algo:extendpattern}. The indices of the vertices correspond to the canonical ordering defined in Figure~\ref{fig:ordering:canonical}. In each vertex, the upper part corresponds to the current subFDAG $\Delta$ whereas the lower part stands for the set $S(\Delta)$. Numbers in \textcolor{red}{red} indicate what changes for an heir compared to its parent.}
\label{fig:stuff:enum}
\end{figure}

\section{Frequent subFDAG mining problem}\label{sec:datamining}

Using the reverse search formalism defined in Subsection~\ref{ss:def:reverse}, the \emph{frequent pattern mining problem} can be formulated as follows: from a dataset $X=\lbrace s_1,\dots, s_n \rbrace$ with $s_i \in \mathbb{S}$, and a fixed threshold $\sigma$, find all elements $s\in\mathbb{S}$ that satisfy $\text{freq}(s,X) \geq \sigma$, where $\text{freq}(\cdot)$ is a function that counts the frequency of appearance of $s$ in the dataset $X$. This problem can be solved using Algorithm~\ref{algo:reverse:def} with the function $g(s,X,\sigma) = (\text{freq}(s,X) \geq \sigma )$, which is trivially anti-monotone.

\medskip
\noindent
We emphasize here that each possible definition of ``$s$ appears in $X$'' leads to a different data mining problem. The choice of this definition is therefore of prime importance. In particular, this choice should induce a way of calculating $\text{freq}(s,X)$ that reflects the specificity of the reduction rule $f$, so that $\lbrace s\in f^{-1}(s_0) | g(s)=\top\rbrace$ can be constructed directly, instead of first generating $f^{-1}(s_0)$ and then filtering according to the value of $g$. Indeed, if $g$ is too restrictive, and $f^{-1}(s_0)$ too large, one would have to enumerate objects that are not relevant to the enumeration problem, which is not desirable.

\medskip
\noindent
In this article, the problem we consider is the following: given a set of trees $X=\lbrace T_1,\dots, T_n\rbrace$, account for forests of subtrees that appear simultaneously in different $T_i$'s. In other words, if we denote $\forest_i$ the set of all forests of subtrees appearing in the forest formed by $\lbrace T_i\rbrace$, we are interested in the study of $\cap_{i\in I_\sigma} \forest_i$ where $I_\sigma\subseteq [\![1,n]\!]$, such that $\#I_\sigma\geq \sigma \cdot n$.

\medskip
\noindent
A first, naive strategy would be to first build the $\forest_i$'s, e.g. by using Algorithm~\ref{algo:extendpattern} on $\red(T_i)$, and then construct $\cap_{i\in I_\sigma} \forest_i$ for all possible choices of $I_\sigma$. Obviously, this approach has its weaknesses: (i) many subFDAGs will be enumerated for nothing or in several copies, and (ii) it does not take into account that $X$ is itself a forest. Our aim is to propose a variant of Algorithm~\ref{algo:extendpattern} that, applied to $\red(X)$, would enumerate only subFDAGs appearing in the $\red(T_i)$'s with a large enough frequency.

\medskip
\noindent
Given a forest $F=\lbrace T_1,\dots, T_n\rbrace$ and its DAG compression $D=\red(F)$, we have to retrieve, for each vertex in $D$, their origin in the dataset, that is, from which tree they come from. This issue has already been addressed in a previous article \cite[Section 3.3]{azais2020weight}, and has led to the concept of \emph{origin}. For any vertex $v\in D$, the origin of $v$ is defined as
$$\origin(v) = \big\lbrace i \in [\![1,n]\!] : {\red}^{-1}(D[v])\in \subtrees(T_i)  \big\rbrace, $$
where $\red^{-1}(D[v])$ designates the tree compressed by the subDAG $D[v]$ rooted in $v$. In other words, $\origin(v)$ represents the set of trees for which $\red^{-1}(D[v])$ is a subtree. We state in \cite[Proposition 3.4]{azais2020weight} that origins can be iteratively computed in one exploration of $D$. The proof lies in the property that if $i\in\origin(v)$, then for all $v'\in\des(v)$, $i\in\origin(v')$ -- as $\red^{-1}(D[v'])\in\subtrees(\red^{-1}(D[v]))$.

\medskip
\noindent
Let $\Delta$ be a subFDAG of $D$. For $\Delta$ to compress a forest of subtrees of a tree $T_i$, it is necessary that $i\in\origin(v)$ for all $v\in \Delta$. Therefore, the set of trees for which $\Delta$ compress a forest of subtrees -- the \emph{origin} of $\Delta$, denoted by $\stufforigin(\Delta)$ -- is equal to
$$\stufforigin(\Delta) = \bigcap_{v\in \Delta} \origin(v).$$

\noindent
If $\Delta'= \Delta \cup \lbrace s \rbrace$ is an heir of $\Delta$ -- as defined earlier, then $\stufforigin(\Delta') = \stufforigin(\Delta) \cap \origin(s)$. Algorithm~\ref{algo:extendpattern} can therefore be refined so that $\Delta'$ should be ignored if $\stufforigin(\Delta')=\emptyset$ -- as $\Delta'$ does not anymore compresses any forest of subtrees actually present in the trees of $F$.

\medskip
\noindent
So far we neglected the threshold $\sigma$. We only want to keep subFDAGs that appear in at least $\sigma\%$ of the data. If $\# \stufforigin(\Delta) / \#F < \sigma$, then the successors of $\Delta$ are not investigated. Indeed, as $\stufforigin(\cdot)$ is a decreasing function, successors of $\Delta$ can not exceed the threshold again.

\noindent
\begin{minipage}{0.35\textwidth}
We can finally introduce Algorithm~\ref{algo:extendpatterndata} that solves the frequent subFDAG mining problem for trees. With the notations of Subsection~\ref{ss:def:reverse}, this algorithm builds the set $\lbrace \Delta' \in f^{-1}(\Delta) | \text{freq}(\Delta', F) \geq \sigma \rbrace$, with $\text{freq}(\Delta',F)= \#\stufforigin(\Delta')/\#F$. The set is also built directly, without any posterior filtering, which is suitable as discussed at the beginning of the present subsection.
\end{minipage}
\hfill
\begin{minipage}{0.6\textwidth}
\begin{algorithm}[H]
\caption{\textsc{FrequentHeirs}}\label{algo:extendpatterndata}
\KwIn{$D=\red(F), \big[\Delta, S(\Delta), \stufforigin(\Delta) \big]$, $\sigma$}
Set $L$ to the empty list\\
\For{$s\in S(\Delta)$}{
\If{$\stufforigin(\Delta)\cap \origin(s) \neq \emptyset$ \textup{\textbf{and}} $\#(\stufforigin(\Delta)\cap\origin(s)) \geq \sigma \cdot \#F$}{
Let $S'$ be a copy of $S(\Delta)$\\
$S'\gets S'\setminus \lbrace v' \in S' : \childc(v')\lex{\leq} \childc(s)\rbrace$\\
$S' \gets S'\cup \big\lbrace v'\in D : s\in {\child}(v') \subseteq D_0\cup \lbrace s \rbrace  \big\rbrace$\\
Add $\big[\Delta\cup \lbrace s\rbrace , S', \stufforigin(\Delta)\cap\origin(s)\big]$ to $L$}}
\Return{$L$}
\end{algorithm}
\end{minipage}

\medskip
\noindent
We stated earlier that we wanted to avoid generating unnecessary or multiple copies of subFDAGs, which is achieved with Algorithm~\ref{algo:extendpatterndata}. We now empirically study what we have gained from this, by comparing the use of Algorithm~\ref{algo:extendpatterndata} on $D=\red(F)$, with the use of Algorithm~\ref{algo:extendpattern} on each $\red(T_i)$. As in Subsection~\ref{ss:branching}, we generated 1 000 random FDAGs $D_k$, 10 repetitions for each $k\in \lbrace 1,\dots, 100\rbrace$, creating $D_k$ as in Definition~\ref{def:random}. We assume $D_k=\red(f)$ where $f= \lbrace \red^{-1}(D_k[r]) : r \text{ source of } D_k\rbrace$. For each $D_k$, we have computed the quotient

$$Q(D_k) = \frac{\text{number of subFDAGs of } D_k \text{  enumerated via Algorithm~\ref{algo:extendpatterndata}}}{\displaystyle\sum_{r \text{ source of } D_k} \text{number of subFDAGs of } D_k[r] \text{ enumerated via Algorithm~\ref{algo:extendpattern}}}$$

\noindent
with parameter $\sigma=0$ when using Algorithm~\ref{algo:extendpatterndata}. The results are provided in Figure~\ref{fig:mined_stuff}. Despite a rather marked variability, there is a general trend of decreasing as the number of vertices increases. We obtain fairly low quotients, around 20\%, quite quickly. Given the combinatorial explosion of the objects to be enumerated, such an advantage is of the greatest interest.

\medskip
\noindent
\begin{minipage}{0.45\textwidth}
\subsection*{Implementation}
The \verb+treex+ library for Python \cite{azais2019treex} is designed to manipulate rooted trees, with a lot of diversity (ordered or not, labeled or not). It offers options for random generation, visualization, edit operations, conversions to other formats, and various algorithms. The enumeration of Section~\ref{sec:enumeration} and the algorithms of Section~\ref{sec:mining} and~\ref{sec:datamining} have been implemented as a module of \verb+treex+ so that the interested reader can manipulate the concepts discussed in this paper in a ready-to-use manner. Installing instructions and the documentation of \verb+treex+ can be found from  \cite{azais2019treex}.
\end{minipage}\hfill
\begin{minipage}{0.5\textwidth}
\centering
\includegraphics[width=\textwidth]{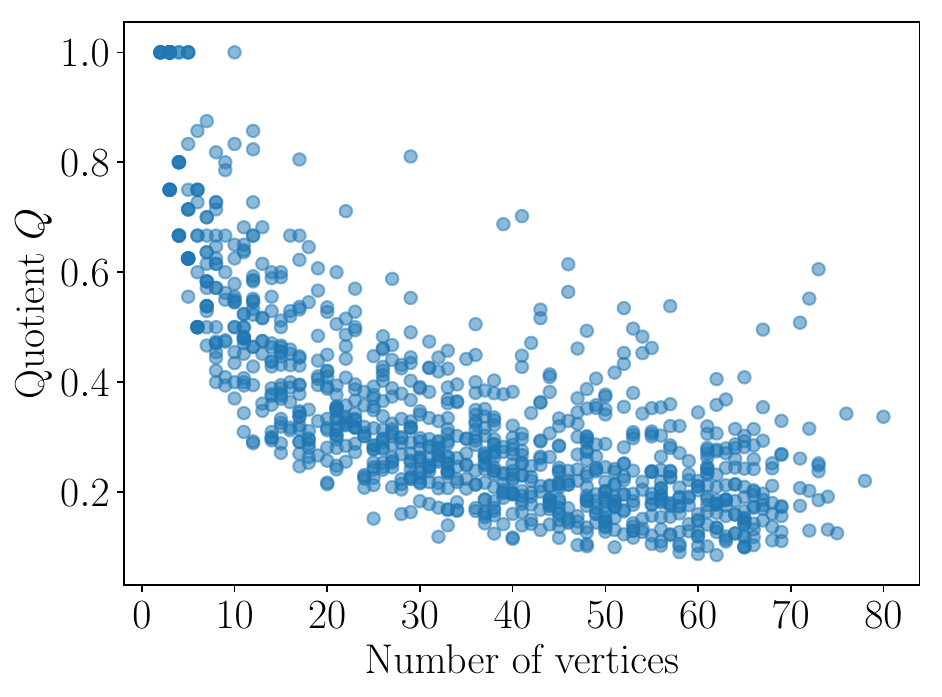}
\captionof{figure}{Quotient $Q(D)$ according to the number of vertices of $D$. Here 1 000 random FDAGs are displayed.}
\label{fig:mined_stuff}
\end{minipage}

\subsection*{Acknowledgments}
This work has been supported by the European Union's H2020 project ROMI. The authors would like to thank Dr. Arnaud Mary for his helpful suggestions on the draft of the article. Finally, the authors would like to thank two anonymous reviewers for their valuable comments, which help them to significantly improve the quality of their manuscript.

{\small

\bibliography{main}}

\begin{thebibliography}{10}

\bibitem{asai2003discovering}
Tatsuya Asai, Hiroki Arimura, Takeaki Uno, and Shin-Ichi Nakano.
\newblock Discovering frequent substructures in large unordered trees.
\newblock In {\em International Conference on Discovery Science}, pages 47--61.
  Springer, 2003.

\bibitem{avis1996reverse}
David Avis and Komei Fukuda.
\newblock Reverse search for enumeration.
\newblock {\em Discrete Applied Mathematics}, 65(1-3):21--46, 1996.

\bibitem{azais2019treex}
Romain Aza{\"\i}s, Guillaume Cerutti, Didier Gemmerl{\'e}, and Florian Ingels.
\newblock treex: a python package for manipulating rooted trees.
\newblock {\em The Journal of Open Source Software}, 4, 2019.

\bibitem{azais:hal-01294013}
Romain Aza{\"i}s, Jean-Baptiste Durand, and Christophe Godin.
\newblock {Approximation of trees by self-nested trees}.
\newblock In {\em {ALENEX 2019 - Algorithm Engineering and Experiments}}, pages
  1--24, San Diego, United States, January 2019.
\newblock URL: \url{https://hal.archives-ouvertes.fr/hal-01294013}, \href
  {http://dx.doi.org/10.10860} {\path{doi:10.10860}}.

\bibitem{azais2020weight}
Romain Aza{\"\i}s and Florian Ingels.
\newblock The weight function in the subtree kernel is decisive.
\newblock {\em Journal of Machine Learning Research}, 21:1--36, 2020.

\bibitem{bender2010lists}
Edward~A Bender and S~Gill Williamson.
\newblock {\em Lists, Decisions and Graphs}.
\newblock S. Gill Williamson, 2010.

\bibitem{bringmann2014asymptotics}
Kathrin Bringmann, Yingkun Li, and Robert~C Rhoades.
\newblock Asymptotics for the number of row-fishburn matrices.
\newblock {\em European Journal of Combinatorics}, 41:183--196, 2014.

\bibitem{collins2002convolution}
Michael Collins and Nigel Duffy.
\newblock Convolution kernels for natural language.
\newblock In {\em Advances in neural information processing systems}, pages
  625--632, 2002.

\bibitem{flajolet2009analytic}
Philippe Flajolet and Robert Sedgewick.
\newblock {\em Analytic combinatorics}.
\newblock cambridge University press, 2009.

\bibitem{GODI}
Christophe Godin and Pascal Ferraro.
\newblock Quantifying the degree of self-nestedness of trees: application to
  the structural analysis of plants.
\newblock {\em IEEE/ACM Transactions on Computational Biology and
  Bioinformatics (TCBB)}, 7(4):688--703, 2010.

\bibitem{han2007frequent}
Jiawei Han, Hong Cheng, Dong Xin, and Xifeng Yan.
\newblock Frequent pattern mining: current status and future directions.
\newblock {\em Data mining and knowledge discovery}, 15(1):55--86, 2007.

\bibitem{Hart:1991:EAR:127719.122728}
John~C. Hart and Thomas~A. DeFanti.
\newblock Efficient antialiased rendering of 3-d linear fractals.
\newblock {\em SIGGRAPH Comput. Graph.}, 25(4):91--100, July 1991.
\newblock URL: \url{http://doi.acm.org/10.1145/127719.122728}, \href
  {http://dx.doi.org/10.1145/127719.122728} {\path{doi:10.1145/127719.122728}}.

\bibitem{hwang2019asymptotics}
Hsien-Kuei Hwang and Emma~Yu Jin.
\newblock Asymptotics and statistics on fishburn matrices and their
  generalizations.
\newblock {\em arXiv preprint arXiv:1911.06690}, 2019.

\bibitem{jelinek2012counting}
V{\'\i}t Jel{\'\i}nek.
\newblock Counting general and self-dual interval orders.
\newblock {\em Journal of Combinatorial Theory, Series A}, 119(3):599--614,
  2012.

\bibitem{johnson1988generating}
David~S Johnson, Mihalis Yannakakis, and Christos~H Papadimitriou.
\newblock On generating all maximal independent sets.
\newblock {\em Information Processing Letters}, 27(3):119--123, 1988.

\bibitem{jones1994generalized}
Charles~H Jones.
\newblock Generalized hockey stick identities and n-dimensional blockwalking.
\newblock 1994.

\bibitem{kahn1962topological}
Arthur~B Kahn.
\newblock Topological sorting of large networks.
\newblock {\em Communications of the ACM}, 5(11):558--562, 1962.

\bibitem{land2010automatic}
Ailsa~H Land and Alison~G Doig.
\newblock An automatic method for solving discrete programming problems.
\newblock In {\em 50 Years of Integer Programming 1958-2008}, pages 105--132.
  Springer, 2010.

\bibitem{nakano2002efficient}
Shin-Ichi Nakano.
\newblock Efficient generation of plane trees.
\newblock {\em Information Processing Letters}, 84(3):167--172, 2002.

\bibitem{nakano2003efficient}
Shin-ichi Nakano and Takeaki Uno.
\newblock Efficient generation of rooted trees.
\newblock {\em National Institute for Informatics (Japan), Tech. Rep.
  NII-2003-005E}, 8, 2003.

\bibitem{nowozin2009learning}
Sebastian Nowozin.
\newblock {\em Learning with structured data: applications to computer vision.}
\newblock PhD thesis, Berlin Institute of Technology, 2009.

\bibitem{schwikowski2002enumerating}
Benno Schwikowski and Ewald Speckenmeyer.
\newblock On enumerating all minimal solutions of feedback problems.
\newblock {\em Discrete Applied Mathematics}, 117(1-3):253--265, 2002.

\bibitem{Sutherland:1963:SMG:1461551.1461591}
Ivan~E. Sutherland.
\newblock Sketchpad: A man-machine graphical communication system.
\newblock In {\em Proceedings of the May 21-23, 1963, Spring Joint Computer
  Conference}, AFIPS '63 (Spring), pages 329--346, New York, NY, USA, 1963.
  ACM.
\newblock URL: \url{http://doi.acm.org/10.1145/1461551.1461591}, \href
  {http://dx.doi.org/10.1145/1461551.1461591}
  {\path{doi:10.1145/1461551.1461591}}.

\bibitem{vishwanathan2002fast}
S.V.N. Vishwanathan and Alexander~J Smola.
\newblock Fast kernels on strings and trees.
\newblock {\em Advances on Neural Information Proccessing Systems}, 14, 2002.

\bibitem{yamazaki2020enumeration}
Kazuaki Yamazaki, Toshiki Saitoh, Masashi Kiyomi, and Ryuhei Uehara.
\newblock Enumeration of nonisomorphic interval graphs and nonisomorphic
  permutation graphs.
\newblock {\em Theoretical Computer Science}, 806:310--322, 2020.

\bibitem{yan2002gspan}
Xifeng Yan and Jiawei Han.
\newblock gspan: Graph-based substructure pattern mining.
\newblock In {\em 2002 IEEE International Conference on Data Mining, 2002.
  Proceedings.}, pages 721--724. IEEE, 2002.

\end{thebibliography}

\appendix

\section{A bijection between FDAGs and row-Fishburn matrices}\label{sec:dagmatrix}

This section is dedicated to the proof of Theorem~\ref{th:bijdagmatrix}, which is in two steps. First, we recall the natural bijection between FDAGs and their adjacency matrices; the latter are then put into bijection with the row-Fishburn matrices.

\paragraph{FDAG $\leftrightarrow$ Reduced adjacency matrix}
Let $D=(v_0,\dots,v_n)$ be a FDAG constructed in $k$ steps from $D_0$ in the enumeration tree defined in Subsection~\ref{ss:enumtree}. The adjacency matrix of $D$ is defined as $A=(A_{i,j})_{i,j \in [\![n,0]\!]^2}$ where, if $m$ is the multiplicity of $v_j$ in $\child(v_i)$, then $A_{i,j}=m$ -- possibly $0$ if $v_j\notin\child(v_i)$. By construction of $D$, as $v_n$ is the last inserted vertex, it has no parents, so $A_{n,\cdot}$ is a column of zeros; and as $v_0$ is a leaf, it has no children, so $A_{\cdot,0}$ is a row of zeros. We define the \emph{reduced} adjacency matrix $M$ as the matrix $A$ deprived of this column and this row. Therefore, $M=(A_{i,j})_{i\in [\![n,1]\!], j\in [\![n-1,0]\!]}$. As a vertex can not be a parent to any vertex introduced after it, we have $A_{i,j}=0$ for all $i\leq j$ -- so that $M$ is an upper-triangular matrix. In addition, as all vertices except $v_0$ have at least one child, there is at least one non-zero entry in each row of $M$. Therefore, $M$ is a row-Fishburn matrix. However, we have no guarantee that this matrix verifies $\text{size}(M)=k$.

\paragraph{Reduced adjacency matrix $\rightarrow$ Incremental adjacency matrix} Let $D=(v_0,\dots,v_n)$ be a FDAG, and $M$ its reduced adjacency matrix. Let $M_i$ be the row of $M$ corresponding to $\childc(v_i)$. The incremental adjacency matrix $\hat{M}$ is defined as:
\begin{align*}[left=\empheqlbrace]
\hat{M}_1&= M_1\\
\hat{M}_{i+1}&=M_{i+1}  \ominus M_i
\end{align*}
where the $\ominus$ operation is defined as follow: given two rows $a_0\cdots a_n$ and $b_0\cdots b_n$, then denoting $j=\min \lbrace i : a_i\neq b_i\rbrace$, and $c=a_j-b_j$,

\begin{center}
\begin{tabular}{c@{\,}c@{\,}c@{\,}c@{\,}c@{\,}c}
  &   $a_0$&$\cdots $&$a_{j-1}$&\colorbox{red!50}{$a_j$}&\colorbox{blue!50}{$a_{j+1}\cdots a_n$}\\
 $ \ominus$&$b_0$&$\cdots$ &$b_{j-1}$&\colorbox{red!50}{$b_j$}&$b_{j+1}\cdots b_n$ \\
\hline
=  &$0$&$\cdots$ &$0$&\colorbox{red!50}{$c$}&\colorbox{blue!50}{$a_{j+1}\cdots a_n$}\\
\end{tabular}\quad.
\end{center}

\noindent
We claim that this new matrix $\hat{M}$ is a row-Fishburn matrix of size $k$, if $D\in E_k$. Actually, since $M$ was already a row-Fishburn matrix, we just have to check that the size is correct. Let us consider $v_i$ and $v_{i+1}$. The vertex $v_{i+1}$ has been constructed from $v_i$ by using either \elongation or \widening, and potentially several \branching after that -- let us say $p\geq 0$ times. Therefore, if the claim is correct, the sum over $\hat{M}_{i+1}$ should be exactly $p+1$. Consider the operation by which $v_{i+1}$ was added in the first place:

\begin{description}[leftmargin=!,labelwidth=2.5em,align=right]
\item[\elongation] $\childc(v_{i+1})$ is reduced to a single element $a$, such that $a\lex{>}\childc(v_i)$. Therefore, the index $j$ of the first non-zero coefficient of $M_{i+1}$ is ahead of the one of $M_i$ so that the coefficient \colorbox{red!50}{$c$} of $\ominus$ is equal to the $j$-th coefficient of $M_{i+1}$ minus zero. Since the \branching rule adds children to respect decreasing words, the $p$ extra coefficients are added to the right of the $j$-th coefficient (including it) and therefore they are kept unchanged in the $\ominus$ operation. Eventually, the sum over $M_{i+1}$ is $p+1$ and so is the sum over $\hat{M}_{i+1}$.
\item[\widening] $\childc(v_{i+1})$ is built from $\childc(v_i)$ with Algorithm~\ref{algo:widening}, and therefore they (i) share a common prefix, possibly empty and (ii) then differ by a single letter. The index of that letter in $M_{i+1}$ corresponds to the index $j$ defined in $\ominus$. Therefore, the coefficient \colorbox{red!50}{$c$} is -- before any \branching -- equal to one. The argument of \branching letters being added to the right of $j$ still hold and therefore the sum over $\hat{M}_{i+1}$ is also $p+1$.
\end{description}

\noindent
To conclude the proof, we have to exhibit the inverse function of the mapping we just defined. This will prove that this mapping is indeed a bijection, and then the theorem holds.

\paragraph{Incremental adjacency matrix $\rightarrow$ Reduced adjacency matrix}
Let $M$ and $\hat{M}$ be constructed as before. From $\hat{M}$, we can define a matrix $M'$ as:
\begin{align*}[left=\empheqlbrace]
M'_1&= \hat{M}_1\\
M'_{i+1}&=M'_{i}  \oplus \hat{M}_{i+1}
\end{align*}
where the $\oplus$ operation is defined as follow: given two words $a_0\cdots a_n$ and $b_0\cdots b_n$, then denoting $j=\min\lbrace i : b_i\neq 0\rbrace$, and $c=a_j+b_j$,
\begin{center}
\begin{tabular}{c@{\,}c@{\,}c@{\,}c}
  &   \colorbox{blue!50}{$a_0\cdots a_{j-1}$}&\colorbox{red!50}{$a_j$}&$a_{j+1}\cdots a_n$\\
 $ \oplus $&$b_0\cdots b_{j-1}$&\colorbox{red!50}{$b_j$}&\colorbox{green!50}{$b_{j+1}\cdots b_n$} \\
\hline
  =&\colorbox{blue!50}{$a_0\cdots a_{j-1}$}&\colorbox{red!50}{$c$}&\colorbox{green!50}{$b_{j+1}\cdots b_n$}\\
\end{tabular}\quad.
\end{center}
\noindent
By construction, $\oplus$ is the inverse operation of $\ominus$, so that we have the following lemma:
\begin{lemma}
The following properties hold:
\begin{itemize}
\item $M_i \oplus ( M_{i+1} \ominus M_i)=M_{i+1}$
\item $(M_i\oplus \hat{M}_{i+1})\ominus M_i=\hat{M}_{i+1}$
\end{itemize}
\end{lemma}

\noindent
Thefore, $M=M'$.

\medskip
\noindent
The FDAG of Figure~\ref{fig:ordering:canonical} is reproduced below to illustrates the stages of the proof. This FDAG is constructed in 7 steps, that are (in this order): \elongation, \widening, \widening, \elongation, \widening, \branching and \branching. The matrices $A$, $M$ and $\hat{M}$ are given in Figure~\ref{fig:adjmat}. One can see that $\hat{M}$ is of size 7, as expected.

\begin{figure}[h]
\centering
\begin{minipage}[c]{0.33\textwidth}
\centering
\def\xscale{0.7}
\def\yscale{0.7}
\def\nodescale{0.7}
\begin{tikzpicture}[xscale=\xscale,yscale=\yscale]
\tikzstyle{noeud}=[draw,circle,fill=blue,scale=\nodescale*1]
\node[noeud] (0) at (0,0) {};
\tikzstyle{noeud}=[draw,circle,fill=magenta,scale=\nodescale*1]
\node[noeud] (1) at (-1,1) {};
\tikzstyle{noeud}=[draw,circle,fill=cyan,scale=\nodescale*1]
\node[noeud] (2) at (0,1) {};
\tikzstyle{noeud}=[draw,circle,fill=brown,scale=\nodescale*1]
\node[noeud] (3) at (1,1) {};
\tikzstyle{noeud}=[draw,circle,fill=orange,scale=\nodescale*1]
\node[noeud] (4) at (-1,2) {};
\tikzstyle{noeud}=[draw,circle,fill=yellow,scale=\nodescale*1]
\node[noeud] (5) at (0,2) {};
\tikzstyle{fleche}=[->-,>=latex,black]
\draw[fleche] (1)--(0) {};
\draw[fleche] (2)--(0) {}node [right,midway,scale=\nodescale] {$2$};
\draw[fleche] (3)--(0) {}node [right,midway,scale=\nodescale] {$3$};
\draw[fleche] (4)--(1) {};
\draw[fleche] (5)--(2) {};
\draw[fleche] (5)--(1) {}node [left,midway,scale=\nodescale] {$2$};
\end{tikzpicture}

$\begin{array}{c|c:ccc:cc}
v & \node[blue] & \node[magenta] &\node[cyan]&\node[brown]&\node[orange]&\node[yellow]\\
\topord(v)&0&1&2&3&4&5 \\
\childc(v) & &0&00& 000 & 1 &211
\end{array}$
\end{minipage}~
\begin{minipage}[c]{0.33\textwidth}
\centering
$$\begin{blockarray}{ccccccc}
A&v_5 & v_4 & v_3 & v_2 & v_1 &v_0\\
\begin{block}{c(cccccc)}
  v_5 & . & 0 & 0 & 1 & 2&0 \\
  v_4 &. & .& 0 & 0 & 1& 0\\
  v_3 & . & . & . & 0 & 0& 3\\
  v_2 & . & . & . & . & 0 &2\\
  v_1 & . & . & . & . & . &1\\
    v_0 & . & . & . & . & . &.\\
\end{block}
\end{blockarray}$$
\end{minipage}

\begin{minipage}[c]{0.33\textwidth}
\centering
 $$  \begin{blockarray}{cccccc}
M & v_4 & v_3 & v_2 & v_1 &v_0\\
\begin{block}{c(ccccc)}
  v_5 &  0 & 0 & 1 & 2&0 \\
  v_4 & .& 0 & 0 & 1& 0\\
  v_3 &  . & . & 0 & 0& 3\\
  v_2 &  . & . & . & 0 &2\\
  v_1 &  . & . & . & . &1\\
\end{block}
\end{blockarray}$$
\end{minipage}~
\begin{minipage}[c]{0.33\textwidth}
\centering
$$\begin{blockarray}{cccccc}
\hat{M} & v_4 & v_3 & v_2 & v_1 &v_0\\
\begin{block}{c(ccccc)}
  v_5 &  0 & 0 & 1 & 2&0 \\
  v_4 & .& 0 & 0 & 1& 0\\
  v_3 &  . & . & 0 & 0& 1\\
  v_2 &  . & . & . & 0 &1\\
  v_1 &  . & . & . & . &1\\
\end{block}
\end{blockarray}$$
\end{minipage}
\vspace{-\baselineskip}
\caption{The FDAG of Figure~\ref{fig:ordering:canonical} reproduced (top left), its adjacency matrix $A$ (top right), its reduced adjacency matrix $M$ (bottom left) and its incremental adjacency matrix $\hat{M}$ (bottom right). Dots represent zeros corresponding to $A_{i,j}$ elements with $i\leq j$.}
\label{fig:adjmat}
\end{figure}

\begin{remark}
It should be noted that (general) Fishburn matrices, with at least one non-zero entry on each row and column, are in bijection with FDAGs compressing forests made of a unique tree. Indeed, via the bijection above, as such FDAG have a unique root, it must be the last inserted vertex, and therefore, each column admits at least one non-zero entry (otherwise it would be another root).
\end{remark}

\begin{remark}
It is possible to enumerate row-Fishburn matrices by using the previous bijection and the FDAGs enumeration tree together. Nevertheless, things are a little simpler in this case and the equivalent of the operations \elongation and \widening can be merged, giving two rules for matrix expansion:
\begin{enumerate}
\item[(R1)] Increase one coefficient to the (inclusive) left of the rightmost nonzero coefficient of the top row by 1.
\item[(R2)] Increase the dimension of the matrix by $1$ (to the left and top), all new coefficients set to zero. Set one coefficient of the top row to 1.
\end{enumerate}
\end{remark}

\section{Index of frequent notations}\label{annex:notations}
{\small
\paragraph{Trees \& DAGs} $v$ designates indifferently a vertex of a tree $T$ or a DAG $D$.
\begin{description}[leftmargin=!,labelwidth=7em,align=right]
\item[$\child(v)$] \emph{children} of $v$: all vertices connected to an arc leaving $v$
\item[$\deg(v)$]  \emph{outdegree} of $v$: number of children of $v$
\item[$\des(v)$]  \emph{descendants} of $v$: children of $v$, their children, and so on
\item[$\height(v)$] \emph{height} of $v$: length of the longest path from $v$ to a leaf
\item[$\#T, \#D$] number of vertices
\item[{$T[v], D[v]$}] \emph{subtree/subDAG} rooted in $v$ and composed of $v$ and $\des(v)$
\item[$\leaves(T), \leaves(D)$] \emph{leaves}: vertices without any children
\item[$\deg(T), \deg(D)$]  \emph{outdegree}: maximum outdegree among all vertices
\item[$\subtrees(T)$]  the set of all distinct subtrees of $T$
\item[$\tree$] the set of all trees
\item[$\forest$] the set of all \emph{forests}, i.e. sets of trees such that no tree is a subtree of another
\end{description}

\paragraph{DAG reduction} $D$ designates a FDAG, $F$ a forest.
\begin{description}[leftmargin=!,labelwidth=7em,align=right]
\item[$\red(F)$] DAG reduction of the forest $F$
\item[$\red^{-1}(D)$]  the forest $F$ compressed by $D$, so that $\red(F)=D$
\item[{$\red^{-1}(D[v])$}] the tree $T$ compressed by $D[v]$, so that $\red(\lbrace T\rbrace) = D[v]$
\end{description}

\paragraph{Canonical FDAGs} Let $D$ be a fixed FDAG, $v$ any vertex of $D$ and $v_n$ the vertex with highest index in the canonical ordering.
\begin{description}[leftmargin=!,labelwidth=7em,align=right]
\item[$\psi(\cdot)$] canonical topological ordering of $D$
\item[$\childc(v)$] $\child(v)$ sorted by decreasing order on the indices defined by $\psi(\cdot)$
\item[$\alphabet{=}$] the indices of all vertices of $D$ with same height as $v_n$
\item[$\alphabet{<}$] the indices of all vertices of $D$ with strictly inferior height as $v_n$
\item[{$\lang[<]{}$}] the set of all \emph{decreasing words} on $\alphabet{<}$, i.e. where each letter is greater than or equal to those who follow, w.r.t. the lexicographical order
\item[{ $\lang[<]{\overline{w}}$}] the set of all decreasing words \emph{bounded by $\overline{w}$}, i.e. all words in $\lang[<]{}$ that are greater than or equal to $\overline{w}$, w.r.t. the lexicographical order
\item[$\scut(w)$] \emph{suffix-cut operator}: the word $w$ deprived of its last letter
\end{description}

\paragraph{Enumeration tree}
\begin{description}[leftmargin=!,labelwidth=7em,align=right]
\item[$D_0$] the FDAG with one vertex and no arcs
\item[ $E_k$] the set of FDAGs that are accessible in exactly $k$ steps from $D_0$ in the FDAGs enumeration tree
\item[$\presence_D$] \emph{presence vector}: $\presence_D(i)$ counts how many times the tree $\red^{-1}(D[v_i])$ appears in the forest $\red^{-1}(D)$, for any FDAG $D=(v_0,\dots,v_n)$ and $i\in \lbrace 0,\dots, n\rbrace$.
\end{description}

\paragraph{Forests of subtrees} Let $\Delta$ be a subFDAG of $D$, and $v$ any vertex. Let $F=\red^{-1}(D)$.
\begin{description}[leftmargin=!,labelwidth=7em,align=right]
\item[$S(\Delta)$] \emph{candidate vertices of $\Delta$}:  the set of all vertices $v'$ of $D$ so that $\Delta \cup \lbrace v'\rbrace$ is still a subFDAG of $D$
\item[$\origin(v)$] \emph{origin of $v$}: the set of indices $i$ so that $\red^{-1}(D[v])$ is a subtree of $T_i\in F$.
\item[$\stufforigin(\Delta)$] \emph{origin of $\Delta$}: the set of indices $i$ so that $\red^{-1}(\Delta)$ is a forest of subtrees of $T_i \in F$.
\end{description}}

\end{document}